\documentclass[11pt]{article}
\usepackage{mathtools,amsmath,amssymb,amsfonts,amsthm,siunitx,bm,bbm,siunitx}
\usepackage{moreverb,rotating,graphics,esint}
\usepackage[T1]{fontenc}
\usepackage{epsfig}
\usepackage[french,english]{babel} 
\usepackage{csquotes}
\usepackage{dsfont}
\usepackage{esint}

\usepackage{graphicx}
\usepackage{subcaption}
\usepackage{tikz}
\usetikzlibrary{decorations.pathreplacing}

\numberwithin{equation}{section}

\newtheorem{theorem}{Theorem}[section]
\newtheorem{proposition}[theorem]{Proposition}
\newtheorem{remark}[theorem]{Remark}
\newtheorem{lemma}[theorem]{Lemma}
\newtheorem{corollary}[theorem]{Corollary}
\newtheorem{definition}[theorem]{Definition}

\newtheorem{assumption}[theorem]{Assumption}

\allowdisplaybreaks

\def\tr{{\rm Tr \,}}

\def\N{{\mathbb N}}

\def\R{{\mathbb R}}
\def\C{{\mathbb C}}

\def\1{{\mathds{1}}}

\def\cE{{\mathcal E}}
\def\cH{{\mathcal H}}
\def\cN{{\mathcal N}}

\def\cT{{\mathcal T}}
\def\cX{{\mathcal X}}

\def\cB{{\mathcal B}}
\def\cI{{\mathcal I}}

\def\cL{{\mathcal L}}
\def\cU{{\mathcal U}}

\def\bk{{\mathbf k}}

\def\twi{{\mathrm{twi}}}
\def\px{\mathrm{ad}_{\partial_x}}
\def\weps{\eta}

\newcommand*{\Trb}{\underline{\textrm{\rm Tr}}}
\newcommand*{\Tr}{\mathrm{Tr}}
\newcommand*{\dist}{\mathrm{dist}}
\newcommand*{\rank}{\mathrm{Rank}}

\newcommand{\Op}{{\mathrm{Op}_\epsilon}}
\newcommand{\Opc}{{\mathrm{Op}_\epsilon^{\mathrm{c}}}}
\newcommand{\Supp}{{\mathrm{Supp}\,}}

\newcommand \dps{\displaystyle }

\usepackage[
    backend=biber,
    style=numeric,
    sorting=ynt,
    datamodel=eprint-hal,maxbibnames=99
]{biblatex}
\usepackage{hyperref}
\title{Semiclassical analysis of two-scale electronic Hamiltonians for twisted bilayer graphene}
\author{
Eric Canc\`es \footnote{ \textsc{Eric Canc\`es, CERMICS, \'Ecole des Ponts ParisTech and Inria Paris, 6 and 8 av. Pascal, 77455 Marne-la-Vall\'ee, France}
  \textit{E-mail address}: \texttt{\href{mailto:eric.cances@enpc.fr}{eric.cances@enpc.fr}}}
\and Long Meng\footnote{ \textsc{Long Meng, CERMICS, \'Ecole des ponts ParisTech, 6 and 8 av. Pascal, 77455 Marne-la-Vall\'ee, France}
  \textit{E-mail address}:\texttt{\href{mailto:long.meng@enpc.fr}{long.meng@enpc.fr}}}
}
\date{}

\begin{document}

\selectlanguage{english}

\maketitle

\begin{abstract}
    This paper investigates the mathematical properties of independent-electron models for twisted bilayer graphene by examining the density-of-states of corresponding single-particle Hamiltonians using tools from semiclassical analysis. This study focuses on a specific atomic-scale Hamiltonian  $H_{d,\theta}$ constructed from Density-Functional Theory, and a family of moir\'e-scale Hamiltonians $H_{d,K,\theta}^{\rm eff}$ containing the Bistritzer-MacDonald model. The parameter $d$ represents the interlayer distance, and $\theta$ the twist angle. It is shown that the density-of-states of $H_{d,\theta}$ and $H_{d,K,\theta}^{\rm eff}$ admit asymptotic expansions in the twist angle parameter $\epsilon:=\sin(\theta/2)$. The proof relies on a twisted version of the Weyl calculus and a trace formula for an exotic class of pseudodifferential operators suitable for the study of twisted 2D materials. We also show that the density-of-states of $H_{d,\theta}$ admits an asymptotic expansion in $\eta:=\tan(\theta/2)$ and comment on the differences between the expansions in $\epsilon$ and $\eta$.
\end{abstract}

\setcounter{tocdepth}{2}
\tableofcontents

\section{Introduction}

Graphene is a material consisting of a layer of carbon atoms sitting on a honeycomb lattice (Fig. \ref{fig:1a}). It has been studied theoretically since the late 40's~\cite{bandgraphite}, but was only experimentally realized in 2004~\cite{graphenfilm}. Graphene is the most famous example of single-layer 2D materials, namely periodic crystals with 2D atomic-scale lattice periodicity. Other important single-layer 2D materials include hexagonal boron nitride (hBN), transition metal dichalcogenides (TMD), phosphorene, etc. Over the past decade, a significant body of mathematical literature has been dedicated to the study of single-layer 2D materials, see e.g. \cite{fefferman2012honeycomb,fefferman2014wave,HoneyStrong,EdgeHoney,ContinnuumSchr,EdgeHall,RigorDirac,MagnGraphene,CantorGraphene} and references therein.

\begin{figure}[h!]
  \centering
  \begin{subfigure}[b]{0.4\linewidth}
    \begin{tikzpicture}[x=0.6cm,y=0.6cm]

\filldraw[-, fill=red!5, very thick] (-1.5,-0.866)node[below] {\footnotesize$B$} -- (-1.5,0.866)node[above] {\footnotesize$A$}  -- (0,1.732) node[above,xshift=0.15cm] {\footnotesize$B$} --(1.5,0.866) node[above] {\footnotesize$A$} -- (1.5,-0.866) node[right, yshift=0.1cm] {\footnotesize$B$} -- (0,-1.732) node[above] {\footnotesize$A$} -- (-1.5,-0.866);

\draw[-] (1.5,-0.866) -- (1.5,0.866) -- (3,1.732)  --(4.5,0.866)  -- (4.5,-0.866) -- (3,-1.732) -- (1.5,-0.866);

\draw[-] (0,1.732) -- (0,3.464) -- (1.5,4.33)  --(3,3.464)  -- (3,1.732) --  (1.5,0.866) -- (0,1.732);

\draw[-] (0,-1.732) -- (0,-3.464) -- (1.5,-4.33)  --(3,-3.464)  -- (3,-1.732) --  (1.5,-0.866) -- (0,-1.732);

\filldraw[->,ultra thick,color=red] (0,0) node[left,xshift=-0.2cm] {\footnotesize$\Omega$}--(1.5,2.598) node [right] {$a_2$}; 
\filldraw[->, ultra thick,color=blue] (0,0) --(1.5,-2.598) node [right] {$a_1$}; 


\end{tikzpicture}
    \caption{Lattice vectors, atomic sites, and A/B sublattices.}
    \label{fig:1a}
  \end{subfigure}
  \begin{subfigure}[b]{0.5\linewidth}
   \begin{tikzpicture}[x=1.6cm,y=1.6cm]

\draw[-]  (-1.732, 0) node[right] {$\mathbf{K'}$}-- (-0.866, 1.5) -- (0.866, 1.5)--(1.732, 0) node[right] {$\mathbf{K}$} -- (0.866, -1.5)-- (-0.866, -1.5)--(-1.732, 0);

\filldraw[->,ultra thick,color=red] (0,0) --(2.598,1.5) node [right] {$a_2^*$}; 
\filldraw[->, ultra thick,color=blue] (0,0) --(2.598,-1.5) node [right] {$a_1^*$}; 

\node at (0, -1.0){$\mathrm{B.Z.}$};


\end{tikzpicture}
    \caption{Reciprocal lattice vectors, first Brillouin zone, and Dirac points $K$ and $K'$.}
    \label{fig:1b}
  \end{subfigure}
  \caption{Monolayer graphene}
  \label{fig:1}
\end{figure}
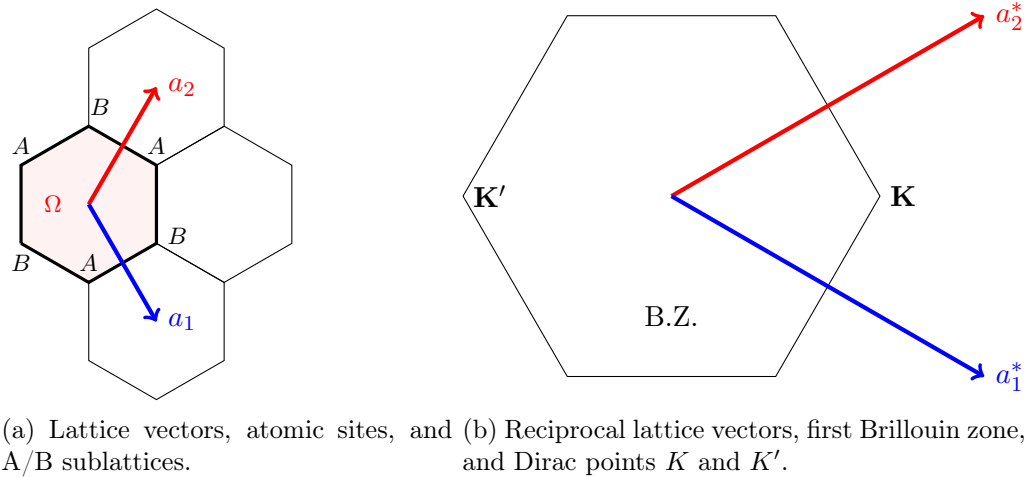

\medskip

Single-layer 2D materials can be stacked on top of each other to form 2D homostructures (e.g. bilayer graphene, trilayer graphene) or 2D heterostructures (e.g. graphene on hBN).  

Untwisted bilayer graphene \cite{mccann2012electronic} is obtained by stacking two parallel graphene layers with the same orientation and is characterized by two parameters: the interlayer distance~$d$, and the stacking vector $\mathrm b$ (Figs. \ref{fig:2a}-\ref{fig:2b}). The equilibrium geometry of untwisted bilayer graphene, for which the energy is minimal, is obtained for $d\simeq  \SI{3.4}{\angstrom}$ and $\mathrm{b}= \frac 13(a_2-a_1)$ (BA Bernal stacking, Fig. \ref{fig:2d}) or $\mathrm{b}= \frac 13(a_1-a_2)$ (AB Bernal stacking).  The configuration with $\mathrm b =0$ (AA stacking, Fig. \ref{fig:2c}) is unstable for any interlayer distance $d$.

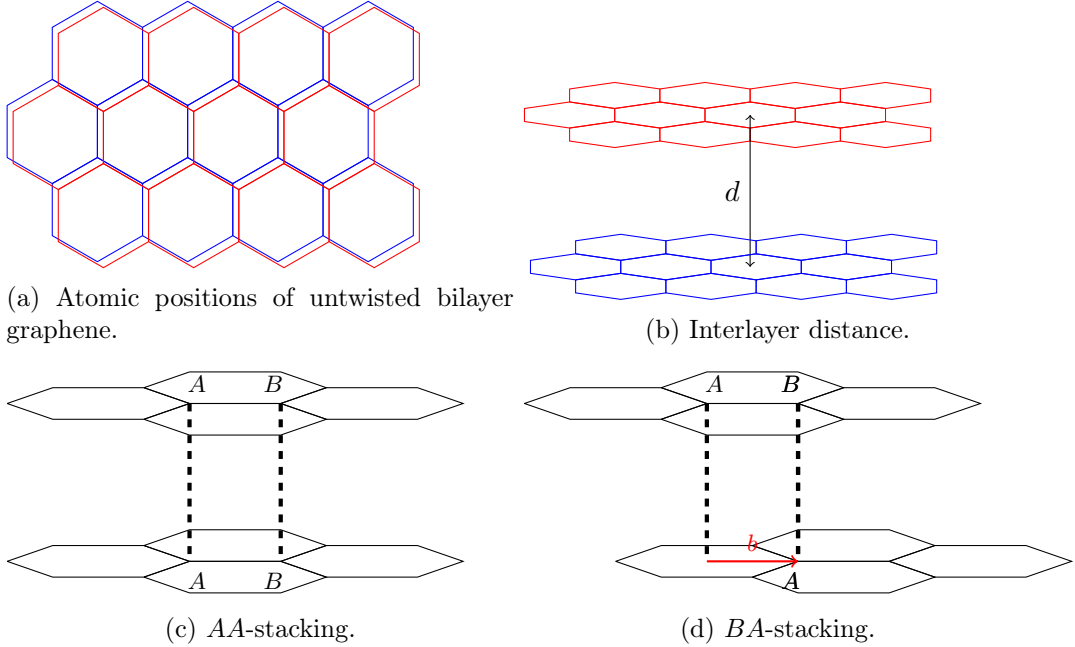
\begin{figure}[ht]
  \centering
  \begin{subfigure}[b]{0.45\linewidth}
  \begin{tikzpicture}[x=0.4cm,y=0.4cm]

\draw[-,blue] (-1.5,-0.866) -- (-1.5,0.866)  -- (0,1.732) --(1.5,0.866) -- (1.5,-0.866) -- (0,-1.732)  -- (-1.5,-0.866);

\draw[-,blue] (1.5,-0.866) -- (1.5,0.866) -- (3,1.732)  --(4.5,0.866)  -- (4.5,-0.866) -- (3,-1.732) -- (1.5,-0.866);

\draw[-,blue] (4.5,-0.866) -- (4.5,0.866) -- (6,1.732)  --(7.5,0.866)  -- (7.5,-0.866) -- (6,-1.732) -- (4.5,-0.866);

\draw[-,blue] (7.5,-0.866) -- (7.5,0.866) -- (9,1.732)  --(10.5,0.866)  -- (10.5,-0.866) -- (9,-1.732) -- (7.5,-0.866);


\draw[-,blue] (0,1.732) -- (0,3.464) -- (1.5,4.33)  --(3,3.464)  -- (3,1.732) --  (1.5,0.866) -- (0,1.732);

\draw[-,blue] (3,1.732) -- (3,3.464) -- (4.5,4.33)  --(6,3.464)  -- (6,1.732) --  (4.5,0.866) -- (3,1.732);

\draw[-,blue] (6,1.732) -- (6,3.464) -- (7.5,4.33)  --(9,3.464)  -- (9,1.732) --  (7.5,0.866) -- (6,1.732);

\draw[-,blue] (9,1.732) -- (9,3.464) -- (10.5,4.33)  --(12,3.464)  -- (12,1.732) --  (10.5,0.866) -- (9,1.732);

\draw[-,blue] (0,-1.732) -- (0,-3.464) -- (1.5,-4.33)  --(3,-3.464)  -- (3,-1.732) --  (1.5,-0.866) -- (0,-1.732);

\draw[-,blue] (3,-1.732) -- (3,-3.464) -- (4.5,-4.33)  --(6,-3.464)  -- (6,-1.732) --  (4.5,-0.866) -- (3,-1.732);

\draw[-,blue] (6,-1.732) -- (6,-3.464) -- (7.5,-4.33)  --(9,-3.464)  -- (9,-1.732) --  (7.5,-0.866) -- (6,-1.732);

\draw[-,blue] (9,-1.732) -- (9,-3.464) -- (10.5,-4.33)  --(12,-3.464)  -- (12,-1.732) --  (10.5,-0.866) -- (9,-1.732);


\draw[-,red] (-1.3,-1.066) -- (-1.3,0.666)  -- (0.2,1.532) --(1.7,0.666) -- (1.7,-1.066) -- (0.2,-1.932)  -- (-1.3,-1.066);

\draw[-,red] (1.7,-1.066) -- (1.7,0.666) -- (3.2,1.532)  --(4.7,0.666)  -- (4.7,-1.066) -- (3.2,-1.932) -- (1.7,-1.066);

\draw[-,red] (4.7,-1.066) -- (4.7,0.666) -- (6.2,1.532)  --(7.7,0.666)  -- (7.7,-1.066) -- (6.2,-1.932) -- (4.7,-1.066);

\draw[-,red] (7.7,-1.066) -- (7.7,0.666) -- (9.2,1.532)  --(10.7,0.666)  -- (10.7,-1.066) -- (9.2,-1.932) -- (7.7,-1.066);


\draw[-,red] (0.2,1.532) -- (0.2,3.264) -- (1.7,4.13)  --(3.2,3.264)  -- (3.2,1.532) --  (1.7,0.666) -- (0.2,1.532);

\draw[-,red] (3.2,1.532) -- (3.2,3.264) -- (4.7,4.13)  --(6.2,3.264)  -- (6.2,1.532) --  (4.7,0.666) -- (3.2,1.532);

\draw[-,red] (6.2,1.532) -- (6.2,3.264) -- (7.7,4.13)  --(9.2,3.264)  -- (9.2,1.532) --  (7.7,0.666) -- (6.2,1.532);

\draw[-,red] (9.2,1.532) -- (9.2,3.264) -- (10.7,4.13)  --(12.2,3.264)  -- (12.2,1.532) --  (10.7,0.666) -- (9.2,1.532);

\draw[-,red] (0.2,-1.932) -- (0.2,-3.664) -- (1.7,-4.53)  --(3.2,-3.664)  -- (3.2,-1.932) --  (1.7,-1.066) -- (0.2,-1.932);

\draw[-,red] (3.2,-1.932) -- (3.2,-3.664) -- (4.7,-4.53)  --(6.2,-3.664)  -- (6.2,-1.932) --  (4.7,-1.066) -- (3.2,-1.932);

\draw[-,red] (6.2,-1.932) -- (6.2,-3.664) -- (7.7,-4.53)  --(9.2,-3.664)  -- (9.2,-1.932) --  (7.7,-1.066) -- (6.2,-1.932);

\draw[-,red] (9.2,-1.932) -- (9.2,-3.664) -- (10.7,-4.53)  --(12.2,-3.664)  -- (12.2,-1.932) --  (10.7,-1.066) -- (9.2,-1.932);
\end{tikzpicture}
    \caption{Atomic positions of untwisted bilayer graphene.}
    \label{fig:2a}
  \end{subfigure}
  \begin{subfigure}[b]{0.45\linewidth}
  \begin{tikzpicture}[x=0.4cm,y=0.1cm]

\draw[-,red] (-1.5,-0.866) -- (-1.5,0.866)  -- (0,1.732) --(1.5,0.866) -- (1.5,-0.866) -- (0,-1.732)  -- (-1.5,-0.866);

\draw[-,red] (1.5,-0.866) -- (1.5,0.866) -- (3,1.732)  --(4.5,0.866)  -- (4.5,-0.866) -- (3,-1.732) -- (1.5,-0.866);

\draw[-,red] (4.5,-0.866) -- (4.5,0.866) -- (6,1.732)  --(7.5,0.866)  -- (7.5,-0.866) -- (6,-1.732) -- (4.5,-0.866);

\draw[-,red] (7.5,-0.866) -- (7.5,0.866) -- (9,1.732)  --(10.5,0.866)  -- (10.5,-0.866) -- (9,-1.732) -- (7.5,-0.866);


\draw[-,red] (0,1.732) -- (0,3.464) -- (1.5,4.33)  --(3,3.464)  -- (3,1.732) --  (1.5,0.866) -- (0,1.732);

\draw[-,red] (3,1.732) -- (3,3.464) -- (4.5,4.33)  --(6,3.464)  -- (6,1.732) --  (4.5,0.866) -- (3,1.732);

\draw[-,red] (6,1.732) -- (6,3.464) -- (7.5,4.33)  --(9,3.464)  -- (9,1.732) --  (7.5,0.866) -- (6,1.732);

\draw[-,red] (9,1.732) -- (9,3.464) -- (10.5,4.33)  --(12,3.464)  -- (12,1.732) --  (10.5,0.866) -- (9,1.732);

\draw[-,red] (0,-1.732) -- (0,-3.464) -- (1.5,-4.33)  --(3,-3.464)  -- (3,-1.732) --  (1.5,-0.866) -- (0,-1.732);

\draw[-,red] (3,-1.732) -- (3,-3.464) -- (4.5,-4.33)  --(6,-3.464)  -- (6,-1.732) --  (4.5,-0.866) -- (3,-1.732);

\draw[-,red] (6,-1.732) -- (6,-3.464) -- (7.5,-4.33)  --(9,-3.464)  -- (9,-1.732) --  (7.5,-0.866) -- (6,-1.732);

\draw[-,red] (9,-1.732) -- (9,-3.464) -- (10.5,-4.33)  --(12,-3.464)  -- (12,-1.732) --  (10.5,-0.866) -- (9,-1.732);


\filldraw[<->] (6,0)-- (6,-10.1) node [left] {$d$} --(6,-20.2) ; 


\draw[-,blue] (-1.3,-21.066) -- (-1.3,-19.334)  -- (0.2,-18.468) --(1.7,-19.334) -- (1.7,-21.066) -- (0.2,-21.932)  -- (-1.3,-21.066);

\draw[-,blue] (1.7,-21.066) -- (1.7,-19.334) -- (3.2,-18.468)  --(4.7,-19.334)  -- (4.7,-21.066) -- (3.2,-21.932) -- (1.7,-21.066);

\draw[-,blue] (4.7,-21.066) -- (4.7,-19.334) -- (6.2,-18.468)  --(7.7,-19.334)  -- (7.7,-21.066) -- (6.2,-21.932) -- (4.7,-21.066);

\draw[-,blue] (7.7,-21.066) -- (7.7,-19.334) -- (9.2,-18.468)  --(10.7,-19.334)  -- (10.7,-21.066) -- (9.2,-21.932) -- (7.7,-21.066);


\draw[-,blue] (0.2,-18.468) -- (0.2,-16.736) -- (1.7,-15.87)  --(3.2,-16.736)  -- (3.2,-18.468) --  (1.7,-19.334) -- (0.2,-18.468);

\draw[-,blue] (3.2,-18.468) -- (3.2,-16.736) -- (4.7,-15.87)  --(6.2,-16.736)  -- (6.2,-18.468) --  (4.7,-19.334) -- (3.2,-18.468);

\draw[-,blue] (6.2,-18.468) -- (6.2,-16.736) -- (7.7,-15.87)  --(9.2,-16.736)  -- (9.2,-18.468) --  (7.7,-19.334) -- (6.2,-18.468);

\draw[-,blue] (9.2,-18.468) -- (9.2,-16.736) -- (10.7,-15.87)  --(12.2,-16.736)  -- (12.2,-18.468) --  (10.7,-19.334) -- (9.2,-18.468);

\draw[-,blue] (0.2,-21.932) -- (0.2,-23.664) -- (1.7,-24.53)  --(3.2,-23.664)  -- (3.2,-21.932) --  (1.7,-21.066) -- (0.2,-21.932);

\draw[-,blue] (3.2,-21.932) -- (3.2,-23.664) -- (4.7,-24.53)  --(6.2,-23.664)  -- (6.2,-21.932) --  (4.7,-21.066) -- (3.2,-21.932);

\draw[-,blue] (6.2,-21.932) -- (6.2,-23.664) -- (7.7,-24.53)  --(9.2,-23.664)  -- (9.2,-21.932) --  (7.7,-21.066) -- (6.2,-21.932);

\draw[-,blue] (9.2,-21.932) -- (9.2,-23.664) -- (10.7,-24.53)  --(12.2,-23.664)  -- (12.2,-21.932) --  (10.7,-21.066) -- (9.2,-21.932);
\end{tikzpicture}
    \caption{Interlayer distance.}
    \label{fig:2b}
  \end{subfigure}
\begin{subfigure}[b]{0.45\linewidth}
\medskip
\begin{tikzpicture}[x=0.7cm,y=0.14cm]

\draw[-]  (-1.732, 0) -- (-0.866, 1.5) -- (0.866, 1.5)--(1.732, 0)  -- (0.866, -1.5)-- (-0.866, -1.5)--(-1.732, 0);

\draw[-]  (-1.732+5.196, 0) -- (-0.866+5.196, 1.5) -- (0.866+5.196, 1.5)--(1.732+5.196, 0)  -- (0.866+5.196, -1.5)-- (-0.866+5.196, -1.5)--(-1.732+5.196, 0);

\draw[-]  (-1.732+2.598, 0+1.5) -- (-0.866+2.598, 1.5+1.5) -- (0.866+2.598, 1.5+1.5)--(1.732+2.598, 0+1.5)  -- (0.866+2.598, -1.5+1.5)-- (-0.866+2.598, -1.5+1.5)--(-1.732+2.598, 0+1.5);

\draw[-]  (-1.732+2.598, 0-1.5) -- (-0.866+2.598, 1.5-1.5) -- (0.866+2.598, 1.5-1.5)--(1.732+2.598, 0-1.5)  -- (0.866+2.598, -1.5-1.5)-- (-0.866+2.598, -1.5-1.5)--(-1.732+2.598, 0-1.5);


\draw[-]  (-1.732, 0-15) -- (-0.866, 1.5-15) -- (0.866, 1.5-15)--(1.732, 0-15)  -- (0.866, -1.5-15)-- (-0.866, -1.5-15)--(-1.732, 0-15);

\draw[-]  (-1.732+5.196, 0-15) -- (-0.866+5.196, 1.5-15) -- (0.866+5.196, 1.5-15)--(1.732+5.196, 0-15)  -- (0.866+5.196, -1.5-15)-- (-0.866+5.196, -1.5-15)--(-1.732+5.196, 0-15);

\draw[-]  (-1.732+2.598, 0+1.5-15) -- (-0.866+2.598, 1.5+1.5-15) -- (0.866+2.598, 1.5+1.5-15)--(1.732+2.598, 0+1.5-15)  -- (0.866+2.598, -1.5+1.5-15)-- (-0.866+2.598, -1.5+1.5-15)--(-1.732+2.598, 0+1.5-15);

\draw[-]  (-1.732+2.598, 0-1.5-15) -- (-0.866+2.598, 1.5-1.5-15) -- (0.866+2.598, 1.5-1.5-15)--(1.732+2.598, 0-1.5-15)  -- (0.866+2.598, -1.5-1.5-15)-- (-0.866+2.598, -1.5-1.5-15)--(-1.732+2.598, 0-1.5-15);

\draw[-,ultra thick,dashed]  (-1.732+5.196, 0) node[above,xshift=-0.1cm] {\footnotesize$B$} -- (-1.732+5.196, 0-15) node[below,xshift=-0.1cm] {\footnotesize$B$};

\draw[-,ultra thick,dashed]  (1.732, 0) node[above,xshift=0.1cm] {\footnotesize$A$} -- (1.732, 0-15) node[below,xshift=0.1cm] {\footnotesize$A$};

\end{tikzpicture}
    \caption{$AA$-stacking.}
    \label{fig:2c}
\end{subfigure}
  \begin{subfigure}[b]{0.45\linewidth}
  \begin{tikzpicture}[x=0.7cm,y=0.14cm]

\draw[-]  (-1.732, 0) -- (-0.866, 1.5) -- (0.866, 1.5)--(1.732, 0)  -- (0.866, -1.5)-- (-0.866, -1.5)--(-1.732, 0);

\draw[-]  (-1.732+5.196, 0) -- (-0.866+5.196, 1.5) -- (0.866+5.196, 1.5)--(1.732+5.196, 0)  -- (0.866+5.196, -1.5)-- (-0.866+5.196, -1.5)--(-1.732+5.196, 0);

\draw[-]  (-1.732+2.598, 0+1.5) -- (-0.866+2.598, 1.5+1.5) -- (0.866+2.598, 1.5+1.5)--(1.732+2.598, 0+1.5)  -- (0.866+2.598, -1.5+1.5)-- (-0.866+2.598, -1.5+1.5)--(-1.732+2.598, 0+1.5);

\draw[-]  (-1.732+2.598, 0-1.5) -- (-0.866+2.598, 1.5-1.5) -- (0.866+2.598, 1.5-1.5)--(1.732+2.598, 0-1.5)  -- (0.866+2.598, -1.5-1.5)-- (-0.866+2.598, -1.5-1.5)--(-1.732+2.598, 0-1.5);


\draw[-]  (-1.732+1.732, 0-15) -- (-0.866+1.732, 1.5-15) -- (0.866+1.732, 1.5-15)--(1.732+1.732, 0-15)  -- (0.866+1.732, -1.5-15)-- (-0.866+1.732, -1.5-15)--(-1.732+1.732, 0-15);

\draw[-]  (-1.732+5.196+1.732, 0-15) -- (-0.866+5.196+1.732, 1.5-15) -- (0.866+5.196+1.732, 1.5-15)--(1.732+5.196+1.732, 0-15)  -- (0.866+5.196+1.732, -1.5-15)-- (-0.866+5.196+1.732, -1.5-15)--(-1.732+5.196+1.732, 0-15);

\draw[-]  (-1.732+2.598+1.732, 0+1.5-15) -- (-0.866+2.598+1.732, 1.5+1.5-15) -- (0.866+2.598+1.732, 1.5+1.5-15)--(1.732+2.598+1.732, 0+1.5-15)  -- (0.866+2.598+1.732, -1.5+1.5-15)-- (-0.866+2.598+1.732, -1.5+1.5-15)--(-1.732+2.598+1.732, 0+1.5-15);

\draw[-]  (-1.732+2.598+1.732, 0-1.5-15) -- (-0.866+2.598+1.732, 1.5-1.5-15) -- (0.866+2.598+1.732, 1.5-1.5-15)--(1.732+2.598+1.732, 0-1.5-15)  -- (0.866+2.598+1.732, -1.5-1.5-15)-- (-0.866+2.598+1.732, -1.5-1.5-15)--(-1.732+2.598+1.732, 0-1.5-15);

\draw[-,ultra thick,dashed]  (-1.732+5.196, 0)node[above,xshift=-0.1cm] {\footnotesize$B$} -- (-1.732+5.196, 0-15)node[below,xshift=-0.1cm] {\footnotesize$A$};

\draw[-,ultra thick,dashed]  (-1.732+5.196, 0)node[above,xshift=-0.1cm] {\footnotesize$B$} -- (-1.732+5.196, 0-15)node[below,xshift=-0.1cm] {\footnotesize$A$};


\draw[-,ultra thick,dashed]  (1.732, 0)node[above,xshift=0.1cm] {\footnotesize$A$} -- (1.732, 0-15);

\filldraw[red,->,thick] (1.732, 0-15)-- (1.732+0.866, 0-15) node[above] {\footnotesize$b$}--(-1.732+5.196, 0-15);

\end{tikzpicture}
    \caption{$BA$-stacking.}
    \label{fig:2d}
\end{subfigure}
  \caption{Untwisted bilayer graphene.}
  \label{fig:2}
\end{figure}

\medskip

Twisted bilayer graphene (TBG) \cite{carr2017twistronics} is obtained by stacking two identical graphene sheets on top of one another and rotating them in opposite directions around the transverse direction by a relative, typically small, angle $\theta$. Seen from above, this gives rise to a moiré pattern whose diameter scales as $\theta^{-1}$ (Fig. \ref{fig:3}). TBG is a controllable quantum system, in the sense that its properties can be finely tuned by playing with the twist angle (with, currently, an experimental accuracy of $\sim 0.1^\circ$), the transverse external electric field generated by gating (which allows one to control the doping, that is the density of charge carriers), the in-plane external electric field, the external magnetic field, the temperature, the interlayer distance (which can be changed by applying a pressure field), etc. The theoretical and experimental study of TBG has been one of the hottest topics in condensed matter physics since notably, the experimental discovery of supposedly unconventional superconducting regions in the phase diagram of TBG at ``magic'' twist angle $\theta \simeq 1.08^\circ$~\cite{magicangle}. We refer to~\cite{bernevig2021twistedI,bernevig2021twistedII,bernevig2021twistedIII,bernevig2021twistedIV,bernevig2021twistedV,bernevig2021twistedVI,carr2020electronic} for recent reviews on TBG.

\begin{figure}[ht]
     \begin{center}
    \includegraphics{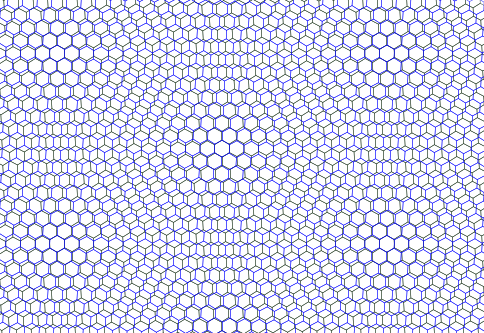}
    \caption{Moir\'e lattice vectors for twist angle $\theta\approx 2.2^{\circ}$.}
    \label{fig:3}
     \end{center}
\end{figure}
\medskip

While thousands of physics papers on TBG have been published in the past 10 years, the mathematical literature about TBG is sparse. 

A first series of works focuses on the chiral limit of the Bistritzer-MacDonald (BM) model~\cite{bistritzer2011moire2}. The BM model is an empirical non-interacting model describing low-energy excitations in TBG in the $K$ (or $K'$) valley, on which most theoretical works on TBG are based (see~\cite{faulstich2023interacting,bernevig2021twistedIII} for interacting versions of this model). The BM Hamiltonian $H_{K,\theta}^{\rm BM}$ is an unbounded self-adjoint operator on $L^2(\R^2;\C^4)$, which has the periodicity of the moiré cell. Its mathematical expression will be recalled below in Section~\ref{sec:MS-TBG}. This Hamiltonian depends on $\epsilon$ and on two empirical parameters $w_{\rm AA}$ and $w_{\rm AB}$ characterizing interlayer interactions at AA and AB (or BA) stackings~\cite{bistritzer2011moire2}. In this framework, magic angles are characterized by the existence of almost flat Bloch bands at the Fermi level $\mu=0$. The chiral limit of the BM Hamiltonian is obtained by setting $w_{\rm AA}=0$. In this limit, the BM Hamiltonian has chiral symmetry and can be studied using advanced mathematical analysis~{\cite{Mathmagicangles,becker2022fine,becker2023degenerate} sometimes combined with rigorously founded computer-aided proofs~{\cite{integrabilitychiral,firstmagicangle,TightBM}. The most recent result states that the chiral BM Hamiltonian has perfectly flat Bloch bands for an infinite discrete sequence of twist angles $\theta_1 > \theta_2 > \theta_3 > ...$ accumulating at $\theta_\infty=0$~\cite{integrabilitychiral}.

\medskip

Except for a countable set of twist angles (which has nothing to do with the countable set of magic angles mentioned above), TBG is {\em not} a periodic system at the atomic scale. It is an incommensurate system, which makes its mathematical description quite involved. At the tight-binding level of theory, incommensurate multilayer 2D materials can be efficiently described using a C*-algebra formalism~\cite{cances2017generalized}, in the spirit of the approach introduced by Belissard and co-workers to handle disordered cyrstals~\cite{Bellissard2002,bellissard1994non,prodan2017computational}. The C*-algebra formalism has been used to compute macroscopic physical quantities of interest such as conductivity matrices~\cite{etter2020modeling,watson2023mathematical,massatt2020efficient}, or to study atomic relaxation in multilayer 2D materials~\cite{massatt2023electronic,zhu2020modeling,cazeaux2023relaxation,cazeaux2020energy}. For TBG, it has been shown both theoretically and experimentally that when the twist angle $\theta$ becomes very small (typically lower than $0.5^\circ$), atomic relaxation becomes crucial and leads to the formation of moir\'e-scale energy favorable AB or BA Bernal stacking triangular regions separated by relatively thin walls. Formally extending the BM model to this setting, it has been shown that in the presence of gating, topologically protected transport on the triangular lattice of domain walls should occur~\cite{bal2023mathematical}.

\medskip

A natural mathematical question is to investigate the connections between atomic-scale models of TBG and effective moiré-scale models such as the BM model. One approach is to try and connect the two scales by analyzing the dynamics of low-energy traveling wave-packets. This idea was already used in e.g.~\cite{gerard1997homogenization} to derive the effective mass model for conducting electrons and holes in semiconductors, and in~\cite{fefferman2014wave} to derive the massless Dirac model for low-energy excitations in monolayer graphene. Using this approach, it was shown in \cite{TightBM} that the BM model can be obtained as the limit when the twist angle goes to zero of an atomic-scale tight-binding model provided the hoping terms satisfy some decay and scaling properties. Likewise, it was shown in~\cite{cances2022simple}, that an effective moiré-scale model similar, though not identical, to the BM model could be formally obtained in the small twist angle limit by a simple variational approximation of the time-dependent Schr\"odinger equation built from an approximate two-scale Kohn-Sham Hamiltonian 
\begin{equation}\label{eq:approx_KSH}
H_{d,\theta}:=-\frac 12 \Delta_{x,z} + V_d\left( (1-\epsilon^2)^{1/2} x, \epsilon x,z \right) \quad \mbox{on } L^2(\R^3;\C), \quad \mbox{with} \quad \epsilon:=\sin \frac\theta 2, 
\end{equation}
where $V_d:\R^2 \times \R^2 \times \R \to \R$ is a potential periodic with respect to each of the first two variables (with different periods) and vanishing at infinity in the third variable (see Section~\ref{sec:KS-TBG} for details).
Note that $H_{d,\theta}$ is incommensurate for all but a countable number of values of $\theta$.

\medskip

In this article, we follow a different approach and study the density-of-states (DoS) of the two-scale Hamiltonian $H_{d,\theta}$ using tools from semiclassical analysis. We first prove that the DoS of $H_{d,\theta}$ admits an asymptotic expansion in $\epsilon:=\sin(\theta/2)$. We then establish a similar result for a family of effective moiré-scale Hamiltonians $H_{d,K,\theta}^{\rm eff}$ containing the BM Hamiltonian as well as the leading term of the Hamiltonian derived in~\cite{cances2022simple}. We also show that $H_{d,\theta}$ admits an asymptotic expansion in $\eta:=\tan(\theta/2)$ and comment on the differences between the expansions in $\epsilon$ and $\eta$.

 Semiclassical analysis has been widely used in solid states physics in the past decades, notably to derive effective semiclassical dynamics of Bloch electrons in periodic crystals subjected to external electromagnetic fields~\cite{Buslaev_1987,panati2003effective,helffer1990diamagnetism,guillot1988semi,rammal1990algebraic,nenciu1991dynamics}, or to mathematically justify reduced quantum models~\cite{gerard1991mathematical,dimassi1993developpements,gerard1997homogenization}. In particular, two-scale single-particle Hamiltonians modeling perfect crystals in the presence of space-dependent electromagnetic potentials $(A,\phi)$ slowly varying at the atomic scale, have attracted a lot of attention. Such Hamiltonians are of the form
\begin{align}\label{eq:Hamiltonian-GMS}
    H_\epsilon^{A,\phi}=(-i\nabla_x +A(\epsilon x))^2+V(x)+\phi(\epsilon x) \quad \mbox{on } L^2(\R^3;\C),
\end{align}
where $V: \R^3 \to \R$ is a $\Gamma$-periodic potential, with $\Gamma$ a 3D lattice of $\R^3$, representing the mean-field potential in the perfect crystal, $A: \R^3 \to \R^3$ and $\phi: \R^3 \to \R$, and can be embedded in a more general family of pseudodifferential operators of the form
$$
P_\epsilon=P(\epsilon x,x,-i\nabla_x+A(\epsilon x)) \quad \mbox{on } \R^d.
$$
It was observed in \cite{Buslaev_1987} that if $\widetilde u_\epsilon : \R^3 \times \R^3 \to \C$ is a solution $\Gamma$-periodic with respect to the first variable of 
\begin{align*}
  \widetilde{P}_\epsilon \widetilde u_\epsilon= \lambda \widetilde u_\epsilon \quad \mbox{with} \quad \widetilde{P}_\epsilon = P(X,x,-i\nabla_x-\epsilon\nabla_X +A(X)) \quad \mbox{on } \R^{2d},
\end{align*}
then ${u}_\epsilon(x):=\widetilde u(x,\epsilon x)$ satisfies
\begin{align}\label{eq:u-tilde-eigen}
    P_\epsilon u_\epsilon=\lambda u_\epsilon.
\end{align}
The author then uses this idea to construct asymptotic solutions of \eqref{eq:u-tilde-eigen} by considering $ \widetilde{P}_\epsilon$ as a pseudo-differential operator in the mesoscopic variable $X$ with operator-valued symbol. A similar approach, combined with the introduction of a suitable Grushin problem, was used in \cite{gerard1991mathematical} to construct a reduced effective Hamiltonian having the same spectrum as $P_\epsilon$ in the neighborhood of some given $\lambda_0 \in \R$.
An asymptotic expansion in $\epsilon$ of  $\tr(f(H^{A,\phi}_\epsilon))$ for $f: \R \to \R$ supported in a spectral gap of $H^{0,0}_0$, was then given  
in \cite{dimassi1993developpements}. Semiclassical analysis with operator-valued symbols, combined with adiabatic methods, was also used in \cite{panati2003effective} to provide a mathematical justification of Peierls substitution method, and compute the first-order correction to the standard semiclassical dynamics of Bloch electrons, involving Berry curvature terms (see also \cite{nenciu1991dynamics} and references therein).

\medskip

To show that the DoS of the two-scale Hamiltonian $H_{d,\theta}$ defined in~\eqref{eq:approx_KSH} can be expanded as an asymptotic series in powers of $\epsilon:=\sin(\theta/2)$, we start from the observation that $H_{d,\theta}$ is unitary equivalent to a pseudodifferential operator with operator-valued symbol $h_{d,\epsilon}(k,X)$, see \eqref{eq:Hdeps=Opeps} below. The symbol $h_{d,\epsilon}(k,X)$ is quasi-periodic with respect to the momentum variable $k \in \R^2$, and depends on the variable $X \in \R^2$ representing physically the position in the moir\'e cell, or equivalently the local stacking parameter. However, this is not  a semiclassical symbol in the sense that it cannot be directly expanded as an asymptotic series in $\epsilon$ of symbols in a suitable class. To circumvent this technical issue, we introduce a twisted version of the standard Weyl calculus (see \eqref{eq:TWQ}) and show that $H_{d,\theta}$ is unitary equivalent to the twisted Weyl quantization of the $\epsilon$-independent symbol $h_{d,0}$:
$$
H_{d,\theta} = \cU^{-1} {\rm Op}_\epsilon^c(h_{d,0}) \cU,
$$
where $\cU$ is the 2D Bloch transform with respect to the in-plane coordinates and ${\rm Op}_\epsilon^c$ the twisted Weyl quantization operator.
We then use this relation to derive an asymptotic expansion in $\epsilon:=\sin(\theta/2)$ of the density of states of $H_{d,\theta}$. To derive the asymptotic expansion of the DoS of $H_{d,\theta}$ in powers of $\eta:=\tan(\theta/2)$, we first rescale the $x$-variable and then show that the so-obtained operator is unitary equivalent to a pseudodifferential operator with operator-valued symbol in a usual class of symbols.

\medskip

The article is organized as follows. In Section~\ref{sec:models}, we present the various mathematical models we will be dealing with, namely the monolayer graphene Kohn-Sham Hamiltonian $H_{\rm MG}$ (Section~\ref{sec:MG}), the TBG approximate Kohn-Sham Hamiltonian $H_{d,\theta}$ introduced in~\cite{cances2022simple}  (Section~\ref{sec:KS-TBG}), and the TBG moir\'e-scale Hamiltonians $H_{K,\theta}^{\rm BM}$ and $H_{d,K,\theta}^{\rm eff}$ (Section~\ref{sec:MS-TBG}). In Section~\ref{sec:mainresult}, we state our main results. In Section~\ref{sec:hd0}, we investigate the properties of the operator-valued symbol $h_{d,0}(k,X)$. In Section~\ref{sec:Weyl-moire}, we first recall the basics of the Weyl quantization of operator-valued symbols (Section~\ref{sec:Weyl}). We then derive trace formulae for operator-valued symbols useful for the study of the DoS of $H_{d,\theta}$, $H_{K,\theta}^{\rm BM}$ and $H_{d,K,\theta}^{\rm eff}$ (Section~\ref{sec:traceform}), and study the twisted Weyl quantization operator ${\rm Op}_\epsilon^c$ (Section~\ref{sec:TWC}). The proofs of our main results, namely Theorems~\ref{DoS-TBG} and \ref{DoS-BM}, are postponed until Sections~\ref{sec:red-TBG} and \ref{sec:red-BM}, respectively.

The relationships between the density of states of the atomic-scale Hamiltonian $H_{d,\theta}$ around the Fermi level and the density of states of the effective moir\'e-scale Hamiltonian $H_{d,K,\theta}^{\rm eff}$ will be investigated in a forthcoming paper.

\section{Mathematical models}
\label{sec:models}

\subsection{Monolayer graphene}\label{sec:MG}

We first set up some notation and review the basic properties of the monolayer graphene (MG) Kohn-Sham Hamiltonian that are relevant to this work. Before going further, we shall point out that the MG Hamiltonian we are considering here is an operator on $L^2(\mathbb{R}^3;\mathbb{C})$, where $\mathbb{R}^3$ is identified by the physical three-dimensional space, while the model studied in \cite{fefferman2012honeycomb,fefferman2014wave} is an effective atomic-scale Hamiltonian acting on $L^2(\mathbb{R}^2;\mathbb{C})$, where $\mathbb{R}^2$ is identified with the plane containing the monolayer. 

We denote the 3D position variable as $(x,z)^T$ where $x=(x_1,x_2)\in\mathbb{R}^2$ and $z\in\mathbb{R}$ are respectively the longitudinal (in-plane) and transverse (out-of-plane) position variables.

We consider in this section a graphene monolayer in the horizontal plane $z=0$ with Bravais lattice (see Fig. \ref{fig:1a})
\begin{align*}
\mathbb{L}=\mathbb{Z}{\rm a}_1+\mathbb{Z}{\rm a}_2\qquad \mathrm{with}\quad {\rm a}_1=a_0\begin{pmatrix}1/2\\ -\sqrt{3}/2\end{pmatrix}\quad\mathrm{and}\quad {\rm a}_2=a_0\begin{pmatrix}1/2\\ \sqrt{3}/2\end{pmatrix},
\end{align*}
where $a_0$ is the graphene lattice constant ($a_0=\sqrt{3}a$, where $a$ the carbon-carbon distance in the honeycomb lattice). 
We also define the reciprocal lattice of $\mathbb{L}$ by
\[
\mathbb{L}^*:=\{G\in\mathbb{R}^2 \;| \; \forall R\in\mathbb{L}, \;   e^{2\pi i G \cdot R}=1\},
\]
(see Fig. \ref{fig:1b}). More explicitly,
\begin{align*}
    \mathbb{L}^* ={\rm a}_1^*\mathbb{Z}+{\rm a}_2^*\mathbb{Z},\quad\textrm{with}\quad {\rm a}_1^*=\sqrt{3}k_D\begin{pmatrix}\sqrt{3}/2\\ -1/2\end{pmatrix}\quad\mathrm{and}\quad {\rm a}_2^*=\sqrt{3}k_D\begin{pmatrix}\sqrt{3}/2\\ 1/2\end{pmatrix},
\end{align*}
where $k_D:=\frac{4\pi}{3a_0}$. Then the corresponding Wigner-Seitz cell $\Omega$ and first Brillouin zone $\Omega^*$ can be identified with the two-dimensional tori
\[
\Omega \equiv \mathbb{R}^2/\mathbb{L},\quad \Omega^* \equiv \mathbb{R}^2/\mathbb{L}^*,
\]
respectively. The carbon atoms in the unit cell $\Omega$ are located at points 
$\frac 13 a_1 + \frac 23 a_2$ and $\frac 23 a_1 + \frac 13 a_2$. The origin of the Cartesian frame then lies at the center of a regular hexagon formed by six carbon atoms.

The MG Kohn-Sham Hamiltonian is the operator on $L^2(\mathbb{R}^3;\mathbb{C})$ with domain $H^2(\mathbb{R}^3;\mathbb{C})$ defined by
\begin{align}\label{eq:H-MG}
    h_{\mathrm{MG}}:=-\frac{1}{2}\Delta_x^2-\frac{1}{2}\partial_z^2+V_{\mathrm{MG}}(x,z).
\end{align}
We assume that the potential $V_{\mathrm{MG}}$ satisfies the following properties (compare with \cite[Definition 2.1]{fefferman2012honeycomb}). 
\begin{assumption}[MG potentials]\label{def:MGpotential}
The potential $V_{\rm MG}$ satisfies
\begin{enumerate}
\item\label{Ass:2.1.1} (Smoothness) $V_{\mathrm{MG}} \in C^\infty(\mathbb{R}^2\times \mathbb{R};{\mathbb R})$;
    \item\label{Ass:2.1.2} (Symmetry)  $V_{\mathrm{MG}}$ is invariant under the Dg80$={\rm D}_{\rm 6h} \ltimes  {\mathbb L}$ (p6/mmm) symmetry group of graphene, i.e. for all $(x,z) \in \mathbb{R}^2 \times \mathbb{R}$,
    \begin{align*}
    & V_{\mathrm{MG}}(x+R,z)=V_{\mathrm{MG}}(x,z) \quad \mbox{for all } R \in {\mathbb L},  \\
    & V_{\mathrm{MG}}(x_1,-x_2,z)=V_{\mathrm{MG}}(x_1,x_2,z)=V_{\mathrm{MG}}(-x_1,x_2,z) \\
    &V_{\mathrm{MG}}(x,z)= V_{\mathrm{MG}}(x,-z), \\
    & [{\rm R}_{\frac{\pi}{3}}V_{\mathrm{MG}}](x,z)=V_{\mathrm{MG}}(\mathcal{R}_{\frac{\pi}{3}}^{-1} x,z)=V_{\mathrm{MG}}(x,z),
    \end{align*}
    where $\mathcal{R}_\theta$ denotes the 2D rotation matrix of angle $\theta\in\mathbb{R}$:
    \begin{align}\label{eq:R}
        \mathcal{R}_\theta:=\begin{pmatrix}
\cos\theta&-\sin\theta\\
\sin\theta&\cos\theta
\end{pmatrix}
    \end{align}
and ${\rm R}_{\frac{\pi}{3}}$ its action on the functions of $L^1_{\rm loc}(\R^3;\C)$;
\item\label{Ass:2.1.3} (Decay property in the transverse direction) for all $\alpha\in \mathbb{N}^{3}$, there is $\delta_\alpha>0$ and $C_\alpha \in \R_+$ such that 
    \begin{align} \label{eq:decay_MG}
        \sup_{x\in \mathbb{R}^2} |\partial_{x,z}^\alpha V_{\mathrm{MG}}(x,z)|\leq C_{\alpha}(1+|z|^2)^{-\delta_\alpha}\qquad\textrm{with}\quad\partial_{x,z}:=(\partial_{x},\partial_z).
    \end{align}
\end{enumerate}
\end{assumption}

In this setting, the Bloch transform (also called Zak transform) is the unitary operator 
\begin{equation}\label{eq:Bloch_1}
\cU : L^2(\R^3;\mathbb{C}) \to \cH:=L^2_{\rm qp}(\Omega^*;L^2_{\rm per})
\end{equation}
such that
\begin{align*}
&\forall u \in C^\infty_{\rm c}(\R^3;\C), \quad (\cU u)_k(x,z) = \sum_{R \in \mathbb{L}} u(x+R,z) e^{-ik \cdot (x+R)},  \\
& L^2_{\rm per}:=\{f \in L^2_{\rm loc}(\mathbb{R}^2; L^2(\mathbb{R};\mathbb{C})) \; |  \; \forall R\in\mathbb{L}, 
\; f(x-R,z)= f(x,z) \mbox{ for a.a. } (x,z) \in \R^2 \times \R \}, \\
&L^2_{\rm qp}(\Omega^*;L^2_{\rm per}):= \{ u_\bullet \in L^2_{\rm loc}(\R^2;L^2_{\rm per}) \; | \; \forall G\in\mathbb{L}^*,  \; u_{k-G} = \tau_G u_k \mbox{ for a.a. } k \in \R^2 \},
\end{align*}
where $\tau_G$ is the unitary operator on $L^2_{\rm per}$ acting by multiplication by the $\mathbb{L}$-periodic function $\R^3 \ni (x,z) \mapsto e^{iG \cdot x} \in U(1)$. The space $L^2_{\rm per}$ is
endowed with the inner product
\begin{align*}
    \langle u,v \rangle_{L^2_{\rm per}} :=\int_{\Omega\times \mathbb{R}} u^*(x,z)v(x,z) \, dx \, dz,
\end{align*}
and the space $\cH$ with the inner product
$$
\langle u_\bullet ,v_\bullet \rangle_\cH :=\fint_{\Omega^*} \langle u_k, v_k \rangle_{L^2_{\rm per}}  \, dk.
$$
We also set, for all $s\in \N$,
$$
H^s_{\rm per}:=\{ u \in L^2_{\rm per} \; | \; \partial_x^\alpha\partial_z\beta u \in L^2_{\rm per}, \, \forall (\alpha,\beta) \in \N^2 \times \N \mbox{ s.t. } |\alpha|+\beta\le s \}
$$
(endowed with its natural inner product). Using Fourier series in the $x$ variable and Fourier transform in the $z$ variable, Sobolev spaces $H^s_{\rm per}$ can be defined for any $s \in \R$. Note that the functions in $L^2_{\rm per}$ and $H^s_{\rm per}$ are $\mathbb L$-periodic in the longitudinal (in-plane) variable $x$ but not in the transverse (out-of-plane) variable $z$. 

\medskip

Since $H_{\mathrm{MG}}$ is $\mathbb{L}$-periodic, it is decomposed by the Bloch transform $\cU$: 
$$
H_{\mathrm{MG}} = \cU^{-1} \left( \fint_{\Omega^*}^\oplus  h_{\mathrm{MG}}(k) \, dk \right) \cU,
$$
where $h_{\mathrm{MG}}(k)$ is the operator on $L^2_{\rm per}$ with domain $H^2_{\rm per}$ given by
$$
h_{\mathrm{MG}}(k):= \frac{1}{2}(-i\nabla_x+k)^2-\frac{1}{2}\partial_z^2+V_{\rm MG}(x,z).
$$

\medskip 

\begin{remark}
It easily follows from properties \ref{Ass:2.1.1} and \ref{Ass:2.1.3} in Assumption~\ref{def:MGpotential} that $h_{\mathrm{MG}}$, hence all the $h_{\mathrm{MG}}(k)$'s, are self-adjoint, and that for all $k\in \Omega^*$
\begin{align}
    \sigma_{\rm ess}(h_{\mathrm{MG}}(k)):= [0,+\infty). \label{eq:sigma_ess_hMG}
\end{align}
While all-electron Kohn-Sham potentials have Coulomb singularities at nuclear positions and thus fail to satisfy the smoothness property, the local component of the MG Kohn-Sham potential obtained with the Goedecker-Teter-Hutter (GTH) norm-conserving pseudopotential \cite{GTH1996} and standard exchange-correlation functional should satisfy all the properties in Assumption~\ref{def:MGpotential}.
\end{remark}

\subsection{Approximate Kohn-Sham model for twisted bilayer graphene}
\label{sec:KS-TBG}

We consider two parallel layers of graphene, separated by a distance $d>0$, and with a twist angle $\theta\in [0,\frac{\pi}{3})$ between the two, and neglect atomic relaxation. More precisely, the upper layer is placed in the plane $z=d/2$ and rotated by $-\theta/2$ around the $z$-axis, and the lower layer is placed in the plane $z=-d/2$ and rotated by $\theta/2$ around the $z$-axis.

\medskip

We label the so-obtained configuration using the parameters $d$ and $\epsilon:=\sin\left(\frac{\theta}{2}\right)$. For convenience, we also set $c(\epsilon):=\frac{1-\sqrt{1-\epsilon^2}}{\epsilon}=\frac{1-\cos\left(\frac{\theta}{2}\right)}{\sin\left(\frac{\theta}{2}\right)}$. Note that the function $c(\epsilon)$ is continuous on $[-1,1]$ and real-analytic in $(-1,1)$, and that the rotation matrix $\mathcal{R}_{\theta/2}$ can be decomposed as
\[
\mathcal{R}_{\theta/2}=(1-c(\epsilon) \epsilon)\mathbb{I}_2-\epsilon J,\quad \textrm{with}\quad J:=\begin{pmatrix}0&1\\-1&0\end{pmatrix}.
\]
We then introduce the unitary operator $U_{d,\epsilon}$ on $L^2(\mathbb{R}^3;\mathbb{C})$ defined by
\begin{align}\label{eq:U-dX-trans}
    (U_{d,\epsilon}f)(x,z)=f\left(\mathcal{R}_{\theta/2}x,z-\frac{d}{2}\right)=f\left((1-\epsilon c(\epsilon))x-\epsilon J x,z-\frac{d}{2}\right),
\end{align}
and the following approximation of the TBG Kohn-Sham potential:
\begin{align}\label{eq:V-TBG-0}
    V_{{\rm TBG},d,\epsilon}^{\rm approx}(x,z) :&=(U_{d,\epsilon} V_{\mathrm{MG}})(x,z)+(U^{-1}_{d,\epsilon} V_{\mathrm{MG}})(x,z)+V_{{\mathrm{int}},d}(z)\notag\\
    &=\sum_{s=\pm 1}V_{\mathrm{MG}}\left((1- \epsilon c(\epsilon))x -s \epsilon  Jx,z-s\frac{d}{2}\right)+V_{{\mathrm{int}},d}(z),
\end{align}
where $V_{{\mathrm{int}},d}$ is a smooth function of the transverse variable $z$ only, which can be computed numerically \cite{cances2022simple}. This approximation was introduced and justified by numerical simulations in \cite{TSKCLPC2016}. Interestingly, $V_{{\rm TBG},d,\epsilon}^{\rm approx}$ has a simple multiscale structure:
\begin{align}\label{eq:V-TBG}
V_{{\rm TBG},d,\epsilon}^{\rm approx}(x,z) := V_{d}\left(\sqrt{1-\epsilon^2} x,\epsilon x,z\right) = V_{d}(x-\epsilon c(\epsilon) x,\epsilon x,z),
\end{align}
where the function 
$$
V_{d}(x,X,z) := \sum_{\sigma=\pm 1}V_{\mathrm{MG}} \left(x -\sigma J X,z-\sigma \frac{d}{2}\right)+V_{{\mathrm{int}},d}(z)
$$
is $\mathbb L$-periodic w.r.t. the atomic-scale variable $x$ and $J\mathbb L$-periodic w.r.t. the mesoscale variable $X$. The approximate Kohn-Sham Hamiltonian is then defined as
\begin{align}
    H_{d,\theta}&= - \frac{1}{2}\Delta_x-\frac{1}{2}\partial_z^2 +V_{d}\left(\sqrt{1-\epsilon^2} x,\epsilon x,z\right) \nonumber \\
    &=- \frac{1}{2}\Delta_x-\frac{1}{2}\partial_z^2 +V_{d}(x- c(\epsilon) \epsilon x,\epsilon x ,z). \label{eq:H-epsil}
\end{align}

\medskip

To proceed further, we have to make assumptions on the mathematical properties of the potential $V_{{\mathrm{int}},d}$. 
\begin{assumption}[Interlayer correction potential]\label{def:inter}
The interlayer potential $V_{{\mathrm{int}},d} \in C^\infty(\mathbb{R};\R)$ satisfies
\begin{enumerate}
\item (Smoothness) $V_{{\mathrm{int}},d} \in C^\infty(\mathbb{R};{\mathbb R})$;
    \item (Decay property) Let $d > 0$. 
    For all $\alpha\in \mathbb{N}^{3}$, there exist $\delta_\alpha >0$ and $C_\alpha \in \R_+$ such that 
    \begin{align} \label{eq:decay_V_int}
     \forall z \in {\mathbb R}, \quad   |\partial_z^\alpha V_{{\mathrm{int}},d}(z)|\leq C_{\alpha}(1+|z|^2)^{-\delta_\alpha}.
    \end{align}
\end{enumerate}
\end{assumption}
These assumptions are compatible with the numerical experiments we have done with the DFTK software \cite{DFTK} using GTH pseudopotentials.

\subsection{Effective moir\'e-scale models for twisted bilayer graphene}
\label{sec:MS-TBG}

Most theoretical investigations on TBG rely on continuous models \cite{bistritzer2011moire1,bistritzer2011moire2,mele2010commensuration,dos2012continuum} such as the Bistritzer-MacDonald (BM) model \cite{bistritzer2011moire2}, a 2D effective model at the moir\'e scale describing low-energy excitations in the $K$-valley. Up to a unitary transform, the BM Hamiltonian also describes the behavior of low-energy excitations in the $K'$-valley. After suitable gauge transform, the BM Hamiltonian is the self-adjoint operator on $L^2(\mathbb{R}^2;\mathbb{C}^4)$ with domain $H^1(\mathbb{R}^2;\mathbb{C}^4)$ given by
\begin{align*}
    H_{K,\theta}^{\rm BM}=\begin{pmatrix}v_{\rm F} \pmb{\sigma}_{-\theta/2}\cdot (-i\nabla_x-K) & \mathbf{V}(\epsilon x)\\ \mathbf{V}(\epsilon x)^*& v_{\rm F} \pmb{\sigma}_{\theta/2}\cdot (-i\nabla_x-K)\end{pmatrix} \qquad \mbox{(with $\epsilon := \sin \frac \theta 2$)},
\end{align*}
where $v_{\rm F} > 0$ is the Fermi velocity in monolayer graphene, and $\pmb{\sigma}_{\pm\theta/2}:= e^{\mp i\frac{\theta}{4}\sigma_3}(\sigma_1,\sigma_2)e^{\pm i\frac{\theta}{4}\sigma_3}$, with
\begin{align*}
    \sigma_1:=\begin{pmatrix} 0&1\\1&0\end{pmatrix},\quad \sigma_2:=\begin{pmatrix} 0&-i\\i&0\end{pmatrix},\quad \sigma_3:=\begin{pmatrix} 1&0\\0&-1\end{pmatrix} \qquad \mbox{(Pauli matrices)}.
\end{align*}
The function $\mathbf{V}:\mathbb{R}^2\to \mathbb{C}^2$ is $J\mathbb{L}$-periodic and depends on two empirical parameters $w_{\rm AA}$ and $w_{\rm AB}$ describing interlayer coupling in AA and AB (or BA) stackings respectively.

In this paper, we consider a more general class of effective moir\'e-scale Hamiltonians of the form:
\begin{align} \label{eq:EMS_model}
    H_{d,K,\theta}^{\rm eff}:=T_\epsilon(-i\nabla_x-K) +\mathcal{V}_{d,K}(\epsilon x) \qquad \mbox{(with $\epsilon := \sin \frac \theta 2$)},
\end{align}
where $\mathcal{V}_{d,K}(X)\in  C^\infty(\R^2;\C^{4\times 4}_{\rm herm})$ is a $J\mathbb{L}$-periodic potential and
 \begin{align}\label{eq:T-V-1}
  T_\epsilon(\kappa):=\begin{pmatrix}
   v_{\rm F}\pmb{\sigma}_{-\theta/2}\cdot \kappa & 0\\
   0&  v_{\rm F}\pmb{\sigma}_{\theta/2}\cdot \kappa \end{pmatrix} \qquad \mbox{(with $\epsilon := \sin \frac \theta 2$)}.
\end{align}
We use the subscript $d$ to emphasize that $\mathcal{V}_{d,K}$ and $ H_{d,K,\theta}^{\rm eff}$ actually dependent on $d$. This class of Hamiltonian contains the BM Hamiltonian as well as the leading term of the effective Hamiltonian derived in \cite{cances2022simple}.

\medskip

The Hamiltonians $H^{\rm BM}_{K,\theta}$ and $H_{d,K,\theta}^{\rm eff}$ are both $J\mathbb L$-periodic Hamiltonians on $L^2(\mathbb{R}^2;\mathbb{C}^4)$ and are therefore decomposed by the 2D Bloch transform
\begin{equation}\label{eq:Bloch_2}
\widetilde\cU : L^2(\R^2;\mathbb{C}^4) \to \widetilde\cH:=L^2_{\rm qp}(\Omega^*;L^2_\#)
\end{equation}
with
\begin{align*}
&\forall u \in C^\infty_{\rm c}(\R^2;\C^4), \quad (\widetilde\cU u)_k(x) = \sum_{R \in \mathbb{L}} u(x+R) e^{-ik \cdot (x+R)},  \\
& L^2_\#:=\{f \in L^2_{\rm loc}(\mathbb{R}^2;\mathbb{C}^4)) \; |  \; \forall R\in\mathbb{L}, 
\; f(x-R)= f(x) \mbox{ for a.a. } x \in \R^2 \}, \\
&L^2_{\rm qp}(\Omega^*;L^2_\#):= \{ u_\bullet \in L^2_{\rm loc}(\R^2;L^2_\#) \; | \; \forall G\in\mathbb{L}^*,  \; u_{k-G} = \tau_G u_k \mbox{ for a.a. } k \in \R^2 \},
\end{align*}
where $\tau_G$ is the unitary operator on $L^2_\#$ acting by multiplication by the $\mathbb{L}$-periodic function $\R^2 \ni x \mapsto e^{iG \cdot x} \in U(1)$. The space $L^2_\#$ is
endowed with its usual inner product
\begin{align*}
    \langle u,v \rangle_{L^2_\#} :=\int_{\Omega} u^*(x)v(x) \, dx,
\end{align*}
and the space $\widetilde\cH$ with the inner product
$$
\langle u_\bullet ,v_\bullet \rangle_{\widetilde\cH} :=\fint_{\Omega^*} \langle u_k, v_k \rangle_{L^2_\#}  \, dk.
$$
We finally denote by $H^s_\#$ the corresponding periodic Sobolev spaces
(endowed with their usual inner products).

\section{Main results}\label{sec:mainresult}

We will analyze the atomic-scale Hamiltonian $H_{d,\theta}$ and the moir\'e-scale Hamiltonian $H_{d,K,\theta}^{\rm eff}$ through their density-of-states (DoS). 
Consider a family $(\chi_N)_{N \in \N^*} \in (C^\infty_{\rm c}(\R^2;\R))^{\N^*}$ of spatial cutoff functions such that
\begin{align}
    & 0\leq \chi_N(x)\leq 1 \quad \mbox{ for all } x \in \R^2, \label{eq:chi_1} \\
    &  \chi_N(x)=1 \textrm{ if }\;x\in NJ\Omega, \quad \textrm{and}\quad \chi_N(x)=0 \textrm{ if }\;x\notin (N+1)J\Omega, \label{eq:chi_2} \\
    &  \sup_{N \in \N} \|\partial^\alpha \chi_N \|_{L^\infty} < +\infty \quad \mbox{ for all  } \alpha\in\mathbb{N}^2. \label{eq:chi_3}
\end{align}
For any $f \in C^\infty_{\rm c}(\R;\R_+)$,  the traces per unit area of the operators $f(H_{d,\theta})$ and $f(H^{\rm eff}_{d,K,\theta})$ are defined as
\begin{align} \label{eq:DoS_H}
    \Trb[f(H_{d,\theta})]:&=\lim_{N\to +\infty}\frac{1}{|NJ\Omega|}\Tr_{L^2(\mathbb{R}^3)} \left( \chi_N f(H_{d,\theta})\chi_N \right)
\end{align}
and
\begin{align}  \label{eq:DoS_T}
    \Trb[f(H_{d,K,\theta}^{\rm eff})]:&=\lim_{N\to +\infty}\frac{1}{| NJ\Omega|}\Tr_{L^2(\mathbb{R}^2)}\left( \chi_N f(H_{d,K,\theta}^{\rm eff})\chi_N \right)
\end{align}
respectively.
\begin{remark} The limit in \eqref{eq:DoS_T} exists (in $\R_+ \bigcup \{+\infty\}$) because $f$ is nonnegative and $H_{d,K,\theta}^{\rm eff}$ is an $\epsilon^{-1} J\mathbb L$-periodic operator. The limit in \eqref{eq:DoS_H} is also well-defined because $f$ is nonnegative and $H_{d,\theta}$ is either periodic (for a countable set of twist angles) or ergodic.
\end{remark}
 The DoS of $H_{d,\theta}$ in the negative energy window and the DoS of $H_{d,K,\theta}^{\rm eff}$ are then the positive Borel measures $\nu_{H_{d,\theta}}$ and $\nu_{H_{d,K,\theta}^{\rm eff}}$ on $(-\infty,0)$ and $\R$ respectively such that 
\begin{align}
 \forall f \in C^\infty_{\rm c}(\R;\R_+) \mbox{ with support in } (0,+\infty), \quad  & \Trb[f(H_{d,\theta})] = \int_{(-\infty,0)} f(E) \, d\nu_{H_{d,\theta}}(E),  \label{eq:DoS_H_2} \\
\forall f \in C^\infty_{\rm c}(\R;\R_+), \quad  & \Trb[f(H_{d,K,\theta}^{\rm eff})] = \int_\R f(E) \, d\nu_{H_{d,K,\theta}^{\rm eff}}(E).  \label{eq:DoS_T_2}
\end{align}
\medskip

To formulate our main results, we need to recall some concepts and tools of semiclassical and spectral analysis. The first concept is the one of operator-valued symbols introduced in~\cite{balazard1985calcul,gerard1991mathematical} and used notably e.g.~\cite{panati2003effective,teufel2003adiabatic} to analyze the effective dynamics of Bloch electrons in the adiabatic approximation. The key formulae in our approach are
\begin{equation}\label{eq:Hdeps=Opeps}
H_{d,\theta} = \cU^{-1} {\rm Op}_\epsilon(h_{d,\epsilon}) \cU,\qquad H_{d,K,\theta}^{\rm eff} =\widetilde\cU^{-1} {\rm Op}_\epsilon(h^{\rm eff}_{d,K,\epsilon}) \qquad \mbox{(with $\epsilon:=\sin\frac\theta 2$)},
\end{equation}
(see \cite[Eqns. (12)-(13)]{panati2003effective}) where 
\begin{itemize}
\item $\cU: L^2(\R^3) \to \cH$ and $\widetilde\cU: L^2(\R^2) \to \widetilde\cH$  are the Bloch transforms defined in \eqref{eq:Bloch_1} and \eqref{eq:Bloch_2};
\item $h_{d,\epsilon}$ is the operator-valued symbol on $\R^2 \times \R^2$ such that for all $(k,X) \in \R^2 \times \R^2$,  $h_{d,\epsilon}(k,X)$ is the self-adjoint operator on $L^2_{\rm per}$ with domain $H^2_{\rm per}$ defined by: for all $u \in H^2_{\rm per}$ 
\begin{equation}\label{eq:def_hdeps}
[h_{d,\epsilon}(k,X) u](x,z) :=  \left[\frac 12 \left( -i \nabla_x+k \right)^2 u- \frac 12 \partial^2_z u \right](x,z)+ V_d(x-c(\epsilon)X,X,z) u(x,z);
\end{equation}
\item $h^{\rm eff}_{d,K,\epsilon}$ is the operator-valued symbol on $\R^2 \times \R^2$ such that for all $(k,X) \in \R^2 \times \R^2$,  $h^{\rm eff}_{d,K,\epsilon}(k,X)$ is the self-adjoint operator on $L^2_\#$ with domain $H^1_\#$ defined by: for all $u \in H^1_\#$ 
\begin{equation}\label{eq:def_tdeps}
[h^{\rm eff}_{d,K,\epsilon}(k,X) u](x) :=T_\epsilon(-i\nabla_x+k-K)u(x) +\mathcal{V}_{d,K}(X)u(x);
\end{equation}
\item  ${\rm Op}_\epsilon(a)$ is the self-adjoint operator on $\cH$ or $\widetilde \cH$ obtained from the operator-valued symbol $a(k,X)$ by the following Weyl quantization rule:
\begin{equation}\label{eq:Weyl_quantization}
[{\rm Op}_\epsilon(a) \phi]_k (r) = \frac{1}{(2\pi\epsilon)^2} \int_{\R^2 \times \R^2} \left[a\left( \frac{k+k'}2,X \right) \phi_{k'} \right](r) \; e^{-i \frac{(k-k')\cdot X}\epsilon} \, dk' \, dX.
\end{equation}
\end{itemize}

\begin{remark}\label{rem:weyl}
    The sign in the phase factor in the Weyl quantization formula~\eqref{eq:Weyl_quantization} is not the usual one. Formally, the operator defined by \eqref{eq:Weyl_quantization} is 
   $a(k,i \epsilon \nabla_k)$, while the usual definition leads to $a(k,- i \epsilon \nabla_k)$. This is due to the fact that in this setting,  the role usually assigned to the momentum variable is played by the position variable $X$. Note that in the special case when $a(k,-X)=a(k,X)$ for all $(k,X)$, these two definitions agree. 
   \end{remark}

\medskip

\begin{definition}[Operator-valued symbols]\label{def:symbols} Let $\cE$ be a Banach space of bounded operators, and let $\omega:\;\mathbb{R}^2\times \mathbb{R}^2\to [0,+\infty)$ be an order function in the sense that there exist $C_0 \in \R_+$ and $N_0 \in \R_+$ such that
\begin{align*}
 \forall (k,X), (k',X') \in  \mathbb{R}^{2}\times \mathbb{R}^2, \quad  \omega(k,X)\leq C_0(1+|k-k'|^2+|X-X'|^2)^{N_0/2}\omega(k',X').
\end{align*}
\begin{enumerate}
\item A symbol  $a \in C^\infty(\R^2 \times \R^2;\cE)$ is in the class $S^{\omega}(\cE)$ if for any $\alpha,\beta\in\mathbb{N}^2$, 
\begin{align}\label{norm:symbols}
    \cN^{\omega,\cE}_{\alpha,\beta}(a):= \sup_{k,X\in\mathbb{R}^2}\omega(k,X)^{-1}\left\|\partial_k^\alpha\partial_X^\beta a(k,X)\right\|_{\mathcal{E}} < \infty.
\end{align}
Let $\epsilon_0 > 0$. The function $a_\bullet : (0,\epsilon_0] \ni \epsilon \mapsto a_\epsilon \in S^{\omega}(\cE)$ is a symbol in the class $S^{\omega,\epsilon_0}(\cE)$ if for any $\alpha,\beta\in \N^2$, 
\begin{align*}
   \cN^{\omega,\cE,\epsilon_0}_{\alpha,\beta}(a_\bullet):= \sup_{\epsilon\in (0,\epsilon_0]}\cN^{\omega,\cE}_{\alpha,\beta}(a_\epsilon) < \infty.
\end{align*}
Endowed with the family of seminorms $\cN^{\omega,\cE}_{\alpha,\beta}$ (resp. $\cN^{\omega,\cE,\epsilon_0}_{\alpha,\beta}$), the vector space $S^{\omega}(\cE)$ (resp.  $S^{\omega,\epsilon_0}(\cE)$) is a Fr\'echet space.
\item  A function $a_\bullet \in S^{\omega,\epsilon_0}(\cE)$ is called an $\cE$-valued semiclassical symbol with order function $\omega$ if there exists a sequence of $\cE$-valued symbols $(a_j)_{j \in \N} \in (S^{\omega}(\cE))^\N$ such that for each $n \in \N$, 
\begin{equation}\label{eq:expansion_symbol}
a_\epsilon = \sum_{j=0}^n \epsilon^j a_j + O_{S^{\omega,\epsilon_0}(\cE)}(\epsilon^{n+1}).
\end{equation}
The class of semiclassical symbols of order $\omega$ is denoted as $S^{\omega,\epsilon_0}(\epsilon,\cE)$.
\end{enumerate}
\end{definition}

\begin{definition}[$\tau$-equivariant operator-valued symbols]\label{def:tau-equi} Let $\omega:\;\mathbb{R}^2\times \mathbb{R}^2\to [0,+\infty)$ be an order function. Assume $\cE = \cB(H^{s_1}_{\rm per},H^{s_2}_{\rm per})$ for some $s_1,s_2 \in \R$, or $\cE = {\mathfrak S}_1(L^2_{\rm per})$ (resp. $\cE = \cB(H^{s_1}_\#,H^{s_2}_\#)$  for some $s_1,s_2 \in \R$, or $\cE = {\mathfrak S}_1(L^2_\#)$). A symbol $a \in S^{\omega}(\cE)$ is called $\tau$-equivariant in the $k$-variable if 
\begin{equation}\label{eq:tau-equiv}
    \forall (k,X) \in \R^2 \times \R^2, \quad \forall G \in \mathbb{L}^*, \quad a(k-G,X) = \tau_G a(k,X) \tau_G^{-1},
    \end{equation}
    where $\tau_G$ is the multiplication operator on $H^s_{\rm per}$ (resp. $H^s_\#$) by the $\mathbb L$-periodic function $(x,z) \mapsto e^{iG \cdot x}$ (resp. $x\mapsto e^{iG \cdot x}$). We denote by 
    \begin{enumerate}
        \item $S^{\omega}_\tau(\cE)$ (resp. $S^{\omega,\epsilon_0}_\tau(\cE)$) the class of symbols in $S^{\omega}(\cE)$ (resp. in $S^{\omega,\epsilon_0}(\cE)$) which are $\tau$-equivariant with respect to the variable $k$;
        \item $S^{\omega}_{\tau,\rm per}(\cE)$ (resp. $S^{\omega,\epsilon_0}_{\tau,\rm per}(\cE)$) the class of symbols in $S^{\omega}(\cE)$ (resp. in $S^{\omega,\epsilon_0}(\cE)$) which are $\tau$-equivariant with respect to the variable $k$ and $J\mathbb L$-periodic with respect to the variable~$X$.
    \end{enumerate}

\end{definition}
Note that $\tau_G$ defines a bounded operator on $H^s_{\rm per}$ (resp. $H^s_\#$) for any $s \in \R$, which is in addition unitary for $s=0$.

\medskip

We will see in Section~\ref{sec:red-BM} that $h^{\rm eff}_{d,K,\bullet}$ is a $\tau$-equivariant semiclassical symbol in the class $S^{\omega,1}_{\tau,\rm per}(\cE)$  with $\omega(k,X):=(1+|\bk|^2)^{1/2}$ and $\cE:=\cB(H^1_\#;L^2_\#)$. This will allow us to study the asymptotic expansion in $\epsilon$ of $\Trb[f(H_{d,K,\theta}^{\rm eff})]$ using techniques introduced in \cite{dimassi1993developpements,panati2003effective}. The situation is a bit more involved for the atomic-scale model. For any $0 < \epsilon \le 1$, the symbol $h_{d,\epsilon}$ belongs to the class $S^{\omega}(\mathcal{E})$ with $\omega(k,X):=1+|k|^2$ and $\mathcal{E}:=\cB(H^2_{\rm per};L^2_{\rm per})$, and can thus be Weyl-quantized \cite{teufel2003adiabatic}. We have in fact $h_{d,\bullet} \in S^{\omega,1}_\tau(\cE)$. On the other hand, $h_{d,\bullet}$ is not a semiclassical symbol, because the coefficients of its asymptotic expansion are not all in~$S^{\omega}(\cE)$:
\begin{align}\label{eq:hdepsi-non-symbol}
    h_{d,\epsilon}(k,X) = \underbrace{h_{d,0}(k,X)}_{\in S^{\omega}(\cE)} - \frac{\epsilon}2 \underbrace{X \cdot \nabla_xV(x,X)}_{\notin S^{\omega}(\cE)} + \underbrace{\frac{\epsilon^2}{8} [D^2_{xx}V(x,X)](X,X)}_{\notin S^{\omega}(\cE)} + ...
\end{align}
In order to expand $\Trb[f(H_{d,\theta})]$ in powers of $\epsilon$, we circumvent this technical issue by introducing a twisted version of the standard Weyl calculus. The idea is to rewrite~\eqref{eq:Hdeps=Opeps} as
\begin{equation}\label{eq:Hdeps=Opceps}
H_{d,\theta} = \cU^{-1} {\rm Op}^c_\epsilon(h_{d,0}) \cU \qquad \mbox{(with $\epsilon:=\sin\frac\theta 2$)},
\end{equation}
where 
\begin{equation}\label{eq:TWQ}
 {\rm Op}^c_\epsilon(a_\epsilon):={\rm Op}_\epsilon\left( a_\epsilon^c \right) \quad \mbox{with} \quad
a^c_\epsilon(k,X) := \cT_{c(\epsilon) X} a_\epsilon(k,X)  \cT_{c(\epsilon) X}^{-1},
\end{equation}
where for each $x_0 \in \R^2$, $\cT_{x_0}$ is the translation operator defined on each $H^s_{\rm per}$ space by
$$
(\cT_{x_0} u)(x,z) = u(x-x_0,z).
$$
Let us now introduce classes of symbols containing $h_{d,0}$ as well as the symbols of $\cU f(H_{d,\theta}) \cU^{-1}$, and suitable for the twisted Weyl quantization rule ${\rm Op}_\epsilon^c$ defined in \eqref{eq:TWQ}. We set
$$
{\rm ad}_{\partial_x}:=[\partial_x,\bullet].
$$

\begin{definition}[Operator-valued symbols for twisted Weyl calculus]\label{def:tau-equi-g} 
Let $\cE = \cB(H^{s_1}_{\rm per},H^{s_2}_{\rm per})$ for some $s_1,s_2 \in \R$, or $\cE = {\mathfrak S}_1(L^2_{\rm per})$, and let $\omega:\;\mathbb{R}^2\times \mathbb{R}^2\to [0,+\infty)$ be an order function.
\begin{enumerate}
\item A symbol $a \in C^\infty(\R^2 \times \R^2;\cE)$ is in the class $S^{\omega}_{\rm twi}(\cE)$ if the following three conditions are satisfied:
\begin{description}
    \item[(i)] $a$ is $\tau$-equivariant in the variable $k$, i.e.
    $$
    \forall (k,X) \in \R^2 \times \R^2, \quad \forall G \in \mathbb{L}^*, \quad a(k-G,X) = \tau_G a(k,X) \tau_G^{-1};
    $$
    \item[(ii)] $a$ is $J\mathbb{L}$-periodic in the variable $X$, i.e.
    $$
    \forall (k,X) \in \R^2 \times \R^2, \quad \forall R \in J\mathbb{L}, \quad a(k,X-R) =  a(k,X);
    $$
    \item[(iii)] for all $\alpha,\beta,\gamma\in\mathbb{N}^2$, 
\begin{align}\label{norm:symbols_twi}
    \cN^{\omega,\cE}_{\alpha,\beta,\gamma}(a):= \sup_{k,X\in\mathbb{R}^2}\omega(k,X)^{-1}\left\|\partial_k^\alpha\partial_X^\beta\px^\gamma a(k,X)\right\|_{\mathcal{E}} < \infty.
\end{align}
Let $\epsilon_0 > 0$. The function $a_\bullet : (0,\epsilon_0] \ni \epsilon \mapsto a_\epsilon \in S^{\omega}_{\rm twi}(\cE)$ is a symbol in the class $S^{\omega,\epsilon_0}_{\rm twi}(\cE)$ if for any $\alpha,\beta,\gamma\in \N^2$, 
\begin{align*}
   \cN^{\omega,\cE,\epsilon_0}_{\alpha,\beta,\gamma}(a_\bullet):= \sup_{\epsilon\in (0,\epsilon_0]}\cN^{\omega,\cE}_{\alpha,\beta,\gamma}(a_\epsilon) < \infty.
\end{align*}
Endowed with the family of seminorms $\cN^{\omega,\cE}_{\alpha,\beta,\gamma}$ (resp. $\cN^{\omega,\cE,\epsilon_0}_{\alpha,\beta,\gamma}$), the vector space $S^{\omega}_{\rm twi}(\cE)$ (resp.  $S^{\omega,\epsilon_0}(\cE)$) is a Fr\'echet space.
\end{description}

\item  A function $a_\bullet \in S^{\omega,\epsilon_0}_{\rm twi}(\cE)$ is called a semiclassical symbol if there exists a sequence  $(a_j)_{j \in \N} \in (S^{\omega}_{\rm twi}(\cE))^\N$ such that for each $n \in \N$, 
\begin{equation}\label{eq:expansion_symbol-g}
a_\epsilon = \sum_{j=0}^n \epsilon^j a_j + O_{S^{\omega,\epsilon_0}_{\rm twi}(\cE)}(\epsilon^{n+1}).
\end{equation}
The class of semiclassical symbols of order $\omega$ is denoted as $S_{\rm twi}^{\omega,\epsilon_0}(\epsilon,\cE)$.
\end{enumerate}
\end{definition}

Note that if $a_\bullet \in S^{\omega,\epsilon_0}_{\rm twi}(\cE)$, then $a^c_\bullet \in S^{\omega,\epsilon_0}(\cE)$. Indeed, since 
\begin{align}\label{eq:partial}
    &\partial_k (a^c_\epsilon)=(\partial_k a_\bullet)^c_\epsilon,\qquad \partial_X (a^c_\epsilon)=\left(\partial_X a_\bullet \right)^c_\epsilon-c(\epsilon)\left(\px a_\bullet\right)^c_\epsilon, 
\end{align}
and $\cT_{x_0}$ is unitary on $H^s_{\rm per}$ for any $s \in \R$, there exists, for all $\beta \in \N^2$, a constant $C_\beta \in \R_+$ such that for all $\alpha \in \N^2$ and $a_\bullet \in S^{\omega,\epsilon_0}_{\rm twi}(\cE)$
\begin{align}\label{eq:N<N-twi}
     \cN^{\omega,\cE,\epsilon_0}_{\alpha,\beta}(a^c_\bullet)\leq C_{\beta}\sum_{\beta'+\gamma=\beta} \cN^{\omega,\cE,\epsilon_0}_{\alpha,\beta',\gamma}(a_\bullet).
\end{align}
 The symbol $a_\epsilon$ can therefore be quantized for all $0 < \epsilon \le \epsilon_0$ using the twisted Weyl quantization rule ${\rm Op}_{\epsilon}^c$ defined in \eqref{eq:TWQ}.

Let us point out that 
\begin{align}\label{eq:derivative-tau}
    a_\bullet \in S^{\omega,\epsilon_0}_{\rm twi}(\cE) \Rightarrow \partial_k^\alpha\partial_X^\beta{\rm ad}_{\partial_x}^\gamma a_\bullet \in S^{\omega,\epsilon_0}_{\rm twi}(\cE),\qquad \forall\;\alpha,\beta,\gamma\in \N^2.
\end{align}
This is due to that for any $\tau$-equivariant symbol $a_\epsilon$, the symbols $\partial_k a_\epsilon$, $\partial_X a_\epsilon$ and ${\rm ad}_{\partial_x}a_\epsilon$ are also $\tau$-equivariant in the variable $k$ and $J\mathbb{L}$-periodic in the variable $X$.

\medskip

Thanks to \eqref{eq:Hdeps=Opceps}, the study of  $\Trb(f(H_{d,\theta}))$ boils down to 
\begin{description}
\item[(i)] the study of the operator-valued symbol $h_{d,0}(k,X)$, which can be seen as an element of $S^{\omega}_{\rm twi}(\cE)$ with $\omega(k,X):=1+|k|^2$ and $\mathcal{E}:=\cB(H^2_{\rm per};L^2_{\rm per})$; 
\item[(ii)]  the study of the twisted Weyl quantization operator ${\rm Op}^c_\epsilon$, acting on symbols in $S^{\omega,\epsilon_0}_{\rm twi}(\cE)$ (see Section~\ref{sec:TWC}).
\end{description}

\medskip

Another essential tool for working out a semiclassical expansions of $\Trb[f(H_{d,\theta})]$ and $\Trb[f(H_{d,K,\theta}^{\rm eff})]$ is functional calculus for pseudodifferential operators based on Helffer-Sj\"ostrand's formula (see e.g. \cite{dimassi1993developpements,Helffer1989EquationDS}):  for any self-adjoint operator $A$ on a Hilbert space~$\cX$ and any function $f \in C^\infty_{\rm c}(\mathbb{R;\C})$, it holds
\begin{equation}\label{eq:HF_formula}
f(A)=-\frac{1}{\pi}\int_{\mathbb{C}}\overline{\partial}\widetilde{f}(\zeta)(\zeta-A)^{-1}dL(\zeta),
\end{equation}
where $dL(\zeta)$ denotes the Lebesgue measure on $\mathbb{C}$ and where $\widetilde{f}$ is an almost analytic extension of $f$, that is a function  $\widetilde{f}\in C^\infty_{\rm c}(\mathbb{C};\C)$,  satisfying 
\begin{enumerate}
    \item $\Supp(\widetilde{f})$ is a complex neighborhood of $\Supp(f)$;
    \item $\widetilde{f}(\zeta)=f(\zeta)$ for any $\zeta\in \mathbb{R}$;
    \item $|\overline{\partial}\widetilde{f}(\zeta)|=\mathcal{O}(|\Im \zeta|^\infty)$, i.e., for any $n\in \mathbb{N}$, $|\overline{\partial}\widetilde{f}(\zeta)|=\mathcal{O}(|\Im \zeta|^n)$ when $\Im \zeta \to 0$.
\end{enumerate}

\begin{remark}
There are several ways to define such almost analytic extensions. One way uses the Fourier transform $\mathcal{F}(f)$ and two cutoff functions $\chi\in C^\infty_{\rm c}(-1,1)$, $\zeta\in C^\infty_{\rm c}(U)$ ( $U$ is a suitable real neighborhood of $\Supp(f)$) with $\chi=1$ near zero and $\varphi=1$ in $\Supp(f)$:
\begin{align}
    \widetilde{f}(x+iy):=\frac{1}{(2\pi)^{1/2}}\varphi(x)\chi(y)\int_\R e^{i\xi(x+iy)}\chi(y\xi) \, \mathcal{F}(f)(\xi) \, d\xi.
\end{align}
\end{remark}

Lastly, we recall the definition of the Poisson bracket $\{a,b\}$ of two operator-valued symbols $a,b$: 
\begin{align}\label{eq:poisson}
        \{a,b\}:=\nabla_X a \cdot\nabla_k b-\nabla_k a \cdot \nabla_X b.
\end{align}
We also define
\begin{align*}
    \{a,b\}_{2}:&= D^2_{kk}a : D^2_{XX}b + D^2_{XX}a : D^2_{kk}b-2 D^2_{kX} : a D^2_{kX}b \\
    &=\sum_{1\leq i,j\leq 2}\left(\partial_{k_i}\partial_{k_j}a\partial_{X_i}\partial_{X_j}b +\partial_{X_i}\partial_{X_j}a\partial_{k_i}\partial_{k_j}b-2\partial_{k_i}\partial_{X_j}a\partial_{X_i}\partial_{k_j}b\right)
\end{align*}
In the case when $a \in C^\infty(\R^2 \times \R^2;{\mathcal B}(L^2_{\rm per}))$ and $b \in C^\infty(\R^2 \times \R^2;{\mathcal B}(H^2_{\rm per};L^2_{\rm per}))$, we have $\{a,b\},\; \{a,b\}_2 \in C^\infty(\R^2 \times \R^2;{\mathcal B}(H^2_{\rm per};L^2_{\rm per}))$. This property will be used below. 

\medskip

We are now in position to state our main results. The first one provides asymptotic expansions in $\epsilon:=\sin(\theta/2)$ and $\eta:=\tan(\theta/2)$ of  $\Trb(f(H_{d,\theta}))$ for any $f \in C^\infty_{\rm c}(\R;\R)$ with support in $(-\infty,0)$, thus providing information (through~\eqref{eq:DoS_H_2}) on the DoS of $H_{d,\theta}$ on the whole negative energy window.

\begin{theorem} \label{DoS-TBG}
Under Assumptions~\ref{def:MGpotential} and \ref{def:inter}, for all $f \in C^\infty_{\rm c}(\R;\R_+)$ with compact support included in $(-\infty,0)$, there exist operator-valued symbols $f_{d,j} \in S^{1}_{\rm twi}(\mathfrak{S}_1(L^2_{\rm per}))$ and $f_{d,j}' \in S^{1}_{\tau,\rm per}(\mathfrak{S}_1(L^2_{\rm per}))$, $j\in \N$, depending linearly on $f$ and such that the following asymptotic expansions hold true for any $n \in \N$:
\begin{align}
  \Trb[f(H_{d,\theta})]&=\sum_{j=0}^n \frac{\epsilon^{2j}}{(2\pi)^2} \fint_{J\Omega}\int_{\Omega^*}\Tr_{L^2_{\rm per}}[f_{d,2j}(k,X)] \, dk \, dX+\mathcal{O}(\epsilon^{2n+2}), \label{eq:main2} \\
  \Trb[f(H_{d,\theta})]&=\sum_{j=0}^n \frac{\eta^{2j}}{(2\pi)^2} \fint_{J\Omega}\int_{\Omega^*}\Tr_{L^2_{\rm per}}[f_{d,2j}'(k,X)] \, dk \, dX+\mathcal{O}(\eta^{2n+2}), \label{eq:main2'}
\end{align}
with $\epsilon:=\sin(\theta/2)$ and $\eta:=\tan(\theta/2)$.
In particular, we can take
\begin{align*}
    f_{d,0}(k,X)&=f_{d,0}'(k,X)=f(h_{d,0}(k,X)),
\end{align*}
and
\begin{align}
 \MoveEqLeft    f_{d,2}(k,X)=-\frac{1}{\pi}\int_\C \overline{\partial}\widetilde{f}(\zeta)  \Big[-\frac{1}{4}(\zeta-h_{d,0})^{-1}\{(\zeta-h_{d,0}),(\zeta-h_{d,0})^{-1}\}^2\notag\\
     &+\frac{1}{4}\{(\zeta-h_{d,0})^{-1},\{(\zeta-h_{d,0}),(\zeta-h_{d,0})^{-1}\}\}\notag\\
     &+\frac{1}{8}(\zeta-h_{d,0})^{-1}\{(\zeta-h_{d,0}),(\zeta-h_{d,0})^{-1}\}_2\notag\\
        &+\frac{i}{4}(\zeta-h_{d,0})^{-1}\Big(\partial_k h_{d,0}{\rm ad}_{\partial_x} (\zeta-h_{d,0})^{-1}-{\rm ad}_{\partial_x}h_{d,0} \partial_k (\zeta-h_{d,0})^{-1}\Big)\Big](k,X) \, dL(\zeta), \label{eq:f_d2}
\end{align}
\begin{align}
 \MoveEqLeft    f_{d,2}'(k,X)=-\frac{1}{\pi}\int_\C \overline{\partial}\widetilde{f}(\zeta)  \Big[-\frac{1}{4}(\zeta-h_{d,0})^{-1}\{(\zeta-h_{d,0}),(\zeta-h_{d,0})^{-1}\}^2\notag\\
     &+\frac{1}{4}\{(\zeta-h_{d,0})^{-1},\{(\zeta-h_{d,0}),(\zeta-h_{d,0})^{-1}\}\}\notag\\
     &+\frac{1}{8}(\zeta-h_{d,0})^{-1}\{(\zeta-h_{d,0}),(\zeta-h_{d,0})^{-1}\}_2\notag\\
     &-\frac{1}{2}(\zeta-h_{d,0})^{-1}(-i\nabla_x+k)^2(\zeta-h_{d,0})^{-1}-(\zeta-h_{d,0})^{-1} \Big](k,X) \, dL(\zeta), \label{eq:f_d2'}
\end{align}
 where $\widetilde{f}$ is any almost analytic extension of $f$.
\end{theorem}

Note that the odd terms of the expansions vanish because $H_{d,-\theta}$ and $H_{d,\theta}$ are unitary equivalent by symmetry with respect to e.g. the vertical plane $Ox_1z$.

\medskip

Several remarks are in orders:
\begin{enumerate}
    \item the two expansions \eqref{eq:main2} and \eqref{eq:main2'} are somehow redundant since, rewriting them as 
    \begin{align*}
        \Trb[f(H_{d,\theta})] &  \sim \sum_{j=0}^{+\infty} a_{2j} \epsilon^{2j} \quad \mbox{and} \quad
        \Trb[f(H_{d,\theta})]   \sim \sum_{j=0}^{+\infty} a'_{2j} \eta^{2j} = \sum_{j=0}^{+\infty} a'_{2j} \frac{\epsilon^{2j}}{(1-\epsilon^2)^j},
    \end{align*}
    we have 
    $$
    a'_0=a_0, \quad a_2=a'_2, \quad a_4=a'_4+a'_2, \quad a_6=a'_6+2a'_4+a'_2, \cdots;
    $$
     however, the coefficients $a_{2j}$ and $a_{2j}'$ are obtained by computing integrals of traces of operator-valued symbols $f_{d,j}$ and $f_{d,j}'$ which are substantially different (compare e.g. the last lines of \eqref{eq:f_d2}-\eqref{eq:f_d2'}). Notably, $f_{d,j}$ and $f_{d,j}'$ leave in different classes of symbols, namely $S^{1}_{\rm twi}(\mathfrak{S}_1(L^2_{\rm per}))$ and $S^{1}_{\tau,\rm per}(\mathfrak{S}_1(L^2_{\rm per}))$ respectively. It is not clear {\it a priori} which of these two formalisms lends itself best to numerical approximations. This point will be investigated in a future work; 
    
    \item it is not immediate to see from the expressions of $f_{d,2}$ and $f_{d,2}'$ that we actually have
    \begin{equation}\label{eq:a2=a2'}
    a_2=a_2'. 
    \end{equation}
    This equality results from an integration by part, as shown in Section~\ref{sec:check};
    
    \item from each of the two expansions \eqref{eq:main2} and \eqref{eq:main2'}, asymptotic expansions of $\Trb[f(H_{d,\theta})]$ in powers of $\theta$ can be easily obtained by expanding $\epsilon=\sin(\theta/2)$ and $\eta=\tan(\theta/2)$ in powers of $\theta$. Of course, these two expansions agree;
    
    \item our proof of \eqref{eq:main2} is slightly longer than our proof of \eqref{eq:main2'}, since the former makes use of an exotic class of operator-valued symbols and a twisted version of Weyl calculus, while the latter just uses standard classes of operator-valued symbols and standard Weyl calculus. On the other hand, the rescaling argument on which the proof of \eqref{eq:main2'} is based is specific to the considered Hamiltonian and does not apply to more general Hamiltonians, for instance those of the form
    $$
    -\frac 12 \Delta_{x,z} + V_{{\mathfrak d}_\epsilon(\epsilon x)}\left( x-J \Xi_\epsilon(\epsilon x),\Phi_\epsilon(\epsilon x), z \right),
    $$
    where ${\mathfrak d}_\epsilon:\R^2 \to \R_+^*$, $\Phi_\epsilon: \R^2 \to \R^2$, and $\Xi_\epsilon:\R^2 \to \R^2$ are smooth $J\mathbb L$-periodic functions having asymptotic expansions in $\epsilon$. Such Hamiltonians, whose density-of-states can be studied using our $\epsilon$-expansion method, could possibly be used to take atomic relaxation into account.
\end{enumerate}

\begin{theorem}\label{DoS-BM}
Under Assumptions~\ref{def:MGpotential} and \ref{def:inter}, for all $f \in C^\infty_{\rm c}(\R;\R_+)$, there exist matrix-valued symbols $f_{d,K,j} \in S^{1}_\tau(\C^{4\times 4}),\;j\in \N$, depending linearly on $f$ and such that the following asymptotic expansion holds true for any $n \in \N$: 
\begin{align}\label{eq:DoS-BM}
\Trb[f(H_{d,K,\theta}^{\rm eff})]&=\sum_{j=0}^n \frac{\epsilon^j}{(2\pi)^2} \fint_{J\Omega}\int_{\R^2}\Tr_{\C^4}[f_{d,K,j}(\kappa,X)] \, d\kappa \, dX+\mathcal{O}(\epsilon^{n+1}).
\end{align}
In particular, we can take 
\begin{align*}
    f_{d,K,0}(\kappa,X)=f( \mathfrak{h}_{d,K,0}^{\rm eff}(\kappa,X)),
\end{align*}
\begin{align*}
    f_{d,K,1}(\kappa,X)&:=- \frac{1}{\pi}\int_{\C} \overline{\partial}\widetilde{f}(\zeta ) \bigg[ - \frac i2 \{(\zeta- \mathfrak{h}_{d,K,0}^{\rm eff})^{-1},(\zeta- \mathfrak{h}_{d,K,0}^{\rm eff})\}(\kappa,X)(\zeta- \mathfrak{h}_{d,K,0}^{\rm eff}(\kappa,X))^{-1} \\
    &\qquad \qquad + (\zeta- \mathfrak{h}_{d,K,0}^{\rm eff}(\kappa,X))^{-1}T_{0,1}(\kappa)(\zeta- \mathfrak{h}_{d,K,0}^{\rm eff}(\kappa,X))^{-1} \bigg] \, dL(\zeta)
\end{align*}
and
\begin{align*}
  \MoveEqLeft  f_{d,K,2}(\kappa,X)=-\frac{1}{\pi}\int_{\C} \overline{\partial}\widetilde{f}(\zeta )  \Big[-\frac{1}{4}(\zeta-\mathfrak{h}_{d,K,0}^{\rm eff})^{-1}\{(\zeta-\mathfrak{h}_{d,K,0}^{\rm eff}),(\zeta-\mathfrak{h}_{d,K,0}^{\rm eff})^{-1}\}^2\\
     &+\frac{1}{4}\{(\zeta-\mathfrak{h}_{d,K,0}^{\rm eff})^{-1},\{(\zeta-\mathfrak{h}_{d,K,0}^{\rm eff}),(\zeta-\mathfrak{h}_{d,K,0}^{\rm eff})^{-1}\}\}\\
     &+\frac{1}{8}(\zeta-\mathfrak{h}_{d,K,0}^{\rm eff})^{-1}\{(\zeta-\mathfrak{h}_{d,K,0}^{\rm eff}),(\zeta-\mathfrak{h}_{d,K,0}^{\rm eff})^{-1}\}_2\\
     & - \frac 12(\zeta-\mathfrak{h}_{d,K,0}^{\rm eff} )^{-1}T_{0}(\zeta-\mathfrak{h}_{d,K,0}^{\rm eff} )^{-1}\\
     & +(\zeta-\mathfrak{h}_{d,K,0}^{\rm eff} )^{-1}\Big[T_{0,1}(\zeta-\mathfrak{h}_{d,K,0}^{\rm eff} )^{-1}\Big]^2\\
     &+ \frac{i}{2}\{(\zeta-\mathfrak{h}_{d,K,0}^{\rm eff})^{-1}, T_{0,1}\}(\zeta-\mathfrak{h}_{d,K,0}^{\rm eff})^{-1}\\
     &- \frac{i}{2}\{(\zeta-\mathfrak{h}_{d,K,0}^{\rm eff})^{-1}T_{0,1}(\zeta-\mathfrak{h}_{d,K,0}^{\rm eff})^{-1},\zeta-\mathfrak{h}_{d,K,0}^{\rm eff}\}(\zeta-\mathfrak{h}_{d,K,0}^{\rm eff})^{-1}\\
     &-\frac{i}{2}\{(\zeta-\mathfrak{h}_{d,K,0}^{\rm eff})^{-1},\zeta-\mathfrak{h}_{d,K,0}^{\rm eff}\}(\zeta-\mathfrak{h}_{d,K,0}^{\rm eff})^{-1}T_{0,1}(\zeta-\mathfrak{h}_{d,K,0}^{\rm eff})^{-1}\Big](\kappa,X) dL(\zeta).
\end{align*}
where $\widetilde{f}$ is any almost analytic extension of $f$, 
\begin{align} \label{eq:mat_h_eff}
    \mathfrak{h}_{d,K,0}^{\rm eff}(\kappa,X):&=T_0( \kappa)+\mathcal{V}_{d,K}(X), \\
    T_{0,1}(\kappa):&=\begin{pmatrix}
   v_{\rm F} (-\sigma_2,\sigma_1) \cdot \kappa & 0\\
   0&  -v_{\rm F} (-\sigma_2,\sigma_1) \cdot \kappa \end{pmatrix}, \notag
\end{align}
where the matrix-valued functions $T_0(\bullet)$ and ${\mathcal V}_{d,K}(\bullet)$ are given by  \eqref{eq:EMS_model}-\eqref{eq:T-V-1}. 
\end{theorem}

In the above result, we do not assume that the function $\mathcal V_{d,K}(X)$ has any specific point symmetry. In the special case of the BM model, the odd terms in the expansion vanish due to symmetries.

\section{Analysis of \texorpdfstring{$h_{d,0}(k,X)$}{}} \label{sec:hd0}

In order to prove Theorems \ref{DoS-TBG}, we need to establish some useful properties of $h_{d,0}(k,X)$: we study the properties of $h_{d,0}(k,X)$ and $g(h_{d,0}(k,X))$ for specific functions $g:\R \to \C$, seen as operator-valued symbols. These results will be used in the proof of Theorem~\ref{DoS-TBG}.

\subsection{Spectral properties of the operator \texorpdfstring{$h_{d,0}(k,X)$}{}}
As a prelude, we first prove some basic spectral properties of the operator $h_{d,0}(k,X)$.
\begin{lemma}\label{lem:spec}
Under Assumptions \ref{def:MGpotential} and \ref{def:inter}, 
\begin{enumerate}
\item for all $k,X\in \mathbb{R}^2$, $h_{d,0}(k,X)$ is a self-adjoint operator on $L^2_{\rm per}$ with domain $H^2_{\rm per}$ and it holds
\begin{align}\label{eq:ess}
    \sigma_{\rm ess}(h_{d,0}(k,X))=[0,+\infty);
\end{align}
\item there exists $M \in \R_+$ such that
\begin{align}\label{eq:specbound}
    \inf_{\substack{k,X\in \mathbb{R}^2}} \sigma(h_{d,0}(k,X))\geq -M;
\end{align}
\item for any $E < 0$, there is a constant $m_E \in \N$ such that
\begin{align}\label{eq:rankbound}
    \sup_{\substack{ k,X\in \mathbb{R}^2}}\rank(\1_{(-\infty,E]}(h_{d,0}(k,X)))\leq m_E.
\end{align}
\end{enumerate}
\end{lemma}
\begin{proof}
Let $k,X\in\mathbb{R}^2$. From Kato-Rellich theorem, the operator $h_{d,0}(k,X)$ is self-adjoint on $L^2_{\rm per}$ with domain $H^2_{\rm per}$. Let $\xi\in\mathbb{C}\setminus\mathbb{R}$. It follows from Assumptions \ref{def:MGpotential} and \ref{def:inter}  and the Kato-Seiler-Simon inequality \cite{simon1978methods} that 
\begin{align*}
\MoveEqLeft    (h_{d,0}-\xi)^{-1}- \left( \frac{1}{2}(-i\nabla_x+k)^2-\frac{1}{2}\partial_z^2-\xi \right)^{-1}\\
&=-(h_{d,0}-\xi)^{-1}V_d
\left( \frac{1}{2}(-i\nabla_x+k)^2-\frac{1}{2}\partial_z^2-\xi \right)^{-1}
\end{align*}
is a Hilbert-Schmidt operator on $L^2_{\rm per}$. We deduce from Weyl's essential spectrum theorem (see e.g. \cite[Theorem XIII.14 and Corollary 1]{simon1978methods}) that 
$$
\sigma_{\rm ess}(h_{d,0}(k,X))=\sigma_{\rm ess}\left(\frac{1}{2}(-i\nabla_x+k)^2-\frac{1}{2}\partial_z^2\right)=[0,+\infty).
$$
The second assertion is an immediate consequence of the fact that 
\[
\sup_{x,X\in \mathbb{R}^2,z\in\mathbb{R}}|V_d(x,X,z)|=: M < \infty.
\]
We turn now to the proof of the third assertion. Let $E < 0$. As $h_{d,0}(k,X)$ is $\tau$-equivariant w.r.t. $k$ and $J\mathbb{L}$-periodic w.r.t. $X$, it suffices to prove that
\begin{align}\label{eq:rankbound_2}
    \sup_{\substack{ (k,X)\in \overline{\Omega^*} \times J\overline{\Omega}}}\rank(\1_{(-\infty,E]}(h_{d,0}(k,X)))=:m_E < \infty.
\end{align}
For any bounded-from-below self-adjoint operator $A$ on a Hilbert space $\cH$ with domain $D(A)$, and any $j \in \N^*$, we denote by 
\begin{align}\label{eq:courant}
    \sigma_j(A):= \inf_{V_j \in {\mathcal V}_j} \sup_{v \in V_j \, | \, \|v\|_{\cH}=1} \left\langle v,Av\right\rangle_\cH,
\end{align}
where ${\mathcal V}_j$ is the set of vector subspaces of $D(A)$ of dimension $j$. Recall that by Courant-Fischer min-max formula, $\sigma_j(A)$ is equal to the $j^{\rm th}$ lowest eigenvalue of $A$ (counting multiplicities) if $A$ has at least $j$ eigenvalues below $\inf(\sigma_{\rm ess}(A))$, and to $\inf(\sigma_{\rm ess}(A))$ otherwise. Since the functions 
\begin{equation}\label{eq:mu_j}
\mu_j : \R^2 \times \R^2 \ni (k,X) \mapsto \mu_j(k,X):=\sigma_j(h_{d,0}(k,X)) \in \R
\end{equation}
form a nondecreasing sequence of continuous functions, the sets
$$
C_j:=\mu_j^{-1}((-\infty,E])\bigcap (\overline{\Omega^*} \times J\overline{\Omega}) = \mu_j^{-1}([-M,E]) \bigcap (\overline{\Omega^*} \times J\overline{\Omega})
$$
are compact and such that 
\begin{align*}
    C_1\supset C_2\supset\cdots.
\end{align*}
If each $C_j$ was not empty, then so would be $\bigcap_{j \in \N^*}C_j$ by virtue of the Cantor intersection theorem.  For $(k_*,X_*) \in \bigcap_{j \in \N^*}C_j$, we would have $-M \le \sigma_j(h_{d,0}(k_*,X_*)) \le E$ for all $j \in \N^*$, which contradicts the fact that $\sigma_{\rm ess}(h_{d,0}(k_*,X_*)) \bigcap [-M,E]= \emptyset$. The inequality \eqref{eq:rankbound} is therefore satisfied with $m_E:=\max\{j \in \N^* \, | \, C_j \neq \emptyset \}$.
\end{proof}

\subsection{Properties of the operator-valued symbol \texorpdfstring{$h_{d,0}(k,X)$ }{}}\label{sec:hd0-symbol}

Let $f \in C_{\rm c}^\infty(\R;\R_+)$ with support included in $(-\infty,0)$, and let $\widetilde{f}$ be an almost analytic extension satisfying $\Supp(\widetilde{f})\subset (-\infty,0)$. We set
\[
\delta:=\frac{1}{4}\dist\left(0,\Supp(\widetilde{f})\right)>0,
\]
consider $\chi_{f} \in C^{\infty}(\R;\R)$ satisfying 
\begin{equation}\label{eq:chi_f}
0\leq \chi_{f}\leq 1 \mbox{ on } \R, \quad \chi_{f}=0 \textrm{ on } (-\infty,-2\delta],
\quad \textrm{and}\quad \chi_{f}=1 \textrm{ on } [-\delta,+\infty),
\end{equation}
and set
\begin{align}\label{eq:h-t}
    {h}_{d,0,f}(k,X):=\chi_{f}\left(h_{d,0}(k,X)\right) \, h_{d,0}(k,X).
\end{align}

We are going to show that $h_{d,0},{h}_{d,0,f}$ and $h_{d,0}-{h}_{d,0,f}$ are symbols in $S^{\rho}_{\twi}(\mathcal{X})$ for suitable order functions $\rho$ and linear operator spaces $\mathcal{X}$.

\begin{lemma}\label{lem:4.2}
Let $\omega(k,X):=1+|k|^2$. We have
\begin{align*} 
    h_{d,0},h_{d,0,f}\in S^{\omega}_{\twi}(
\mathcal{B}(H^2_{\rm per},L^2_{\rm per}))
\end{align*}
and
\begin{align} \label{eq:S-h-h-tilde}
    (h_{d,0}-h_{d,0,f})\in S^{\omega}_{\twi}(
\mathcal{B}(L^2_{\rm per},H^2_{\rm per}))\bigcap S^{1}_{\twi}(
\mathcal{B}(L^2_{\rm per}))\bigcap S^{\omega}_{\twi}(
\mathcal{B}(H^{-2}_{\rm per},L^2_{\rm per})).
\end{align}
\end{lemma}
\begin{proof}
First, it is easy to see that
\begin{align*}
    \MoveEqLeft\sup_{k,X\in \R^2}\omega^{-1}(k,X)\| h_{d,0}(k,X)\|_{\cB(H^2_{\rm per},L^2_{\rm per})}<\infty.
\end{align*}
Next, we have
\eqref{eq:H-epsil},
\begin{align}\label{eq:partial-k-g}
 \partial_k h_{d,0}(k,X)= -i\nabla_x+k,\quad \partial_Xh_{d,0}(k,X)=\partial_X V_d(\cdot,X,\cdot),
\end{align}
and
\begin{align}\label{eq:partial-X-g}
\px h_{d,0}(k,X)=\partial_x V_d(\cdot,X,\cdot).
\end{align}
We deduce from Assumption \ref{def:MGpotential} that for all $\alpha,\beta,\gamma\in \N^2$, 
\begin{align*}
    \sup_{k,X\in \R^2}\omega^{-1}(k,X)\|\partial_k^\alpha\partial_X^\beta\px^\gamma h_{d,0}(k,X)\|_{\cB(H^2_{\rm per}, L^2_{\rm per})} < \infty.
\end{align*}
As $h_{d,0}$ is $\tau$-equivariant w.r.t. $k$ and $J\mathbb{L}$-periodic w.r.t. $X$, we finally have
\begin{align*}
     h_{d,0}\in S^{\omega}_{\twi}(\cB(H^2_{\rm per}, L^2_{\rm per})).
\end{align*}

\medskip

To show that $h_{d,0,f}\in S^{\omega}_{\twi}(
\mathcal{B}(H^2_{\rm per},L^2_{\rm per}))$, it suffices to prove 
\begin{align*}
    h_{d,0}-
h_{d,0,f}\in S^{\omega}_{\twi}(
\mathcal{B}(H^2_{\rm per},L^2_{\rm per})).
\end{align*} 
As $S^{\omega}_{\twi}(
\mathcal{B}(H^2_{\rm per},L^2_{\rm per}))\subset S^{\omega}_{\twi}(
\mathcal{B}(H^{-2}_{\rm per},L^2_{\rm per}))$, we only need to prove \eqref{eq:S-h-h-tilde}. Let $M$ be the constant in \eqref{eq:specbound}, and $\xi_{f}\in C^\infty(\R;\R)$ be such that 
$$
\mbox{Supp}(\widetilde \xi_f) \subset (-M-2,-\delta), \quad \mbox{and}  \quad \xi_{f}(t) =t(1-\chi_{f}(t)) \mbox{ for } t \in [-M,+\infty).
$$
It follows from Lemma \ref{lem:spec} and \eqref{eq:chi_f} that 
\[
\xi_f(h_{d,0})=(1-\chi_{f}(h_{d,0}))h_{d,0}=  h_{d,0}-
h_{d,0,f}.
\]
From this and Eqn. \eqref{eq:HF_formula}, we infer
\begin{align}\label{eq:h-htilde}
    h_{d,0}-h_{d,0,f}=-\frac{1}{\pi}\int_{\mathbb{C}}\overline{\partial}\widetilde\xi_f(\zeta)(\zeta-h_{d,0})^{-1}dL(\zeta),
\end{align}
where $\widetilde\xi_f$ is an almost analytic extension of $\xi_f$ such that $\widetilde\xi_f\in C^\infty_{\rm c}(\mathbb{C})$ and $|\overline{\partial}\widetilde\xi_f(\zeta)|=\mathcal{O}(|\Im \zeta|^\infty)$. 

Decomposing each $k \in \R^2$ as $k=[k]+G$ with $[k] \in \Omega^*$ and $G \in {\mathbb L}^*$, noticing that for any $r,s\in \mathbb{R}^2$ and $t \in \R$, we have
\begin{align}\label{eq:triangle}
    1+|r+s|^2+t^2\leq 1+2|r|^2+2|s|^2+t^2\leq  2(1+|s|^2+t^2)(1+|r|^2),
\end{align}
and using the $\tau$-equivariance of $h_{d,0}^c$ and the fact that 
$$
\tau_G^{-1}(1-\Delta_x-\partial^2_z)\tau_G= 1+|-i\nabla_x +G|^2-\partial^2_z,
$$
we obtain that for all $\zeta \in {\rm Supp}(\widetilde f)$ and $a,b\in \{0,1\}$ such that $0\leq a+b\leq 1$, 
\begin{align}\label{eq:29}
   \MoveEqLeft \|(1-\Delta_x-\partial^2_z)^a (\zeta-h_{d,0}(k,X))^{-1}(1-\Delta_x-\partial^2_z)^{b}\|_{\mathcal{B}(L^2_{\rm per})}\notag\\
    &=\|(1-\Delta_x-\partial^2_z)^a (\zeta-h_{d,0}([k]+G,X))^{-1}(1-\Delta_x-\partial^2_z)^{b}\|_{\mathcal{B}(L^2_{\rm per})}\notag\\
    &=\|(1-\Delta_x-\partial^2_z)^a\tau_G (\zeta-h_{d,0}([k],X))^{-1}\tau_G^{-1}(1-\Delta_x-\partial^2_z)^{b}\|_{\mathcal{B}(L^2_{\rm per})}\notag\\
    &\leq 4(1+|G|^2)^{a+b}\|(1-\Delta_x-\partial^2_z)^a (\zeta-h_{d,0}([k],X))^{-1}(1-\Delta_x-\partial^2_z)^{b}\|_{\mathcal{B}(L^2_{\rm per})}\notag\\
    &\leq C(1+|G|^2)^{a+b}|\Im \zeta|^{-1}\leq C'(1+|k|^2)^{a+b}|\Im \zeta|^{-1} = C' \omega(k,X)^{a+b} |\Im \zeta|^{-1}.
\end{align}
The first inequality holds because of \eqref{eq:triangle} while the second holds due to the facts that $0\leq a+b\leq 1$ and that there exist two constants $C_1,C_2 \in \R_+$ such that for all $(k,X) \in \Omega^* \times \R^2$ and $u\in H^2_{\rm per}$  
\begin{align}\label{eq:delta-h}
    \|(1-\Delta_x-\partial^2_z)u\|_{L^2_{\rm per}}\leq C_1\|h_{d,0}([k],X)u\|_{L^2_{\rm per}}+C_2\|u\|_{L^2_{\rm per}}.
\end{align}
Since $h_{d,0}(k,X)$ is self-adjoint, we have 
\begin{align*}
  \MoveEqLeft  \sup_{k,X\in\mathbb{R}^2} \omega^{-1} \| (\zeta-h_{d,0})^{-1}\|_{\mathcal{B}(L^2_{\rm per})} \le 
\sup_{k,X\in\mathbb{R}^2}\| (\zeta-h_{d,0})^{-1}\|_{\mathcal{B}(L^2_{\rm per})} = |\Im \zeta|^{-1}, 
\end{align*}
and applying \eqref{eq:29} first with $(a,b)=(1,0)$ then with $(a,b)=(0,1)$, gives
\begin{align*}
\MoveEqLeft \sup_{k,X\in\mathbb{R}^2}\omega^{-1}\| (\zeta-h_{d,0})^{-1}\|_{\mathcal{B}(L^2_{\rm per},H^2_{\rm per})}\\
 &= \sup_{k,X\in\mathbb{R}^2}\omega^{-1}\|(1-\Delta_x-\partial^2_z) (\zeta-h_{d,0})^{-1}\|_{\mathcal{B}(L^2_{\rm per})}\leq C|\Im \zeta|^{-1},
\end{align*}
and
 \begin{align*}
 \MoveEqLeft   \sup_{k,X\in\mathbb{R}^2}\omega^{-1}\| (\zeta-h_{d,0})^{-1}\|_{\mathcal{B}(H^{-2}_{\rm per},L^2_{\rm per})}\\
 &= \sup_{k,X\in\mathbb{R}^2}\omega^{-1}\| (\zeta-h_{d,0})^{-1}(1-\Delta_x-\partial^2_z)\|_{\mathcal{B}(L^2_{\rm per})}\leq C|\Im \zeta|^{-1}.
\end{align*}

Concerning the derivatives $\partial_k,\partial_X$ and $\px$ of $h_{d,0}$, there exists $C\in \R_+$ such that for all $k,X \in \R^2$ and $\zeta \in {\rm Supp}(\widetilde f)$,
\begin{align*}
 \MoveEqLeft   \|(\partial_X h_{d,0}) (\zeta- h^c_{d,0})^{-1} \|_{\mathcal{B}(L^2_{\rm per})}= \|\big(\partial_X V_{d}(\cdot,X,\cdot)\big) (\zeta- h_{d,0})^{-1} \|_{\mathcal{B}(L^2_{\rm per})} \leq C|\Im \zeta|^{-1}
\end{align*}
\begin{align*}
 \MoveEqLeft   \|\px h_{d,0} (\zeta- h_{d,0})^{-1} \|_{\mathcal{B}(L^2_{\rm per})}= \|\big(\partial_x  V_{d}(\cdot,X,\cdot)\big) (\zeta- h_{d,0})^{-1} \|_{\mathcal{B}(L^2_{\rm per})} \leq C|\Im \zeta|^{-1}
\end{align*}
and
\begin{align*}
&\|(\partial_k h_{d,0}) (\zeta- h_{d,0})^{-1} \|_{\mathcal{B}(L^2_{\rm per})} \\
 &\quad=\left\| \frac{(-i\nabla_x+k)}{(-i\nabla_x+k)^2-\partial_z^2+1} \big( (-i\nabla_x+k)^2-\partial_z^2+1\big) (\zeta- h_{d,0})^{-1} \right\|_{\mathcal{B}(L^2_{\rm per})} \leq C|\Im \zeta|^{-1},
\end{align*}

Denoting by $\partial:=(\partial_k,\partial_X,\px)$, we conclude that there exists a constant $C \in \R_+$ such that
\begin{equation} \label{eq:bound_dhhm1}  \sup_{k,X \in \R^2}
\|(\partial h_{d,0}) (\zeta- h_{d,0})^{-1} \|_{\mathcal{B}(L^2_{\rm per})}\leq C|\Im \zeta|^{-1}, 
\end{equation}
and therefore
\begin{align*}
 \MoveEqLeft   \sup_{k,X\in\mathbb{R}^2}\| \partial(\zeta-h_{d,0})^{-1}\|_{\mathcal{B}(L^2_{\rm per})} \\
 &\leq \sup_{k,X\in\mathbb{R}^2}\| (\zeta- h_{d,0} )^{-1}(\partial h_{d,0} ) (\zeta- h_{d,0})^{-1} \|_{\mathcal{B}(L^2_{\rm per})} \leq C|\Im \zeta|^{-2}.
\end{align*}
Analogously,
\begin{align*}
     \sup_{k,X\in\mathbb{R}^2}\omega^{-1}\| \partial(\zeta-h_{d,0} )^{-1}\|_{\mathcal{B}(H^{-2}_{\rm per},L^2_{\rm per})}&\leq C|\Im \zeta|^{-2}
\end{align*}
and
\begin{align*}
   \sup_{k,X\in\mathbb{R}^2}\omega^{-1}\| \partial(\zeta-h_{d,0} )^{-1}\|_{\mathcal{B}(L^2_{\rm per},H^2_{\rm per})}&\leq C|\Im \zeta|^{-2}.
\end{align*}
Higher-order derivatives can be bounded by the same arguments, leading to
\begin{align}
\sup_{k,X\in\mathbb{R}^2}\| \partial_k^\alpha\partial_X^\beta\px^\gamma(\zeta-h_{d,0})^{-1}\|_{\mathcal{B}(L^2_{\rm per})}&\leq C_{\alpha,\beta,\gamma}|\Im \zeta|^{-|\alpha|-|\beta|-|\gamma|-1}, \label{eq:31a}\\
\sup_{k,X\in\mathbb{R}^2}\omega^{-1}\| \partial_k^\alpha\partial_X^\beta\px^\gamma(\zeta-h_{d,0})^{-1}\|_{\mathcal{B}(L^2_{\rm per},H^2_{\rm per})}&\leq C_{\alpha,\beta,\gamma}|\Im \zeta|^{-|\alpha|-|\beta|-|\gamma|-1}, \label{eq:31b}\\
   \sup_{k,X\in\mathbb{R}^2}\omega^{-1}\| \partial_k^\alpha\partial_X^\beta\px^\gamma(\zeta-h_{d,0})^{-1}\|_{\mathcal{B}(H^{-2}_{\rm per},L^2_{\rm per})}&\leq C_{\alpha,\beta,\gamma}|\Im \zeta|^{-|\alpha|-|\beta|-|\gamma|-1}.\label{eq:31c}
\end{align}
Together with \eqref{eq:h-htilde} and the fact that $h_{d,0}$ is $\tau$-equivariant w.r.t. $k$ and $J\mathbb{L}$-periodic w.r.t. $X$, the above three bounds lead to \eqref{eq:S-h-h-tilde}, which completes the proof.
\end{proof}

\begin{lemma}\label{cor:trace}
The symbol $(h_{d,0}-h_{d,0,f})$ belongs to the class $S_{\twi}^{1}(\mathfrak{S}_1(L^2_{\rm per}))$.
\end{lemma}

\begin{proof}
Let $a(k,X):=(h_{d,0}-h_{d,0,f})(k,X)$. By Lemma \ref{lem:4.2}, $a\in S^{1}_{\twi}(
\mathcal{B}(L^2_{\rm per}))$. It is easy to see that for any $k,X\in \mathbb{R}^d$,
\begin{align*}
  \|\partial_X^\alpha\partial_k^\beta \px^\gamma a\|_{\mathfrak{S}_1(L^2_{\rm per})}&\leq \|\partial_X^\alpha\partial_k^\beta\px^\gamma a\|_{\mathcal{B}(L^2_{\rm per})}\, \rank(\partial_X^\alpha\partial_k^\beta\px^\gamma a)\\
  &\leq C_{\alpha,\beta,\gamma}\, \rank(\partial_X^\alpha\partial_k^\beta\px^\gamma a).
\end{align*}
To end the proof, it suffices to show that $\sup_{k,X\in \R^2}\rank(\partial_X^\alpha\partial_k^\beta\px^\gamma a)<\infty$ for any $\alpha,\beta,\gamma\in \N^2$. 

For $\alpha=\beta=\gamma=0$, according to Lemma \ref{lem:spec}, we have
\begin{align*}
\sup_{k,X\in\mathbb{R}^2}\rank(a)
&\leq \sup_{k,X\in\mathbb{R}^2}\rank(\1_{(-\infty,\delta]}(h_{d,0}))\leq m_{-\delta} < \infty.
\end{align*}

Now we consider the derivatives w.r.t. $k,X\in \R^2$. For any $t\in [-1,1]\setminus\{0\}$, we have
\[
\sup_{k,X\in \R^2}\rank \left(\frac{a(k_1+t,k_2,X)-a(k_1,k_2,X)}{t} \right)\leq 2m_{-\delta}.
\]
Letting $t$ go to $0$, we get
\[
\sup_{k,X\in \R^2}\rank \left( \partial_{k_1}a\right) \leq 2m_{-\delta}.
\]
Thus,  for any $\alpha,\beta\in \mathbb{N}^d$, we have
\begin{align} 
  \sup_{k,X\in \R^2}\rank \left( \partial_k^\alpha\partial_X^\beta a \right)  \leq   2^{|\alpha|+|\beta|} \sup_{k,X\in \R^2}\rank \left( a \right)  \leq   2^{|\alpha|+|\beta|} m_{-\delta}. \label{eq:rank-a-kX}
\end{align}
Let $|\psi \rangle\langle \psi|$ be the orthonormal projector onto the vector space spanned by the function $\psi$. We decompose $a(k,X)$ as
\begin{align*}
    a(k,X)=\sum_{1}^{n(k,X)} \xi_f(\mu_j(k,X)) \left|\psi_j(k,X)\right>\left<\psi_j(k,X)\right|,
\end{align*}
where $n(k,X) \le m_{-\delta}$ is the rank of $\1_{(-\infty,-\delta]}(h_{d,0})$, $(\psi_j(k,X))_{j=1,\cdots,n(k,X)}$ form an $L^2_{\rm per}$-orthonormal basis of ${\rm Ran}(\1_{(-\infty,-\delta]}(h_{d,0}))$, and the notation $\mu_j$ was introduced in~\eqref{eq:mu_j}. By elementary elliptic regularity arguments, each $\psi_j(j,X)$ is in $C^\infty_{\rm b}(\R^3) \cap L^2_{\rm per}$. We thus have
\begin{align*}
    \px a(k,X)=\sum_{1}^{n(k,X)} \lambda_j(k,X)\Big(\left|\partial_x\psi_j\right>\left<\psi_j\right|+\left|\psi_j\right>\left<\partial_x\psi_j\right|\Big)(k,X),
\end{align*}
and therefore
\begin{align*}
    \sup_{k,X\in \R^2}\rank \left( \px^\gamma a \right)\leq  2^{|\gamma|}m_{-\delta}.
\end{align*}
From \eqref{eq:rank-a-kX} and proceeding as above, we have for all $\alpha,\beta,\gamma\in \mathbb{N}^d$,
\begin{align*}
 \MoveEqLeft   \sup_{k,X\in \R^2}\rank \left( \partial_k^\alpha\partial_X^\beta  \px^\gamma a \right) \leq   2^{|\alpha|+|\beta|+|\gamma|} m_{-\delta}.
\end{align*}
This implies that
\begin{align*}
  \sup_{k,X\in \R^2}\omega^{-1}\|\partial_X^\alpha\partial_k^\beta\px^\gamma a\|_{\mathfrak{S}_1(H)} < \infty.
\end{align*}
Hence $a\in S_{\twi}^{\omega}(\mathfrak{S}_1(L^2_{\rm per}))$.
\end{proof}

Analogous to Lemma \ref{lem:4.2}, we also have the following result. 
\begin{lemma}\label{lem:h-t-inversible}
 Let $\omega(k,X):=1+|k|^2$.  For any $\zeta\in \Supp(\widetilde{f})$, it holds
    \begin{align*}
        ( h_{d,0,f}-\zeta)^{-1}\in S^{\omega}_{\twi}(\mathcal{B}(L^2_{\rm per},H^2_{\rm per}))\bigcap S^{1}_{\twi}(\mathcal{B}(L^2_{\rm per})).
    \end{align*}
    Furthermore, for any $\alpha,\beta,\gamma\in \N^2$
    \begin{align}\label{eq:inverse-zeta1}
        \sup_{\zeta\in \Supp(\widetilde{f})}\cN_{\alpha,\beta,\gamma}^{\omega,\mathcal{B}(L^2_{\rm per},H^2_{\rm per})}( ( h_{d,0,f}-\zeta)^{-1})<\infty
    \end{align}
    and
    \begin{align}\label{eq:inverse-zeta2}
        \sup_{\zeta\in \Supp(\widetilde{f})}\cN_{\alpha,\beta,\gamma}^{1,\mathcal{B}(L^2_{\rm per})}( ( h_{d,0,f}-\zeta)^{-1})<\infty.
    \end{align}
\end{lemma}
\begin{proof} Let $\zeta \in \Supp(\widetilde{f})$.
 We first deduce from \eqref{eq:chi_f}-\eqref{eq:h-t} that
\begin{align}\label{h-elle}
    \sup_{\substack{  k,X\in\mathbb{R}^2\\\zeta\in \Supp(\widetilde{f})}}\|(h_{d,0,f}-\zeta)^{-1}\|_{\mathcal{B}(L^2_{\rm per})}\leq \frac{1}{2\delta}.
\end{align}
Proceeding as in the proof of \eqref{eq:29},  we obtain
\begin{align}
   \MoveEqLeft \|(1-\Delta_x-\partial^2_z) ( h_{d,0,f}(k,X)-\zeta)^{-1}\|_{\mathcal{B}(L^2_{\rm per})} \nonumber \\ & \leq 2 \omega(k,X) \|(1-\Delta_x-\partial^2_z)  (h_{d,0,f} ([k],X)-\zeta)^{-1}\|_{\mathcal{B}(L^2_{\rm per})}. \label{eq:bound_hd0f}
\end{align}
Using \eqref{eq:delta-h} and the fact that $(h_{d,0}- h_{d,0,f} )\in S^{1}_{\twi}(\mathcal{B}(L^2_{\rm per}))$ (see Lemma \ref{lem:4.2}), we get
\begin{align}\label{eq:delta-h'}
\forall u \in H^2_{\rm per}, \quad     \|(1-\Delta_x-\partial^2_z)u\|_{L^2_{\rm per}}&\leq C_1\|h_{d,0}([k],X)u\|_{L^2_{\rm per}}+C_2\|u\|_{L^2_{\rm per}}\notag\\
    &\leq C_1\|h_{d,0,f}([k],X)u\|_{L^2_{\rm per}}+C_3\|u\|_{L^2_{\rm per}},
\end{align}
where the constants $C_1$, $C_2$, $C_3$ are independent of $k$ and $X$.
Gathering together \eqref{h-elle}, \eqref{eq:bound_hd0f} and \eqref{eq:delta-h'}, we obtain 
\begin{align} \label{eq:bound_hdof-zeta-1}
\sup_{\substack{  k,X\in\mathbb{R}^2\\\zeta\in \Supp(\widetilde{f})}}\omega^{-1} \|(h_{d,0,f}-\zeta)^{-1}\|_{\mathcal{B}(L^2_{\rm per},H^2_{\rm per})} < \infty.
\end{align}
To bound the derivatives $\partial_k,\partial_X$ and commutator $\px$ of $( h_{d,0,f}(k,X)-\zeta)^{-1}$, we proceed as in Lemma \ref{lem:4.2}. Recall that $\partial=(\partial_k,\partial_X,\px)$. Adapting the proof of \eqref{eq:bound_dhhm1}, we get
$$
\sup_{\substack{  k,X\in\mathbb{R}^2\\\zeta\in \Supp(\widetilde{f})  k,X\in\mathbb{R}^2}} \| (\partial h_{d,0})(h_{d,0}-i)^{-1} \|_{\mathcal{B}(L^2_{\rm per})} <+ \infty.
$$
Together with the fact that $(h_{d,0}- h_{d,0,f} )\in S^{1}_{\twi}(\mathcal{B}(L^2_{\rm per}))$, this yields
\begin{align*}
&\sup_{\substack{  k,X\in\mathbb{R}^2\\\zeta\in \Supp(\widetilde{f})  k,X\in\mathbb{R}^2}}  \|    (\partial h_{d,0,f}) (h_{d,0}-i)^{-1} \|_{\mathcal{B}(L^2_{\rm per})} \\
& \qquad =\sup_{\substack{  k,X\in\mathbb{R}^2\\\zeta\in \Supp(\widetilde{f})}}\|(\partial h_{d,0}) (h_{d,0}-i)^{-1}-\partial (h_{d,0}-h_{d,0,f} )(h_{d,0}-i)^{-1} 
\|_{\mathcal{B}(L^2_{\rm per})} \\
& \qquad \leq \sup_{\substack{  k,X\in\mathbb{R}^2\\\zeta\in \Supp(\widetilde{f})}}\left( \|(\partial h_{d,0}) (h_{d,0}-i)^{-1} \|_{\mathcal{B}(L^2_{\rm per})} + \| \partial (h_{d,0}-h_{d,0,f} ) \|_{\mathcal{B}(L^2_{\rm per})} \right)  < \infty
\end{align*}
and, using \eqref{h-elle},
\begin{align*}
&\sup_{\substack{  k,X\in\mathbb{R}^2\\\zeta\in \Supp(\widetilde{f})}} \|      (h_{d,0}-i)( h_{d,0,f}-\zeta)^{-1}  \|_{\mathcal{B}(L^2_{\rm per})} \\ & \qquad = \sup_{\substack{  k,X\in\mathbb{R}^2\\\zeta\in \Supp(\widetilde{f})}}\| 1+(h_{d,0}- h_{d,0,f} +\zeta-i)(h_{d,0,f}-\zeta)^{-1}\|_{\mathcal{B}(L^2_{\rm per})} < \infty.
\end{align*}
Combining the above two bounds yields
\begin{align*}
 \sup_{\substack{  k,X\in\mathbb{R}^2\\\zeta\in \Supp(\widetilde{f})}}\|  (\partial\,h_{d,0,f}) ( h_{d,0,f} -\zeta)^{-1} \|_{\mathcal{B}(L^2_{\rm per})} < \infty.
\end{align*}
Using respectively \eqref{h-elle} and \eqref{eq:bound_hdof-zeta-1}, this implies that 
\begin{align*}
 \MoveEqLeft \sup_{\substack{  k,X\in\mathbb{R}^2\\\zeta\in \Supp(\widetilde{f})  }}\|\partial (h_{d,0,f}-\zeta)^{-1}\|_{\mathcal{B}(L^2_{\rm per})} \\
 &=\sup_{\substack{  k,X\in\mathbb{R}^2\\\zeta\in \Supp(\widetilde{f})}}\| (h_{d,0,f}-\zeta)^{-1} (\partial\,h_{d,0,f}) ( h_{d,0,f} -\zeta)^{-1} \|_{\mathcal{B}(L^2_{\rm per})} < \infty,
\end{align*}
and
\begin{align*}
\MoveEqLeft \sup_{\substack{  k,X\in\mathbb{R}^2\\\zeta\in \Supp(\widetilde{f})  }}\omega^{-1} \|\partial (h_{d,0,f}-\zeta)^{-1}\|_{\mathcal{B}(L^2_{\rm per},H^2_{\rm per})} \\
   &= \sup_{\substack{  k,X\in\mathbb{R}^2\\\zeta\in \Supp(\widetilde{f})}}\omega^{-1}\| ( h_{d,0,f}-\zeta)^{-1} \, (\partial\,h_{d,0,f}) \, ( h_{d,0,f} -\zeta)^{-1}\|_{\mathcal{B}(L^2_{\rm per},H^2_{\rm per})} \\
   & \le\sup_{\substack{  k,X\in\mathbb{R}^2\\\zeta\in \Supp(\widetilde{f})}}\omega^{-1}\| ( h_{d,0,f}-\zeta)^{-1} \|_{\mathcal{B}(L^2_{\rm per},H^2_{\rm per})} \, \|(\partial\,h_{d,0,f}) \, ( h_{d,0,f} -\zeta)^{-1}\|_{\mathcal{B}(L^2_{\rm per})} < \infty.
\end{align*}
Higher order derivatives can be bounded by the same arguments, which leads to the desired result.
\end{proof}

\section{Weyl calculus for moir\'e systems} 
\label{sec:Weyl-moire}

\subsection{Weyl quantization of operator-valued symbols}
\label{sec:Weyl}

This section collects some results on Weyl quantization of operator-valued symbols. Most of them are already known, or are simple corollaries of results in \cite{dimassi1999spectral,gerard1991mathematical,gerard1998scattering,panati2003effective,teufel2003adiabatic}. 

\medskip

Let $\cE_1$, $\cE_2$ and $\cE$ be Banach spaces of linear operators such that 
\begin{align}\label{product_operators_E1_E2}
   \forall a\in\mathcal{E}_1,b\in \mathcal{E}_2, \quad ab \in \mathcal{E} \quad \mbox{and} \quad  \|ab\|_{\mathcal{E}}\leq \|a\|_{\mathcal{E}_1}\|b\|_{\mathcal{E}_2}.
\end{align}
Let $\omega: \R^2 \times \R^2 \to \R_+$ be an order function and let $\epsilon_0 > 0$. 
For any symbols $a_\bullet\in S^{\omega_1,\epsilon_0}(\mathcal{E}_1)$ and $b_\bullet\in S^{\omega_2,\epsilon_0}(\mathcal{E}_2)$, the Moyal product of $a_\bullet$ and $b_\bullet$ is defined as follows (see e.g. \cite[Theorem 4.1]{zworski2012semiclassical}):
\begin{align}\label{eq:Weyl-product}
({a_\epsilon}\# {b_\epsilon})_\epsilon(Y):=\frac{1}{(\pi \epsilon)^{4}}\int\limits_{\mathbb{R}^{4}}\int\limits_{\mathbb{R}^{4}}e^{\frac{2i}{\epsilon}\sigma(Y_1,Y_2)}{a_\epsilon}(Y+Y_1){b_\epsilon}(Y+Y_2)\,dY_1dY_2,
\end{align}
with $Y:=(k,X)$, $\sigma(Y_1,Y_2)=\left<J_2 Y_1,Y_2\right>$ and
\begin{align*}
    J_n:=\begin{pmatrix}
0&I_n\\
-I_n&0
\end{pmatrix}.
\end{align*}
It holds $({a_\epsilon}\#{b_\epsilon})_\bullet \in S^{\omega_1\omega_2,\epsilon_0}(\mathcal{E})$ and
\begin{align*}
  \forall 0 < \epsilon \le \epsilon_0, \quad  \Op(a_\epsilon)\Op(b_\epsilon)=\Op(({a}\#{b})_\epsilon).
\end{align*}

\begin{lemma}  \label{lem:weyl-compact}
Let $U$ be a bounded closed set of $\mathbb{R}^2$ and $\epsilon_0 > 0$. Let $\mathcal{E}_1$, $\mathcal{E}_2$ and $\mathcal{E}$ as in Definition~\ref{def:tau-equi},  $a_\bullet\in S^{1,\epsilon_0}(\mathcal{E}_1)$ and $b_\bullet\in S^{1,\epsilon_0}(\mathcal{E}_2)$ satisfying $b_\epsilon(k,X)=0$ for all $\epsilon\in (0,\epsilon_0]$, $k \in \R^2$ and $X\not\in U$. Then, for all $\alpha,\beta\in\mathbb{N}^2$, $j\in \mathbb{N}$ and $C > 0$, there is $C_j \in \R_+$ such that for all $\epsilon\in (0,\epsilon_0]$, $k\in \R^2$ and all $\dist(X,U)\geq \frac{1}{C}$,
\begin{align}\label{eq:disjoint}
 \dist(X,U)\geq \frac{1}{C} \quad \Longrightarrow \quad  \|\partial_k^{\alpha}\partial_X^{\beta}(a_\epsilon\# b_\epsilon)_\epsilon (k,X)\|_{\mathcal{E}}\leq C_{j}\frac{\epsilon^j}{\dist(X,U)^j}.
\end{align}
\end{lemma}

\begin{proof}
  The result readily follows from \cite[Proposition 9.5]{dimassi1999spectral}.
\end{proof}

The next results are concerned with the inverse of pseudodifferential operators with $\tau$-equivariant elliptic operator-valued symbols.

\begin{definition}[Elliptic operator-valued symbols] Let $\cH_1$ and $\cH_2$ be two separable Hilbert spaces and $\epsilon_0 > 0$. The symbol $a_\bullet\in S^{1,\epsilon_0}(\mathcal{B}(\cH_1,\cH_2))$ is called \textit{elliptic} if there exists a constant $\nu>0$ such that
\[
\inf_{\epsilon\in (0,\epsilon_0]}\inf_{k,X\in\mathbb{R}^2}\|a_\epsilon(k,X)\|_{\mathcal{B}(\cH_1,\cH_2)}\geq \nu>0.
\]
\end{definition}

\begin{proposition}\label{prop:ellip}
 Let $\epsilon_0 > 0$ and $a_\bullet\in S_{\tau}^{1,\epsilon_0}(\mathcal{B}(L^2_{\rm per}))$ be elliptic. Then there exist $0<\epsilon_0' \le \epsilon_0$ and $b_\bullet\in S_{\tau}^{1,\epsilon_0'}(\mathcal{B}(L^2_{\rm per}))$ such that 
\begin{align}
    \Op(a_\epsilon)\Op(b_\epsilon)=\Op(b_\epsilon)\Op(a_\epsilon)=\1_{\cH},
\end{align}
with 
\[
\sup_{\epsilon\in (0,\epsilon_0']}\|\Op(b_\epsilon)\|_{\mathcal{B}(\mathcal{H})}=\sup_{\epsilon\in (0,\epsilon_0']}\|\Op(a_\epsilon)^{-1}\|_{\mathcal{B}(\mathcal{H})} < \infty.
\]
\end{proposition}
\begin{proof}
The proof of the existence of $0 < \epsilon_0' \le \epsilon_0$ and $b_\bullet\in S^{1,\epsilon_0'}(\mathcal{B}(L^2_{\rm per}))$ satisfying 
\begin{align}\label{eq:op-inverse-B}
     a_\epsilon\# b_\epsilon=\1_{L^2_{\rm per}},\qquad 
     \Op( a_\epsilon)\Op(b_{\epsilon})=\Op(b_{\epsilon}) \Op( a_\epsilon)= \1_{\cH}
\end{align}
can be found in \cite[Page 99-100]{dimassi1999spectral} (see also \cite[Lemma 3.5 and Page 790]{dimassi1993developpements} or \cite[Theorem 4.29 and Theorem 8.3]{zworski2012semiclassical}). It remains to show that $b_\bullet$ is $\tau$-equivariant in the $k$ variable, i.e., that 
\begin{align*}
\forall G \in \mathbb{L}^*, \quad \forall (k,X) \in \R^2 \times \R^2, \quad    b_\epsilon(k-G,X)=\tau_Gb_\epsilon(k,X)\tau_G^{-1} .
\end{align*}
Let $G\in\mathbb{L}^*$. For all $0 < \epsilon \le \epsilon_0'$, we have
\begin{align*}
    \Op(a_\epsilon(\bullet-G,\bullet)) \Op(\tau_G b_\epsilon \tau_G^{-1}) &=
      \Op(\tau_G a_\epsilon \tau_G^{-1}) \Op(\tau_G b_\epsilon \tau_G^{-1}) \\
      &= \Op( (\tau_G a_\epsilon \tau_G^{-1})) \# (\tau_G b_\epsilon \tau_G^{-1})) \\ &= \Op(\tau_G (a_\epsilon \# b_\epsilon) \tau_G^{-1}) = \Op(\1_{L^2_{\rm per}}) = \1_{\cH},
\end{align*}
and likewise  $\Op(\tau_G b_\epsilon \tau_G^{-1}) \Op(a_\epsilon(\bullet-G,\bullet)) = \1_{\cH}$. It follows that
$$
\forall 0 < \epsilon \le \epsilon_0', \quad 
\Op(\tau_G b_\epsilon \tau_G^{-1}) = \Op(a_\epsilon(\bullet-G,\bullet))^{-1} = \Op(b_\epsilon (\bullet-G,\bullet)).
$$
It follows from Beals' Theorem (see e.g. \cite[Theorem 8.3]{zworski2012semiclassical} and Step 3. in its proof) that $b_\epsilon (\bullet-G,\bullet) =\tau_G b_\epsilon \tau_G^{-1}$. Thus, $b_\bullet\in S_{\tau}^{1,\epsilon_0'}(\mathcal{B}(L^2_{\rm per}))$.
\end{proof}

Finally, we recall the following variant of Calderon-Vaillancourt theorem.
\begin{theorem}\label{th:caldVa} Let $\epsilon_0 > 0$ and $a_\bullet\in S^{1,\epsilon_0}_{\tau}(\cB(L^2_{\rm per}))$. Then for all $0 < \epsilon \le \epsilon_0$, $\Op(a_\epsilon)\in \mathcal{B}(\mathcal{H})$ and
\begin{align*}
    \sup_{\epsilon\in (0,\epsilon_0]}\|\Op(a_\epsilon)\|_{\mathcal{B}(\mathcal{H})}<\infty.
\end{align*}
\end{theorem}
The proof of the above result can be read in \cite[Theorem 3, Appendix A]{panati2003effective}.

\subsection{Trace formula for \texorpdfstring{$\tau$}{}-equivariant operator-valued symbols} \label{sec:traceform}

The proofs of our main results are based on the following estimates on the trace of operators on $\mathcal{H}$ of the form $\Op(a_\epsilon)$ with $a_\bullet \in S^{1,\epsilon_0}_{\tau}(\mathfrak{S}_1(L^2_{\rm per}))$. 
\begin{lemma}\label{trace} Let $\epsilon_0 > 0$ and $a_\bullet \in S^{1,\epsilon_0}_{\tau}(\mathfrak{S}_1(L^2_{\rm per}))$ be such that for any $\alpha,\beta\in \N^2$,
\begin{align}\label{ass:A-trace}
\sup_{\substack{k\in\R^2\\ \epsilon\in (0,\epsilon_0]}}\int_{\R^2}\|\partial_k^\alpha\partial_X^\beta  a_\epsilon(k,X)\|_{\mathfrak{S}_1(L^2_{\rm per})}dX<\infty.
\end{align}
Then, for all $\epsilon\in (0,\epsilon_0]$, we have $\Op(a_\epsilon)\in \mathfrak{S}_1(\mathcal{H})$ and $\epsilon^2\|\Op(a_\epsilon)\|_{\mathfrak{S}_1(\mathcal{H})}$ is uniformly bounded in $\epsilon\in (0,\epsilon_0]$. Furthermore, for all $n\geq 2$, there exists a constant $C_n \in \R_+$ independent of $a_\bullet$ such that
\begin{align}\label{eq:trace1}
\epsilon^2\Tr_{\mathcal{H}}(\Op(a_\epsilon)) &=\frac{1}{(2\pi)^2}\int_{\mathbb{R}^2}\int_{\Omega^*} \Tr_{L^2_{\rm per}}\Big(a_\epsilon(k,X)\Big)dkdX  + \rho_n(\epsilon) \epsilon^{2n}, 
\end{align}
with
\begin{align}\label{eq:trace11}
|\rho_n(\epsilon)| \le C_n \sup_{\substack{k\in \R^2\\ \epsilon\in (0,\epsilon_0]} }\int_{\R^2}\|\Delta_X^na_\epsilon (k,X)\|_{\mathfrak{S}_1(L^2_{\rm per})}dX.
\end{align}
\end{lemma}

The rest of Section~\ref{sec:traceform} is devoted to the proof of this lemma.

\subsubsection{Integration by parts in semiclassical integrals}

Let us first introduce some notation:
\begin{itemize}
    \item for any function $f \in L^2_{\rm loc}(\R^2)$ or $g \in L^2_{\rm loc}(\R^2 \times \R)$, we set
\begin{align*}
\forall (x,R,z) \in \Omega \times  \mathbb L \times  \R, \quad &    I_{\#}(f)(x+R):=f(x), \\
& I_{\#}(g)(x+R,z):=g(x,z);
\end{align*}
\item for any function $f\in L^2(\R^3)$, we denote by 
\begin{align*}
 f_R:=f|_{(\Omega+R)\times \R_z}   \quad \mbox{so that} \quad f(x,z)=\sum_{R\in \mathbb{L}}f_R(x,z).
\end{align*}
Since $(\cU f_R)_k(x,z)=e^{-ik\cdot(x +R)}f(x+R,z)$ for all $(x,z)\in \Omega \times \R$, we get
\begin{align}\label{def:gR'-fR}
(\cU f)_k=\sum_{R\in \mathbb{L}} (\cU f_R)_k \quad \mbox{with} \quad    (\cU f_R)_k= I_{\#}(e^{-ik\cdot(\bullet +R)}f(\bullet+R)).
 \end{align}
\end{itemize}
Formally, we thus have for all $f,g \in L^2(\R^3)$,
\begin{align} 
&\langle \cU g,\Op(a) \cU f\rangle_\cH = \frac{1}{(2\pi\epsilon)^2}
\sum_{R \in \mathbb L} \sum_{R' \in \mathbb L}  \fint_{\Omega^*}  \int_{\R^2}   \int_{\R^2}   [\cI_{\epsilon,R,R'}^{g,f}({\rm Id},a)](k,k',X) \, dk' \,  dX \, dk 
\label{eq:UgOp(a)Uf_1}
\end{align}
where for all symbol $b \in S^1_\tau(\cB(L^2_{\rm per}))$ and field $\cL \in C^{\infty}(\R^2 \times \R^2;\cB(S^1_\tau(\cB(L^2_{\rm per}))))$, 
\begin{align} \label{eq:Ik+G} 
&[\cI_{\epsilon,R,R'}^{g,f}(\cL,b)](k,k',X):= e^{-i\frac{(k-k') \cdot X}\epsilon}
\left\langle (\cU g_R)_k,  \left[\cL(k-k',X) b\left( \frac{k+k'}2,X \right)\right] (\cU f_{R'})_{k'} \right\rangle_{L^2_{\rm per}} 
\end{align}
Integration by parts with respect to $k$, $k'$ and $X$ is a key tool to reformulate semiclassical expressions such as \eqref{eq:UgOp(a)Uf_1} as convergent series and integrals. To get rid of boundary terms in integrations by parts with respect to $k$, we use the following observation: since $b$ is $\tau$-equivariant and $k \mapsto (\cU v)_k$ is quasi-periodic for all $v \in L^2(\R^3)$, it holds
\begin{align*} 
\forall G \in \mathbb L^*, \quad [\cI_{R,R'}^{g,f}(\cL,b)](k,k'+G,X) = [\cI_{R,R'}^{g,f}(\cL,b)](k-G,k',X),
\end{align*}
from which we deduce that
\begin{align} 
& \fint_{\Omega^*}  \int_{\R^2}   \int_{\R^2}   [\cI_{\epsilon,R,R'}^{g,f}(\cL,b)](k,k',X) \, dk' \,  dX \, dk \nonumber \\
& \qquad \qquad= \int_{\R^2}   \int_{\R^2}   \fint_{\Omega^*} [\cI_{\epsilon,R,R'}^{g,f}(\cL,b)](k,k',X) \, dk' \,  dX \, dk
\label{eq:UgOp(a)Uf_2}
\end{align}
whenever the integrals are well-defined.

\medskip

The integrand in \eqref{eq:UgOp(a)Uf_1} is integrable w.r.t. $k$ (since $k$ takes its value in a compact set) and $X$ (in view of assumption~\eqref{ass:A-trace}), but not {\it a priori} with respect $k'$. As is standard in semiclassical analysis, we use the equality 
$$
\left(\frac{1+i\epsilon (k-k')\cdot\nabla_X}{1+|k-k'|^2} \right) e^{\frac{-i(k-k')X}{\epsilon}}=e^{\frac{-i(k-k')X}{\epsilon}}
$$
to formally integrate by parts with respect to the variable $X$ and obtain
\begin{align}\label{eq:k-k'}
\MoveEqLeft\fint\limits_{\Omega^*}\int\limits_{\R^2}\int\limits_{\R^2} [\cI_{\epsilon,R,R'}^{g,f}({\rm Id},a)](k,k',X) dk'dXdk \notag\\
&= \fint\limits_{\Omega^*}\int\limits_{\R^2}\int\limits_{\R^2}  [\cI_{\epsilon,R,R'}^{g,f}(\cL_{1,\epsilon}^{M_0},a)](k,k',X) dk'dXdk,
\end{align}
where $M_0 \in \N$ and where the field $\cL_{1,\epsilon}$ is defined as
\begin{align}\label{eq:L1}
  \cL_{1,\epsilon}(k_1,X_1):=\frac{1-i\epsilon k_1\cdot\nabla_X}{1+|k_1|^2}
\end{align}
(this field is in fact independent of the variable $X_1$).
The integrand on the RHS of \eqref{eq:k-k'} is integrable on $\Omega^* \times \R^2 \times \R^2$ as soon as $M_0 \ge 3$. We now need to further transform \eqref{eq:k-k'} to reveal that it is in fact the general term of a convergent series in $R$ and $R'$. For this purpose, we integrate by parts with respect to $k'$ and $k$, using \eqref{eq:UgOp(a)Uf_2} to deal with the latter.

\medskip

Since
\begin{align*}
i\nabla_{k} (\cU g_{R})_{k}&=(I_{\#}(\bullet)+R)I_{\#}\left(e^{-ik\cdot (\bullet+R)}g(\bullet+R)\right)=(I_{\#}(\bullet)+R)(\cU g_{R})_{k}, \\
i\nabla_{k'} (\cU f_{R'})_{k'}&=(I_{\#}(\bullet)+R')I_{\#}\left(e^{-ik'\cdot (\bullet+R')}f(\bullet+R')\right)=(I_{\#}(\bullet)+R')(\cU f_{R'})_{k'},  
\end{align*}
we have
\begin{align*}
   & (\nabla_k+\nabla_{k'}) \left([\cI_{\epsilon,R,R'}^{g,f}(\cL_{1,\epsilon}^{M_0},a)](k,k',X)\right)   \\ & \qquad = 
(\nabla_k+\nabla_{k'}) \left(  e^{-i\frac{(k-k') \cdot X}\epsilon}
\left\langle (\cU g_R)_k,  \left[\cL_{1,\epsilon}^{M_0}(k-k',X) a\left( \frac{k+k'}2,X \right)\right] (\cU f_{R'})_{k'} \right\rangle_{L^2_{\rm per}} \right)    \\ & \qquad = e^{-i\frac{(k-k') \cdot X}\epsilon}
\left\langle \nabla_k (\cU g_R)_k,  \left[\cL_{1,\epsilon}^{M_0}(k-k',X) a\left( \frac{k+k'}2,X \right)\right] (\cU f_{R'})_{k'} \right\rangle_{L^2_{\rm per}} \\  & \qquad + e^{-i\frac{(k-k') \cdot X}\epsilon}
\left\langle  (\cU g_R)_k,  \left[\cL_{1,\epsilon}^{M_0}(k-k',X)  (\nabla_k a)\left( \frac{k+k'}2,X \right)\right] (\cU f_{R'})_{k'} \right\rangle_{L^2_{\rm per}}     \\ & \qquad + e^{-i\frac{(k-k') \cdot X}\epsilon}
\left\langle (\cU g_R)_k,  \left[\cL_{1,\epsilon}^{M_0}(k-k',X) a\left( \frac{k+k'}2,X \right)\right] \nabla_{k'}(\cU f_{R'})_{k'} \right\rangle_{L^2_{\rm per}}   \\ & \qquad = i(R-R') \cI_{\epsilon,R,R'}^{g,f}(\cL_{1,\epsilon}^{M_0},a)(k,k',X) - i \cI_{\epsilon,R,R'}^{g,f}\left((i\nabla_k-{\rm ad}_{I_\#(\bullet)}) \cL_{1,\epsilon}^{M_0},a\right) (k,k',X).
\end{align*}
Using again\eqref{eq:Ik+G}-\eqref{eq:UgOp(a)Uf_2}, we obtain that for $M_0 \ge 3$,  $$
\fint_{\Omega^*}\int_{\R^2}\int_{\R^2}   (\nabla_k+\nabla_{k'}) \left([\cI_{\epsilon,R,R'}^{g,f}(\cL_{1,\epsilon}^{M_0},a)](k,k',X)\right) \, dk' \, dX \, dk=0,
$$
and therefore that
 \begin{align*}
 & (R-R') \fint_{\Omega^*}\int_{\R^2}\int_{\R^2}   \cI_{\epsilon,R,R'}^{g,f}(\cL_{1,\epsilon}^{M_0},a)(k,k',X) \, dk' \, dX \, dk \\
 & \qquad = \fint_{\Omega^*}\int_{\R^2}\int_{\R^2} \cI_{\epsilon,R,R'}^{g,f}\left((i\nabla_k-{\rm ad}_{I_\#(\bullet)}) \cL_{1,\epsilon}^{M_0},a\right) (k,k',X) \, dk' \, dX \, dk.
 \end{align*}
Introducing the uniform field  \begin{align}\label{eq:L2}
    \cL_{2,R}(k_1,X_1):= \frac{1+iR\cdot(i\nabla_k-\mathrm{ad}_{I_{\#}(\bullet)})}{1+|R|^2},
\end{align}
this equality yields
\begin{align}\label{eq:R-R'}
  \MoveEqLeft  \fint\limits_{\Omega^*}\int\limits_{\R^2}\int\limits_{\R^2} \cI_{\epsilon,R,R'}^{g,f}\left(\cL_{1,\epsilon}^{M_0},a_\epsilon\right)(k,k',X) \,  dk'dXdk\notag\\  &=\fint\limits_{\Omega^*}\int\limits_{\R^2}\int\limits_{\R^2} \cI_{\epsilon,R,R'}^{g,f}
  \left(\cL_{1,\epsilon}^{M_0}\cL_{2,R-R'}^{M_0},a_\epsilon \right)(k,k',X) \, dk'dXdk.
\end{align}

Finally, we have
$$
\frac{1+i(\epsilon^{-1}X- R)\cdot \nabla_{k}}{(1+|\epsilon^{-1}X- R|^2)} e^{\frac{-ik(X-\epsilon R)}{\epsilon}}=e^{\frac{-i k(X-\epsilon R))}{\epsilon}},
$$
and
\begin{align*}
[\cI_{\epsilon,R,R'}^{g,f}(\cL,b)](k,k',X)=& e^{-i\frac{k \cdot (X-\epsilon R)}\epsilon} e^{i\frac{k' \cdot X}\epsilon} \\
& \times 
\left\langle I_\#(e^{-ik\cdot \bullet}g(\bullet+R)),  \left[\cL(k-k',X) b\left( \frac{k+k'}2,X \right)\right] (\cU f_{R'})_{k'} \right\rangle_{L^2_{\rm per}} .
\end{align*}
Introducing the field  
\begin{align}\label{eq:L3}
    \cL_{3,R,\epsilon}(k_1,X_1):=  \frac{1+(\epsilon^{-1}X_1- R)\cdot (I_{\#}(\bullet)-i\nabla_{k})}{(1+|\epsilon^{-1}X_1- R|^2)}
\end{align}
(this field is independent of the variable $k_1$)
and proceeding as above, we obtain 
\begin{align}\label{eq:X-R}
&\fint\limits_{\Omega^*}\int\limits_{\R^2}\int\limits_{\R^2} \cI_{\epsilon,R,R'}^{g,f}
  \left(\cL_{1,\epsilon}^{M_0}\cL_{2,R-R'}^{M_0},a_\epsilon \right)(k,k',X) \, dk'dXdk \notag\\
  &\qquad =\fint\limits_{\Omega^*}\int\limits_{\R^2}\int\limits_{\R^2} \cI_{\epsilon,R,R'}^{g,f}
\left(\cL_{3,R,\epsilon}^{M_0}\cL_{1,\epsilon}^{M_0}\cL_{2,R-R'}^{M_0},a_\epsilon \right)(k,k',X) \, dk'dXdk.
\end{align}
Gathering together~\eqref{eq:k-k'}, \eqref{eq:R-R'} and \eqref{eq:X-R}, we finally obtain
\begin{align} \label{eq:inte-parts}
&\fint\limits_{\Omega^*}\int\limits_{\R^2}\int\limits_{\R^2} \cI_{\epsilon,R,R'}^{g,f}
  \left({\rm Id},a_\epsilon \right)(k,k',X) \, dk'dXdk \notag \\
  &\qquad =\fint\limits_{\Omega^*}\int\limits_{\R^2}\int\limits_{\R^2} \cI_{\epsilon,R,R'}^{g,f}
\left(\cL_{3,R,\epsilon}^{M_0}\cL_{1,\epsilon}^{M_0}\cL_{2,R-R'}^{M_0},a_\epsilon \right)(k,k',X) \, dk'dXdk. 
\end{align}
Whenever $M_0 \ge 3$, the integrand in the right-hand side is integrable and the resulting integral is the general term of an absolutely convergent series in $R$ and $R'$. Indeed, in view of the definitions of the fields $\cL_{1,\epsilon}$, $\cL_{2,R}$ and $\cL_{3,R,\epsilon}$, there exists a universal constant $C \in \R_+$ such that
\begin{align} \label{eq:bound_gOpaf}
&\left|\cI_{\epsilon,R,R'}^{g,f}
\left(\cL_{3,R,\epsilon}^{M_0}\cL_{1,\epsilon}^{M_0}\cL_{2,R-R'}^{M_0},a_\epsilon \right)(k,k',X)  \right|   \le C  {\dps \max_{|\alpha|,|\beta|,|\gamma_1|,|\gamma_2|\leq 2M_0}} \notag \\
&  \qquad  \frac{ \left(\left\|\left|a_{\epsilon,\alpha,\beta}^{\gamma_1,\gamma_2}\left(\frac{k+k'}{2},X\right)\right|^{1/2} (\cU g_{R})_{k}\right\|_{L^2_{\rm per}}^2 + \left\|\left|a_{\epsilon,\alpha,\beta}^{\gamma_1,\gamma_2}\left(\frac{k+k'}{2},X\right)\right|^{1/2} (\cU f_{R'})_{k'}\right\|_{L^2_{\rm per}}^2 \right) }{(1+|R-R'|^2)^{M_0/2} (1+|k-k'|^2)^{M_0/2} (1+|\epsilon^{-1}X-R|^2)^{M_0/2}},
\end{align}
where 
$$
a_{\epsilon,\alpha,\beta}^{\gamma_1,\gamma_2}(k,X):=I_{\#}(\bullet)^{\gamma_1} \partial_k^\alpha\partial_X^\beta a_\epsilon(k,X) I_{\#}(\bullet)^{\gamma_2},
$$
from which we deduce, using the boundedness of the operator $I_\#(\bullet)$ on $L^2_{\rm per}$, 
\begin{align} \label{eq:bound_gOpaf'}
&\left|\cI_{\epsilon,R,R'}^{g,f}
\left(\cL_{3,R,\epsilon}^{M_0}\cL_{1,\epsilon}^{M_0}\cL_{2,R-R'}^{M_0},a_\epsilon \right)(k,k',X)  \right|   \le C  {\dps \max_{|\alpha|,|\beta|\leq 2M_0}} \left\|\partial_k^\alpha\partial_X^\beta a_\epsilon \left(\frac{k+k'}2,X \right)\right\|_{\cB(L^2_{\rm per})}\notag \\
&  \qquad \times  \frac{ \left(\left\|(\cU g_{R})_{k}\right\|_{L^2_{\rm per}}^2 + \left\|(\cU f_{R'})_{k'}\right\|_{L^2_{\rm per}}^2 \right) }{(1+|R-R'|^2)^{M_0/2} (1+|k-k'|^2)^{M_0/2} (1+|\epsilon^{-1}X-R|^2)^{M_0/2}}.
\end{align}
Using assumption~\eqref{ass:A-trace}, the bound $(1+|\epsilon^{-1}X-R|^2) \ge 1$, the Parseval relation
$$
\forall \phi \in L^2(\R^3), \quad \sum_{R \in \mathbb L} \fint_{\Omega^*} 
\|(\cU \phi_R)_k\|_{L^2_{\rm per}}^2 dk = 
\sum_{R \in \mathbb L} \|\phi_R\|_{L^2(\R^3)}^2 = \|\phi\|_{L^2(\R^3)}^2,
$$
and the fact that for any $M_0\geq 3$,
\begin{align*}
   & \sup_{X\in \R^2}\sum_{R\in\mathbb{L}}\frac{1}{(1+|\epsilon^{-1}X-R|^2)^{M_0}}=\sup_{Y\in\Omega^*}\sum_{R\in\mathbb{L}}\frac{1}{(1+|Y-R|^2)^{M_0}} <\infty, 
\end{align*}
since $Y\mapsto \sum_{R'\in\mathbb{L}}\frac{1}{(1+|Y-R'|^2)^{M_0}}$ is $\mathbb{L}$-periodic,
we get
\begin{align*}
&\sum_{R \in \mathbb L} \sum_{R' \in \mathbb L}  \fint_{\Omega^*}  \int_{\R^2}   \int_{\R^2} \left|\cI_{\epsilon,R,R'}^{g,f}
\left(\cL_{3,R,\epsilon}^{M_0}\cL_{1,\epsilon}^{M_0}\cL_{2,R-R'}^{M_0},a_\epsilon \right)(k,k',X)  \right| \, dk' dX dk \\
& \qquad \le C  {\dps \max_{|\alpha|,|\beta|\leq 2M_0}}  \sup_{\substack{k\in\R^2\\ \epsilon\in (0,\epsilon_0]}}\int_{\R^2}\|\partial_k^\alpha\partial_X^\beta  a_\epsilon(k,X)\|_{\cB(L^2_{\rm per})}dX 
 \left( \|g\|_{L^2(\R^3)}^2 +\|f\|_{L^2(\R^3)}^2 \right) < \infty.
\end{align*}
Note that the crude bound $(1+|\epsilon^{-1}X-R|^2) \ge 1$ is sufficient to obtain this result. The term 
$(1+|\epsilon^{-1}X-R|^2)^{-M_0/2}$ in \eqref{eq:bound_gOpaf} will  play a more important role in the trace estimates below.

\subsubsection{Proof of Lemma \ref{trace}}\label{sec:5.2.2}

Let us first establish the following lemma.

\begin{lemma}\label{lem:trace-L2per}
    Let $a \in \mathfrak{S}_1(L^2_{\rm per})$ and $(e_j)_{j \in \N}$ an orthonormal basis of $L^2(\R^3)$. Then for all $k \in \R^2$ and $R \in \R^2$, it holds
    $$
    \|a\|_{\mathfrak{S}_1(L^2_{\rm per})} = \sum_{j \in \N} \left\| |a|^{1/2} I_\#(e^{-ik (\bullet + R)} e_j(\bullet + R)) \right\|_{L^2_{\rm per}}^2.
    $$
\end{lemma}

\begin{proof}
Let $(\lambda_\ell)_{\ell \in \N}$ be the singular values of $a$, counted with their multiplicities, $(u_\ell)_{\ell \in \N}$ be an orthonormal basis of $L^2_{\rm per}$ such that 
\begin{align*}
       |a|=  \sum_{\ell \in \N} \lambda_\ell\left|u_\ell\right>\left<u_\ell\right|,
    \end{align*}
    and $\widetilde u_\ell \in L^2(\R^3)$ the extension by zero of the function equal to $u_\ell$ on $\Omega \times \R$.
    We have
    \begin{align*}
    &\sum_{j \in \N} \left\| |a|^{1/2} I_\#(e^{-ik (\bullet + R)} e_j(\bullet + R)) \right\|_{L^2_{\rm per}}^2  \\ & \qquad = \sum_{j \in \N} \langle I_\#(e^{-ik (\bullet + R)} e_j(\bullet + R)), |a|  I_\#(e^{-ik (\bullet + R)} e_j(\bullet + R)) \rangle_{L^2_{\rm per}} \\  & \qquad = 
    \sum_{j \in \N} \sum_{\ell \in \N} \lambda_\ell  |\langle I_\#(e^{-ik (\bullet + R)} e_j(\bullet + R)), u_\ell \rangle_{L^2_{\rm per}}|^2 \\ & \qquad = 
     \sum_{\ell \in \N} \lambda_\ell  \left( \sum_{j \in \N}|\langle e^{-ik (\bullet + R)} e_j(\bullet + R), \widetilde u_\ell \rangle_{L^2(\R^3)}|^2 \right).
    \end{align*}
    Observing that for all $k$ and $R$ in $\R^2$, $(e^{-ik(\bullet +R)}e_j(\bullet+R))_{j \in \N}$ is an orthonormal basis of $L^2(\R^3)$, and using Parseval's equality and the fact that $\|\widetilde u_\ell\|_{L^2(\R^3)}=\|u_\ell\|_{L^2_{\rm per}}=1$, we obtain
    $$
    \sum_{j \in \N} \| |a|^{1/2} I_\#(e^{-ik (\bullet + R)} e_j(\bullet + R)) \|_{L^2_{\rm per}}^2 = \sum_{\ell \in \N} \lambda_\ell = \|a\|_{{\mathfrak S}_1(L^2_{\rm per})},
    $$
    which completes the proof.
\end{proof}

The proof of Lemma \ref{trace} is inspired by \cite[Chapter 9]{dimassi1999spectral}. 

\medskip

\noindent
{\bf Step 1.} Let us first show that $\Op(a_\epsilon)\in \mathfrak{S}_1(\cH)$. Using the unitarity of the Bloch transform and the variational characterization of the trace norm, we have
\begin{align*}
    \|\Op(a_\epsilon)\|_{\mathfrak{S}_1(\cH)}&=\|\cU^{-1}\Op(a_\epsilon) \cU\|_{\mathfrak{S}_1(L^2(\R^3))}\\
    &=\sup_{(f_j),(g_j)}\sum_{j \in \N}\left< g_j,\cU^{-1}\Op(a_\epsilon) \cU f_j\right>_{L^2(\R^3)}\\
    &=\sup_{(f_j),(g_j)}\sum_{j \in \N}\fint_{\Omega^*}\left<(\cU g_j)_k,\Op(a_\epsilon)(\cU f_j)_k\right>_{L^2_{\rm per}}dk,
\end{align*}
where the supremum is taken over the set of pairs of orthonormal bases $(f_j)_{j \in \N}$ and $(g_j)_{j \in \N}$ of $L^2(\R^3)$.

From \eqref{eq:UgOp(a)Uf_1} and \eqref{eq:bound_gOpaf}, we deduce that for any integer $M_0 \ge 3$, we have
\begin{align}\label{eq:Ugj-Ufj}
  \MoveEqLeft \epsilon^2\left| \sum_{j\in \N}\fint_{\Omega}\left<(\cU g_j)_k,\Op(a_\epsilon)(\cU f_j)_k\right>_{L^2_{\rm per}}dk\right|\notag\\
  &
  \leq C\sum_{j \in \N}\;\max_{|\alpha|,|\beta|,|\gamma_1|,|\gamma_2|\leq 2M_0}\;\sum_{R,R'\in \mathbb{L}} \frac{1}{(1+|R-R'|^2)^{M_0/2}}\notag\\
  &\quad\times \int_{\R^2}\int_{\R^2}\fint_{\Omega}\Bigg[\frac{1}{(1+|k-k'|^2)^{M_0/2}}\frac{1}{(1+|\epsilon^{-1}X- R|^2)^{M_0/2}}\notag\\
  &\quad\times \Bigg(\left\|\left|a_{\epsilon,\alpha,\beta}^{\gamma_1,\gamma_2}\left(\frac{k+k'}{2},X\right)\right|^{1/2} (\cU g_{j,R'})_{k}\right\|_{L^2_{\rm per}}^2\notag\\
  &\quad+ \left\|\left|a_{\epsilon,\alpha,\beta}^{\gamma_1,\gamma_2}\left(\frac{k+k'}{2},X\right)\right|^{1/2} (\cU f_{j,R})_{k'}\right\|_{L^2_{\rm per}}^2\Bigg)\Bigg]dk'dXdk,
\end{align}
where $C \in \R_+$ is a universal constant. As $I_\# \in \cB(L^2_{\rm per})$, we have
\begin{align*}
\left\|a_{\epsilon,\alpha,\beta}^{\gamma_1,\gamma_2}\left(\frac{k+k'}{2},X\right)\right\|_{\mathfrak{S}_1(L^2_{\rm per})}\leq  \|I_\#\|_{\cB(L^2_{\rm per})}^{4M_0} \left\|\partial_k^\alpha\partial_X^\beta a_\epsilon\left(\frac{k+k'}{2},X\right)\right\|_{\mathfrak{S}_1(L^2_{\rm per})}.
\end{align*}
Then from Lemma \ref{lem:trace-L2per} and \eqref{def:gR'-fR}, we infer
\begin{align*}
   \MoveEqLeft \epsilon^2\left|\sum_{j}\fint_{\Omega}\left<(\cU g_j)_k,\Op(a_\epsilon)(\cU f_j)_k\right>_{L^2_{\rm per}}dk\right|\\
   &\leq C \max_{|\alpha|,|\beta|\leq 2M_0 }\;\sum_{R,R'\in \mathbb{L}} \frac{1}{(1+|R-R'|^2)^{M_0/2}}\int_{\R^2}\int_{\R^2}\fint_{\Omega}\frac{1}{(1+|k-k'|^2)^{M_0/2}}\\
  &\qquad\qquad\qquad\times \frac{1}{(1+|\epsilon^{-1}X- R|^2)^{M_0/2}}\left\|\partial_k^\alpha\partial_X^\beta  a_\epsilon\left(\frac{k+k'}{2},X\right)\right\|_{\mathfrak{S}_1(L^2_{\rm per})}dk'dXdk\\
  &\leq C \max_{|\alpha|,|\beta|\leq 2M_0 }\;\int_{\R^2}\int_{\R^2}\fint_{\Omega}\frac{1}{(1+|k-k'|^2)^{M_0/2}}\\
  &\qquad\qquad\qquad\times\sum_{R'\in \mathbb{L}} \frac{1}{(1+|\epsilon^{-1}X- R'|^2)^{M_0/2}}\left\|\partial_k^\alpha\partial_X^\beta  a_\epsilon\left(\frac{k+k'}{2},X\right)\right\|_{\mathfrak{S}_1(L^2_{\rm per})}dk'dXdk\\
  &\leq C\max_{|\alpha|,|\beta|\leq 2M_0 }\;\int_{\R^2}\int_{\R^2}\fint_{\Omega}\frac{1}{(1+|k-k'|^2)^{M_0/2}}\left\|\partial_k^\alpha\partial_X^\beta  a_\epsilon\left(\frac{k+k'}{2},X\right)\right\|_{\mathfrak{S}_1(L^2_{\rm per})}dk'dXdk\\
  &\leq C \max_{|\alpha|,|\beta|\leq 2M_0 }\; \sup_{k\in\R^2}\int_{\R^2}\left\|\partial_k^\alpha\partial_X^\beta  a_\epsilon(k,X)\right\|_{\mathfrak{S}_1(L^2_{\rm per})}dX.
\end{align*}

As a result, we obtain that for any $M_0\geq 3$,
\begin{align}\label{eq:5.25}
    \sup_{\epsilon\in (0,\epsilon_0]} \epsilon^2 \|\Op(a_\epsilon)\|_{\mathfrak{S}_1(\cH)}\leq C \max_{|\alpha|,|\beta|\leq 2M_0 }\;\sup_{\substack{k\in\R^2\\ \epsilon\in (0,\epsilon_0]}}\int_{\R^2}\|\partial_k^\alpha\partial_X^\beta  a_\epsilon(k,X)\|_{\mathfrak{S}_1(L^2_{\rm per})}dX<\infty.
\end{align}
This shows that $\Op(a_\epsilon)\in \mathfrak{S}_1(\mathcal{H})$ and $\epsilon^2\|\Op(a_\epsilon)\|_{\mathfrak{S}_1(\mathcal{H})}$ is uniformly bounded in $\epsilon\in (0,\epsilon_0]$.

\medskip

\noindent
{\bf Step 2.} Let us now prove \eqref{eq:trace1}. Let $(\psi_{n,R})_{n\in \N,R\in \mathbb{L}}$ be an orthonormal basis of Wannier functions of $L^2(\R^3)$ constructed as follows:
$$
\forall (n,R,x,z) \in \N \times \mathbb L \times \R^2 \times \R, \quad  \psi_{n,R}(x,z) = u_n|_{\Omega \times \R}(x-R,z),
$$
where $(u_n)_{n \in \N}$ is an orthonormal basis of $L^2_{\rm per}$ (each $\psi_{n,R}$ is supported in $(R+\Omega) \times \R$). For such a basis, it holds
$$
(\cU\psi_{n,R})_k = e^{-ik \cdot R} (\cU\psi_{n,0})_k.
$$
We then have
\begin{align}\label{eq:0.11}
  \MoveEqLeft \Tr_{\mathcal{H}}(\Op(a_\epsilon))=\sum_{R\in\mathbb{L}}\sum_{n\in \N}\left<\mathcal{U}\psi_{n,R}, \Op(a_\epsilon)\mathcal{U}\psi_{n,R}\right>_{\mathcal{H}}\notag\\
&=\frac{1}{(2\pi\epsilon)^2} \sum_{n\in \N}\sum_{R\in\mathbb{L}}  \int\limits_{\R^2}\int\limits_{\R^2}\fint\limits_{\Omega^*}\left\langle   (\cU\psi_{n,0})_k,a_\epsilon\left(\frac{k+k'}{2},X\right)(\cU\psi_{n,0})_{k'}\right\rangle_{L^2_{\rm per}}\notag\\
&\quad\qquad\qquad\qquad\qquad\qquad\qquad\qquad\qquad\qquad \times e^{\frac{-i(k-k')\cdot X}{\epsilon}} e^{i(k-k')\cdot R} dk'dXdk.
\end{align}
Using Poisson's formula
\begin{align*}
\sum_{R\in\mathbb{L}}  e^{-i(k-k')\cdot R}=|\Omega^*|\sum_{G\in\mathbb{L}^*}\delta(k-k'-G),
\end{align*}
we obtain
\begin{align*}
     \MoveEqLeft \epsilon^2\Tr_{\mathcal{H}}(\Op(a_\epsilon))\\
     &=\frac{1}{(2\pi)^2}\sum_{G\in\mathbb{L}^*}\int\limits_{\R^2}\int\limits_{\Omega^*}\sum_{n\in \N} \left\langle  (\cU \psi_{n,0})_{k'+G}, a_\epsilon\left(k'+\frac{G}{2},X\right)(\cU \psi_{n,0})_{k'}\right\rangle_{L^2_{\rm per}} e^{i\frac{G\cdot X}{\epsilon}}dk'dX.
\end{align*}

For $G\neq 0$, we have
$-\frac{\epsilon^2}{|G|^2}\Delta_X e^{i\frac{G\cdot X}{\epsilon}}=e^{i\frac{G\cdot X}{\epsilon}}$, to that for any $m\in \N$,
\begin{align*}
 \MoveEqLeft   \left|\int\limits_{\R^2}\fint\limits_{\Omega^*}\sum_{n\in \N} \left\langle   (\cU \psi_{n,0})_{k'+G}, a_\epsilon\left(k'+\frac{G}{2},X\right)(\cU \psi_{n,0})_{k'}\right\rangle_{L^2_{\rm per}} e^{\frac{iG\cdot X}{\epsilon}}dk'dX\right|\\
 &=\frac{\epsilon^{2m}}{|G|^{2m}}\left|\int\limits_{\R^2}\fint\limits_{\Omega^*}\sum_{n\in \N} \left\langle   (\cU \psi_{n,0})_{k'+G}, \Delta_X^m a_\epsilon\left(k'+\frac{G}{2},X\right)(\cU \psi_{n,0})_{k'}\right\rangle_{L^2_{\rm per}} \!\!\!\!\!\!\!\!e^{\frac{iG\cdot X}{\epsilon}}dk'dX\right|\\
 &\leq C\frac{\epsilon^{2m}}{|G|^{2m}}\sup_{k\in \R^2}\int_{\R^2}\|\Delta_X^m a_\epsilon\|_{\mathfrak{S}_1(L^2_{\rm per})}(k,X)dX
\end{align*}
since $((\cU \psi_{n,0})_{k+G})_{n\geq 1}$ and $((\cU \psi_{n,0})_{k})_{n\geq 1}$ are two orthonormal bases of $L^2_{\rm per}$ and
\begin{align*}
 \MoveEqLeft   \sum_{n\in \N} \left|\left\langle   (\cU \psi_{n,0})_{k'+G}, \Delta_X^m a_\epsilon\left(k'+\frac{G}{2},X\right)(\cU \psi_{n,0})_{k'}\right\rangle_{L^2_{\rm per}}\right|\\
 &\leq \frac{1}{2}\sum_{n\in \N} \left\|\left|\Delta_X^m a_\epsilon\left(k'+\frac{G}{2},X\right)\right|^{1/2} (\cU \psi_{n,0})_{k'+G}\right\|_{L^2_{\rm per}}^2\\
 &\quad+\frac{1}{2}\sum_{n\in \N}\left\|\left|\Delta_X^m a_\epsilon\left(k'+\frac{G}{2},X\right)\right|^{1/2} (\cU \psi_{n,0})_{k'}\right\|_{L^2_{\rm per}}^2\\
 &= \left\|\Delta_X^m a_\epsilon \left(k+\frac{G}{2},X\right)\right\|_{\mathfrak{S}_1(L^2_{\rm per})}.
\end{align*}
Denoting by 
\begin{align*}
\rho_m(\epsilon):= \epsilon^{-2m} \sum_{G\in\mathbb{L}^*\setminus\{0\}}\int\limits_{\R^2}\fint\limits_{\Omega^*}\sum_{n\in \N} \left\langle   (\cU \psi_{n,0})_{k'+G}, a_\epsilon\left(k'+\frac{G}{2},X\right)(\cU \psi_{n,0})_{k'}\right\rangle_{L^2_{\rm per}}e^{\frac{-iG\cdot X}{\epsilon}}dk'dX,
\end{align*}
and using $\left( - \frac{\epsilon^2}{|G|^2} \Delta_X \left( e^{\frac{-iG\cdot \bullet}{\epsilon}} \right) \right)(X)= e^{\frac{-iG\cdot X}{\epsilon}}$ for $G\neq 0$, we obtain
$$
|\rho_m(\epsilon)| \le C\sup_{k\in \R^2}\int_{\R^2}\|\Delta_X^m a_\epsilon\|_{\mathfrak{S}_1(L^2_{\rm per})}(k,X)dX,
$$
and 
\begin{align*}
    \epsilon^2\Tr_{\mathcal{H}}(\Op(a_\epsilon))&=\frac{1}{(2\pi)^2}\int\limits_{\R^2}\int\limits_{\Omega^*}\sum_{n\in \N} \left\langle   (\cU \psi_{n,0})_{k'}, a_\epsilon\left(k',X\right)(\cU \psi_{n,0})_{k'}\right\rangle_{L^2_{\rm per}}dk'dX + \epsilon^{2m} \rho_m(\epsilon) \\
     & = \frac{1}{(2\pi)^2}\int\limits_{\R^2}\int\limits_{\Omega^*} \tr_{L^2_{\rm per}}\left(a_\epsilon\left(k,X\right) \right) \, dk \, dX + \epsilon^{2m} \rho_m(\epsilon).
\end{align*}
This gives \eqref{eq:trace1} and completes the proof.

\subsection{Twisted Weyl calculus}
\label{sec:TWC}

We recall that $c:[-1,1] \ni \epsilon \mapsto c(\epsilon):= \frac{1-\sqrt{1-\epsilon^2}}\epsilon \in \R$, and that  the twisted Weyl quantization operator ${\rm Op}^c_\epsilon$ is defined in \eqref{eq:TWQ}. We formally define the twisted Moyal product $\widetilde\#_c$ of two symbols $a$ and $b$ as
\begin{align}
&(a \widetilde\#_c b)_\epsilon(Y) := \mathcal{T}^{-1}_{c(\epsilon)X}\big[a^c_\bullet \# b^c_\bullet )\big]_\epsilon(Y) \mathcal{T}_{c(\epsilon)X} \nonumber \\
&= \frac{1}{(\pi\epsilon)^4} \int_{\R^4 \times \R^4} \!\!\!\!
e^{\frac{2i}\epsilon \sigma(Y_1,Y_2)} \cT_{c(\epsilon)X_1}
a(Y+Y_1) \cT_{c(\epsilon)(X_2-X_1)} b(Y+Y_2) \cT_{-c(\epsilon)X_2} \, dY_1 \, dY_2, \label{eq:Moyal-twi}
\end{align}
with $Y=(k,X)$, $Y_1=(k_1,X_1)$, $Y_2=(k_2,X_2)$.

 \begin{proposition}\label{prop:PS-asymp}
Let $\mathcal{E},\mathcal{E}_1,\mathcal{E}_2$ be Banach spaces of the form $\cB(H^{s_1}_{\rm per},H^{s_2}_{\rm per})$ or ${\mathfrak S}_1(L^2_{\rm per})$ satisfying
\begin{align}\label{Holder}
   \forall a\in\mathcal{E}_1,b\in \mathcal{E}_2, \quad ab \in \mathcal{E} \quad \mbox{and} \quad  \|ab\|_{\mathcal{E}}\leq \|a\|_{\mathcal{E}_1}\|b\|_{\mathcal{E}_2}.
\end{align}
Let $0 < \epsilon_0 \le 1$. For any $a_\bullet \in S_{\twi}^{\omega_1,\epsilon_0}(\mathcal{E}_1)$ and $b_\bullet \in S_{\twi}^{\omega_2,\epsilon_0}(\mathcal{E}_2)$, $(a_\bullet \widetilde\#_c b_\bullet)_\bullet \in S_{\twi}^{\omega_1\omega_2,\epsilon_0}(\cE)$ and it holds
\begin{align}
\forall 0 < \epsilon \le \epsilon_0, \quad \Opc(a_\epsilon)\Opc(b_\epsilon)=\Opc((a_\bullet \widetilde \#_c b_\bullet )_\epsilon).
\end{align}
Furthermore, if $a\in S_{\twi}^{\omega_1}(\mathcal{E}_1)$ and $ b\in S_{\twi}^{\omega_2}(\mathcal{E}_2)$, then $(a \widetilde\#_c b)_\bullet$ is a semiclassical symbol in $S_{\twi}^{\omega_1\omega_2,1}(\epsilon,\cE)$, i.e., there exist $(d_{j})_{j \in \N}\in (S_{\twi}^{\omega_1\omega_2}(\mathcal{E}))^\N$ such that for all $n \in \N$, 
\begin{align}\label{eq:48}
    (a \widetilde\#_c b)_\epsilon = \sum_{j=0}^{n-1}\epsilon^j d_{j} + \mathcal{O}_{S_{\twi}^{\omega_1\omega_2,1}(\mathcal{E})}(\epsilon^{n+1})
\end{align}
In particular, we have
\begin{align}\label{eq:asym-2}
    d_0(k,X)=a(k,X)b(k,X), \quad d_1(k,K)=\frac{i}{2}\{a,b\}(k,X),
\end{align}
and  
\begin{align*}
    d_2(k,X)=-\frac{1}{8}\{a,b\}_2(k,X)+\frac{i}{4}\Big(\partial_k a (k,X) \, {\rm ad}_{\partial_x}b(k,X)-{\rm ad}_{\partial_x}a (k,X) \; \partial_k b(k,X)\Big).
\end{align*}
\end{proposition}
\begin{proof}[Proof of Proposition \ref{prop:PS-asymp}]
In view of \eqref{eq:N<N-twi}, we have $a_\bullet^c\in S^{\omega_1,\epsilon_0}(\cE_1)$ and $b_\bullet^c\in S^{\omega_2,\epsilon_0}(\cE_2)$. By \eqref{eq:Weyl-product} and standard composition for Weyl quantization (see e.g. \cite[Theorem 4.18]{zworski2012semiclassical}), $(a^c_\bullet \# b^c_\bullet)_\bullet$ is a symbol in $S^{\omega_1\omega_2,\epsilon_0}(\cE)$. As for all $a \in \R^2$ and $s \in \R$, $\cT_a$ is a unitary operator on $H^s_{\rm per}$, the function $\R^2 \times \R^2 \ni (k,X) \mapsto (a_\bullet \widetilde\#_c b_\bullet)_\epsilon \in \cE$ is well-defined and that
$$
\cN^{\omega_1\omega_2,\cE,\epsilon_0}_{0,0,0}((a_\bullet \widetilde\#_c b_\bullet)_\bullet) = \cN^{\omega_1\omega_2,\cE,\epsilon_0}_{0,0}((a^c_\bullet \# b^c_\bullet)_\bullet) < \infty.
$$ 
From the integral representation \eqref{eq:Moyal-twi}, we infer that
\begin{align*}
 \partial_k  (a_\bullet \widetilde\#_c b_\bullet)_\epsilon &= 
(\partial_k a_\bullet \widetilde\#_c b_\bullet)_\epsilon + 
    (a_\bullet \widetilde\#_c \partial_k b_\bullet)_\epsilon, \\
     \partial_X  (a_\bullet \widetilde\#_c b_\bullet)_\epsilon &= 
(\partial_X a_\bullet \widetilde\#_c b_\bullet)_\epsilon + 
    (a_\bullet \widetilde\#_c \partial_X b_\bullet)_\epsilon \\
      {\rm ad}_{\partial_x}  (a_\bullet \widetilde\#_c b_\bullet)_\epsilon, &= 
({\rm ad}_{\partial_x} a_\bullet \widetilde\#_c b_\bullet)_\epsilon +     (a_\bullet \widetilde\#_c \, {\rm ad}_{\partial_x} b_\bullet)_\epsilon .
\end{align*}
Since $S^{\omega,\epsilon_0}_{\rm twi}(\cE)$ is invariant under the actions of $\partial_k$, $\partial_X$, and ${\rm ad}_{\partial x}$, we deduce from the same arguments as above that for all $\alpha, \beta,\gamma \in \N^2$, there exists a constant $C_{\alpha,\beta,\gamma} \in \R_+$ such that
\begin{align*}
&\cN^{\omega_1\omega_2,\cE,\epsilon_0}_{\alpha,\beta,\gamma}((a_\bullet \widetilde\#_c b_\bullet)_\bullet) \\
& \quad \le C_{\alpha,\beta,\gamma} \sum_{(\alpha',\beta',\gamma') \le (\alpha,\beta,\gamma)} \cN^{\omega_1\omega_2,\cE,\epsilon_0}_{0,0}\left( \left((\partial_k^{\alpha'}\partial_X^{\beta'}{\rm ad}_{\partial_x}^{\gamma'} a_\bullet)^c_\bullet 
\# 
(\partial_k^{\alpha-\alpha'}\partial_X^{\beta-\beta'}{\rm ad}_{\partial_x}^{\gamma-\gamma'} b_\bullet)^c_\bullet\right)_\bullet \right) < \infty.
\end{align*}
It also follows from the integral representation~\eqref{eq:Moyal-twi} and the fact that 
$$
\forall G \in \mathbb{L}^*, \quad \forall x_0 \in \R^2, \quad \tau_G \cT_{x_0} = e^{iG \cdot x_0} \cT_{x_0} \tau_G,
$$
that for all $0 < \epsilon \le \epsilon_0$, the function $\R^2 \times \R^2 \ni (k,X) \mapsto (a_\bullet \widetilde\#_c b_\bullet)_\epsilon(k,X) \in \cE$ is $\tau$-equivariant in the $k$ variable and $J\mathbb L$-periodic in the $X$ variable.  Thus 
$(a_\bullet \widetilde\#_c b_\bullet)_\bullet \in S^{\omega_1\omega_2,\epsilon_0}_{\rm twi}(\cE)$.

\medskip

When $a\in S_{\twi}^{\omega_1}(\mathcal{E}_1)$ and $b\in S_{\twi}^{\omega_2}(\mathcal{E}_2)$, we can apply the change of variables $(k_j,X_j) \mapsto (k_j,\epsilon X_j)$ to the integral representation~\eqref{eq:TWQ} to get
\begin{align*}
   (a_\bullet \widetilde\#_c b_\bullet)_\epsilon(k,X) 
   = \frac 1{\pi^4} \int_{(\R^2)^4} & e^{2i(k_2\cdot X_1-k_1 \cdot X_2)}
 \cT_{g(\epsilon) X_1} a(k+k_1,X+\epsilon X_1)  \cT_{g(\epsilon) (X_2-X_1)} \\ & \qquad  
  b(k+k_2,X+\epsilon X_2)  \cT_{-g(\epsilon) X_2} \, dk_1 \, dX_1 \, dk_2 \, dX_2,
\end{align*}
with 
$$
g(\epsilon):=\epsilon c(\epsilon) = 1-\sqrt{1-\epsilon^2}
= \sum_{n=1}^{+\infty} (-1)^{n+1} \left( \begin{array}{c} 1/2 \\ n \end{array} \right) \epsilon^{2n} = \frac{\epsilon^2}2 + {\mathcal O}(\epsilon^4).  
$$
Since
$$
\forall x_0,k,X \in \R^2, \quad \cT_{x_0} = e^{-x_0 \cdot \partial_x} \quad \mbox{and} \quad 
\cT_{x_0} d(k,X) \cT_{-x_0} = e^{-x_0 \cdot {\rm ad}_{\partial_x}}d(k,X) , 
$$
we get the following final expression of the twisted Moyal product
\begin{align*}
     (a_\bullet \widetilde\#_c b_\bullet)_\epsilon(k,X) 
   = \frac 1{\pi^4} \int_{(\R^2)^4} & e^{2i(k_2\cdot X_1-k_1 \cdot X_2)}
 e^{-\epsilon c(\epsilon)X_1 \cdot {\rm ad}_{\partial_x}} a(k+k_1,X+\epsilon X_1)   \\ & \qquad  
 e^{-\epsilon c(\epsilon)X_2 \cdot {\rm ad}_{\partial_x}} b(k+k_2,X+\epsilon X_2)  \, dk_1 \, dX_1 \, dk_2 \, dX_2.
\end{align*}
Reasoning as in the proof of the asymptotic expansion of $(a\# b)_\epsilon$ for $\epsilon$-independent symbols $a \in S^{\omega_1,1}(\C)$ and $b \in S^{\omega_2,1}(\C)$ (see e.g. \cite[Theorem 4.17]{zworski2012semiclassical}), and using the expansion 
$$
\forall x_0,k,X \in \R^2, \quad \cT_{x_0} d(k,X) \cT_{-x_0} = \sum_{\alpha \in \N^2} \frac{(-1)^{|\alpha|}}{\alpha !} x_0^\alpha \,  {\rm ad}_{\partial_x}^\alpha d(k,X),
$$
we obtain \eqref{eq:48} with
\begin{align*}
    d_0(k,X) :&= \frac 1{\pi^4} \int_{(\R^2)^4}  e^{2i(k_2\cdot X_1-k_1 \cdot X_2)}
  a(k+k_1,X) b(k+k_2,X)  \, dk_1 \, dX_1 \, dk_2 \, dX_2 \\
  &= a(k,X) b(k,X), \\
   d_1(k,X) :&= \frac 1{\pi^4} \int_{(\R^2)^4}  e^{2i(k_2\cdot X_1-k_1 \cdot X_2)} \bigg(X_1 \cdot \partial_X a(k+k_1,X) b(k+k_2,X) \\
   & \qquad\qquad\qquad\qquad
   + a(k+k_1,X) X_2 \cdot \partial_X b(k+k_2,X) \bigg) \,  dk_1 \, dX_1 \, dk_2 \, dX_2   \\
  &= \frac{i}{2}\{ a, b \} (k,X),
 \\
  d_2(k,X) :&= \frac 1{\pi^4} \int_{(\R^2)^4}  e^{2i(k_2\cdot X_1-k_1 \cdot X_2)} \bigg( \frac 12 \sum_{\mu,\nu=1}^2  X_1^\mu  X_1^\nu  \partial_{X^\mu X^\nu}a(k+k_1,X) \, b(k+k_2,X)    \\
 &\hspace{25mm} +   \frac 12 \sum_{\mu,\nu=1}^2   X_2^\mu  X_2^\nu  a(k+k_1,X) \, \partial_{X^\mu X^\nu} b(k+k_2,X)   \\
  &\hspace{25mm} +  \sum_{\mu,\nu=1}^2   X_1^\mu  X_2^\nu  \partial_{X^\mu} a(k+k_1,X) \, \partial_{X^\nu} b(k+k_2,X) \\
  &\hspace{25mm} - \frac 12    X_1 \cdot {\rm ad}_{\partial_x} a(k+k_1,X) b(k+k_2,X)  \\
  & \hspace{25mm}  
  - \frac 12   a (k+k_1,X) X_2 \cdot {\rm ad}_{\partial_x} b(k+k_2,X) \bigg)  \,  dk_1 \, dX_1 \, dk_2 \, dX_2   \\
    &= -\frac{1}{8}\{a,b\}_2(k,X)+\frac{i}{4}\Big(\partial_k a (k,X) \, {\rm ad}_{\partial_x}b(k,X)-{\rm ad}_{\partial_x}a (k,X) \; \partial_k b(k,X)\Big),
 \end{align*}
and so on. This completes the proof.
\end{proof}

\begin{corollary}[Inverse of twisted elliptic symbols]\label{cor:ellip}
Let $0<\epsilon_0 \le 1$. Let $a_\bullet\in S_{\twi}^{1,\epsilon_0}(\mathcal{B}(L^2_{\rm per}))$ be elliptic. Then there exist $0<\epsilon_0'\le \epsilon_0$ and $b_\bullet\in S_{\twi}^{1,\epsilon_0'}(\mathcal{B}(L^2_{\rm per}))$ such that 
\begin{align}
\mathrm{Op}^c(a_\epsilon)\mathrm{Op}^c(b_\epsilon)=\mathrm{Op}^c(b_\epsilon)\mathrm{Op}^c(a_\epsilon)=\1_{\cH}.
\end{align}
\end{corollary}
\begin{proof}
From \eqref{eq:N<N-twi}, it is easy to see that $a_{\bullet}^c\in S^{1,\epsilon_0}_{\tau}(\cB(L^2_{\rm per}))$. Then from Proposition \ref{prop:ellip}, we infer that there exists $0 < \epsilon_0'<\epsilon_0$ and a symbol $b^c_\bullet\in S^{1,\epsilon_0'}_{\tau}(\mathcal{B}(L^2_{\rm per}))$ such that
\begin{align*}
a_\epsilon^c\# b_\epsilon^c=\1_{L^2_{\rm per}}.
\end{align*}
From \eqref{eq:Weyl-product}, we infer
\begin{align*}
  0=  \px(a_\epsilon^c\# b_\epsilon^c)=(\px a_\epsilon^c)\# b_\epsilon^c+a_\epsilon^c\# (\px b_\epsilon^c).
\end{align*}
Thus, by using \eqref{eq:derivative-tau},
\begin{align*}
    \px b_\epsilon^c=-b_\epsilon^c\#(\px a_\epsilon^c)\# b_\epsilon \in S^{1,\epsilon_0'}_{\tau}(\mathcal{B}(L^2_{\rm per})).
\end{align*}
Analogously, we know that for any $\gamma\in \N^2$, $\px^\gamma b_\bullet\in S^{1,\epsilon_0'}_{\tau}(\mathcal{B}(L^2_{\rm per}))$. This and \eqref{eq:partial} imply that for any $\alpha,\beta,\gamma\in \N^2$,
\begin{align*}
    \cN^{1,\epsilon_0'}_{\alpha,\beta,\gamma}(b_\epsilon)<\infty.
\end{align*}

To show $b_\bullet\in S^{1,\epsilon_0'}_{\twi}$, it remains to show that $b_{\epsilon}$ is $J\mathbb{L}$-periodic w.r.t. $X\in \R^2$. Now we define for any $R\in J\mathbb{L}$,
\begin{align*}
 a_\epsilon^R:=a_\epsilon(k,X+R)  ,\quad b_\epsilon^R:=b_\epsilon(k,X+R).
\end{align*}
As $a_\bullet\in S_{\twi}^{1,\epsilon_0}(\mathcal{B}(L^2_{\rm per}))$, we know $a^R_\epsilon(k,X)=a_\epsilon(k,X)$. Then from \eqref{eq:TWQ}.
\begin{align*}
  \1_{L^2_{\rm per}}= (a_\epsilon\widetilde{\#}_{c} b_\epsilon)(k,X+R)= (a_\epsilon^R\widetilde{\#}_{c} b_\epsilon^R)(k,X)=(a_\epsilon \widetilde{\#}_{c}b_\epsilon^R)(k,X).
\end{align*}
Thus,
\begin{align*}
    b_\epsilon^R=b_\epsilon.
\end{align*}
This shows that $b_{\epsilon}$ is $J\mathbb{L}$-periodic w.r.t. $X\in \R^2$. Now we show that $b_\bullet\in S^{1,\epsilon_0'}_{\twi}(\mathcal{B}(L^2_{\rm per}))$. This ends the proof.
\end{proof}

To prove Theorem \ref{DoS-TBG}, we will use the following corollary of Lemma \ref{trace}.
\begin{corollary}\label{cor:trace-comp}
Let $0 < \epsilon_0 \le 1$, $a_\bullet\in S^{1,\epsilon_0}_{\twi}(\mathfrak{S}_1(L^2_{\rm per}))$, and $\chi_N$ satisfying \begin{align}
    & 0\leq \chi_N(X)\leq 1 \quad \mbox{ for all } X \in \R^2, \label{eq:chi_1'} \\
    &  \chi_N(X)=1 \textrm{ if }\;X\in NJ\Omega, \quad \textrm{and}\quad \chi_N(X)=0 \textrm{ if }\;X\notin (N+1)J\Omega, \label{eq:chi_2'} \\
    &  \sup_{N \in \N} \|\partial^\alpha \chi_N \|_{L^\infty} < +\infty \quad \mbox{ for all  } \alpha\in\mathbb{N}^2. \label{eq:chi_3'}
\end{align}
Then there exists $C\in \R_+$ such that for all $N\in \N^*$,
\begin{align}\label{eq:trace2}
\sup_{\epsilon\in (0,\epsilon_0]}\|\mathrm{Op}^c_\epsilon(a_\epsilon)\mathrm{Op}_\epsilon(\chi^2_N)\|_{\mathfrak{S}_1(\mathcal{H})} \leq C \, N^2,
\end{align}
and 
\begin{align}\label{eq:trace3}
  \MoveEqLeft  \epsilon^2\Tr_{\mathcal{H}}\left(\mathrm{Op}^c_\epsilon(a_\epsilon)\mathrm{Op}_\epsilon(\chi^2_N) \right) \notag\\
  &=\frac{1}{ (2\pi)^2}\!\!\!\!\int\limits_{(N+2)J\Omega}\!\!\!\!\int_{\Omega^*} \Tr_{L^2_{\rm per}}\left(\left(a_\epsilon\widetilde{\#}_c\chi^2_N\right)_\epsilon(k,X)\right)dkdX+\mathcal{O}(\epsilon^\infty N^2).
\end{align}
In addition, if $a\in S^{1}_{\twi}(\mathfrak{S}_1(L^2_{\rm per}))$, then 
\begin{align}\label{eq:trace4}
    \frac{\epsilon^2}{|NJ\Omega|}\Tr_{\mathcal{H}}\left(\mathrm{Op}^c_\epsilon(a)\mathrm{Op}_\epsilon(\chi^2_N) \right) \mathop{\longrightarrow}_{N \to \infty} \frac{1}{ (2\pi)^2} \fint_{J\Omega} \int_{\Omega^*} \Tr_{L^2_{\rm per}}\left(a(k,X)\right)dkdX + \mathcal{O}(\epsilon^\infty).
\end{align}
\end{corollary}
\begin{proof}
{\bf Estimate \eqref{eq:trace2}.} Recall that $a_\epsilon^c(k,X)=\mathcal{T}_{c(\epsilon)X} a_\epsilon\mathcal{T}_{c(\epsilon)X}^{-1}$. In addition, in the sense of operators on $L^2_{\rm per}$, it holds $\chi_N(X)=\mathcal{T}_{c(\epsilon)X} \,\chi_N(X) \,\mathcal{T}_{c(\epsilon)X}^{-1}$.  By Proposition \ref{prop:PS-asymp} we have $(a_\epsilon^c\widetilde{\#}_{c}\chi^2_N)_\epsilon\in S_{\twi}^{1,\epsilon_0}(\mathfrak{S}_{1}(L^2_{\rm per}))$, i.e.,
\begin{align}\label{eq:5.34}
    \sup_{\epsilon\in (0,\epsilon_0]}\|\partial_k^\alpha \partial_X^\beta (a_\epsilon^c\# \chi^2_N)\|_{\mathfrak{S}_1(L^2_{\rm per})}\leq C_{\alpha,\beta}<\infty.
\end{align}
As $\Supp(\chi_N)\subset (N+1)J\Omega$, according to Lemma \ref{lem:weyl-compact}, for any $k\in \R^2$ and $X\in \R^2\setminus(N+2)J\Omega$, we also have that for any $j\in \N$,
\begin{align}\label{eq:Achi}
   \sup_{\substack{k\in \R^2\\\epsilon\in (0,\epsilon_0]}}\|\partial_k^\alpha \partial_X^\beta (a_\epsilon^c\# \chi^2_N)_\epsilon(k,X)\|_{\mathfrak{S}_1(L^2_{\rm per})}\leq C'_{\alpha,\beta,j} \frac{\epsilon^j}{\dist(X,(N+1)J\Omega)^j}.
\end{align}
Thus Assumption \ref{ass:A-trace} is satisfied. Then from \eqref{eq:5.25} we get \eqref{eq:trace2}. 

\medskip

{\bf Estimate \eqref{eq:trace3}.}  On the other hand, from \eqref{eq:trace1}, \eqref{eq:trace11} and \eqref{eq:Achi}, we infer that
\begin{align*}
\MoveEqLeft    \epsilon^2\Tr_{\mathcal{H}}(\Op(a_\epsilon)\Op(\chi_N^2)) \\
    &=\frac{1}{ (2\pi)^2}\int_{\mathbb{R}^2}\int_{\Omega^*} \Tr_{L^2_{\rm per}}\Big(a_\epsilon\#\chi_N^2\Big)dkdX+\mathcal{O}(\epsilon^\infty N^2)\\
    &=\frac{1}{ (2\pi)^2}\int_{(N+2)J\Omega)}\int_{\Omega^*} \Tr_{L^2_{\rm per}}\Big(a_\epsilon\#\chi_N^2\Big)dkdX+\mathcal{O}(\epsilon^\infty N^2).
\end{align*}
Now \eqref{eq:trace3} follows directly from the fact that 
\begin{align*}
    \Tr_{L^2_{\rm per}}[a^c_\epsilon\# \chi_N]= \Tr_{L^2_{\rm per}}[\mathcal{T}_{c(\epsilon)X}^{-1} a_\epsilon\widetilde{\#}_{c} \chi_N \mathcal{T}_{c(\epsilon)X}]=\Tr_{L^2_{\rm per}}[a_\epsilon\widetilde{\#}_{c} \chi_N].
\end{align*}
\medskip

{\bf Estimate \eqref{eq:trace4}.}  It remains to show \eqref{eq:trace4}. Notice that by the definition of $\chi_N$, for any $\alpha\in \N^2\setminus\{0\}$, we have
\begin{align*}
    \Supp(\partial_X^\alpha\chi_N^2)\subset \overline{[(N+1)J\Omega]\setminus [NJ\Omega]}.
\end{align*}
This and Proposition \ref{prop:PS-asymp} imply that there exists a sequence $a_{j} \in S^{1}_{\twi}(\mathfrak{S}(L^2_{\rm per}))$ satisfying $\Supp(a_j)\subset \overline{[(N+1)J\Omega]\setminus[NJ\Omega]}$ for any $j\geq 1$, such that
\begin{align*}
    \left(a\widetilde{\#}_c\chi_N^2\right)_\epsilon(k,X)= a(k,X)\chi_N^2(X) + \sum_{j=1}^{n} \epsilon^j a_j(k,X)+ \mathcal{O}_{S^{1,\epsilon_0}_{\twi}(\mathfrak{S}_1(L^2_{\rm per}))}(\epsilon^{n+1}).
\end{align*}
Thus, thanks to the $J\mathbb{L}$-periodicity of $a_j$ w.r.t. $X$, from \eqref{eq:trace3} we infer that
\begin{align}\label{eq:trace-5}
    \MoveEqLeft \limsup_{N\to +\infty}  \frac{\epsilon^2}{|NJ\Omega|} \int_{(N+2)J\Omega}\int_{\Omega^*}\Tr_{L^2_{\rm per}}\left(\left(a\widetilde{\#}_c\chi_N^2\right)_\epsilon\right)(k,X)dkdX\notag\\
    &=\limsup_{N\to +\infty}  \frac{\epsilon^2}{|NJ\Omega|} \int_{(N+2)J\Omega}\int_{\Omega^*}\Tr_{L^2_{\rm per}}\left(a(k,X)\chi_N^2(X)\right)dkdX \notag\\
    &\quad+\limsup_{N\to +\infty}  \sum_{j=1}^{n}\frac{\epsilon^{2+j}}{|NJ\Omega|} \int_{\overline{[(N+1)J\Omega]\setminus [NJ\Omega]}}\int_{\Omega^*}\Tr_{L^2_{\rm per}}[a_j]dkdX+ \mathcal{O}(\epsilon^{n+1})\notag\\
    &=\fint_{J\Omega}\int_{\Omega^*}\Tr_{L^2_{\rm per}}(a(k,X))dkdX+ \mathcal{O}(\epsilon^{n+1}). 
\end{align}
and analogously for the limit inferior. This proves \eqref{eq:trace4}. Hence the corollary.
\end{proof}

\section{Proof of Theorem~\ref{DoS-TBG}}{}\label{sec:red-TBG}

In order to study the density of state of $H_{d,\theta}$, it is convenient to consider the cutoff operator $\chi_N$ in \eqref{eq:DoS_H} as a pseudodifferential operator and rewrite \eqref{eq:DoS_H} as
\begin{align} \label{eq:DoS_H'}
    \Trb[f(H_{d,\theta})]&=\lim_{N\to +\infty}\frac{\epsilon^2}{|NJ\Omega|}\Tr_{L^2(\mathbb{R}^3)} \left( \chi_N(\epsilon \cdot) f(H_{d,\theta})\chi_N(\epsilon\cdot) \right),
\end{align}
where $\chi_N$ is now a function of the mesoscopic variable $X$ satisfying \eqref{eq:chi_1'}-\eqref{eq:chi_3'}.
We indeed have for any bounded function $\phi \in C^\infty(\R^2)$, 
$$
\phi(\epsilon \cdot) = \cU^{-1}{\rm Op}_\epsilon(\phi) \cU ,
$$
where on the LHS, $\phi(\epsilon\cdot)$ denotes the multiplication operator on $L^2(\R^2;\C)$ by the function $(x,z) \mapsto \phi(x)$, and on the RHS, $\phi$ denotes the operator-valued symbol $(k,X) \mapsto \phi(X)$, where for each $X \in \R^2$, $\phi(X)$ is the multiplication operator on $L^2_{\rm per}$ by the constant $\phi(X)$.

\medskip

We first prove that  $\frac{\epsilon^2}{|NJ\Omega|}\Tr_{L^2(\mathbb{R}^3)}[\chi_N(\epsilon \cdot)f(H_{d,\theta})\chi_N(\epsilon \cdot)]$ is uniformly bounded w.r.t. $N$ and $\epsilon\in (0,\epsilon_0]$, which verifies that the LHS of \eqref{eq:main2} is well-defined (since $f$ is nonnegative and $H_{d,\theta}$ is either periodic or ergodic) and finite for any $\epsilon\in (0,\epsilon_0]$. Then we prove the asymptotic expansion \eqref{eq:main2}.

Throughout this section, we use the order function $\omega(k,X):=1+|k|^2$. We recall that $\epsilon:=\sin(\theta/2)$.

\subsection{Finiteness of the LHS of \texorpdfstring{\eqref{eq:main2}}{}}\label{sec:6.1}
In this subsection, we are going to prove the following.
\begin{proposition}\label{th:sigma1}
With the assumptions of Theorem \ref{DoS-TBG}, there exists $0<\epsilon_0 \le 1$ such that for any $N\in \N^*$ and $\epsilon\in (0,\epsilon_0]$, we have $\chi_N(\epsilon\cdot) f(H_{d,\theta})\chi_N(\epsilon \cdot)\in\mathfrak{S}_1(L^2(\mathbb{R}^3))$. Moreover, there exists a constant $C\in \R_+$ such that for any $N\in \mathbb{N}^2$,
\begin{align}\label{eq:unibound}
    \sup_{\epsilon\in (0,\epsilon_0]} \frac{\epsilon^2}{|NJ\Omega|}\|\chi_N(\epsilon\cdot)f(H_{d,\theta})\chi_N(\epsilon\cdot)\|_{\mathfrak{S}_1(L^2(\mathbb{R}^3))}\leq C.
\end{align}
\end{proposition}
The proof of Proposition \ref{th:sigma1} is postponed until the end of Section \ref{sec:6.1}.

As the Bloch transform operator $\mathcal{U}$ is unitary, it follows from \eqref{eq:Hdeps=Opceps} and \eqref{eq:HF_formula} that
\begin{align}\label{eq:Ztrans1}
    \mathcal{U}f(H_{d,\theta})\chi_N(\epsilon \cdot)\mathcal{U}^{-1}&= f\left(\mathcal{U}H_{d,\theta}\mathcal{U}^{-1}\right)\mathcal{U}\chi_N(\epsilon \cdot)\mathcal{U}^{-1}=f(\Opc(h_{d,0}))\Op(\chi_N),
\end{align}
and
\begin{align}\label{eq:Hd0-chi}
\|\chi_N(\epsilon \cdot)f(H_{d,\theta})\chi_N(\epsilon \cdot)\|_{\mathfrak{S}_1(L^2(\mathbb{R}^3))}&\leq \|f(H_{d,\theta})\chi_N(\epsilon \cdot)\|_{\mathfrak{S}_1(L^2(\mathbb{R}^3))}\notag\\
&= \|f(\Opc(h_{d,0}))\Op(\chi_N)\|_{\mathfrak{S}_1(\mathcal{H})}.
\end{align}

The main step is to rewrite $f(\Opc(h_{d,0}))$ by using \eqref{eq:HF_formula} and the operator $h_{d,0,f}(k,X)$ defined by \eqref{eq:h-t}. To do so, we first need to prove that $\zeta \mapsto (\Opc(h_{d,0,f})-\zeta)^{-1}$ is holomorphic on $\Supp(\widetilde{f})$. Since $\Opc(h_{d,0,f})$ is self-adjoint, it suffices to show that $(\Opc(h_{d,0,f})-\zeta)^{-1}$ is uniformly bounded on $\Supp(\widetilde{f})$. This step is inspired by \cite{dimassi1993developpements}.
\begin{lemma}\label{lem:inverse}
Let $\epsilon_0 > 0$ be small enough. Then, for all $\zeta\in \Supp(\widetilde{f})$, there exits an operator-valued symbol $g_\bullet(\zeta)\in S^{\omega,\epsilon_0}_{\twi}(\mathcal{B}(L^2_{\rm per},H^2_{\rm per}))\bigcap S^{1,\epsilon_0}_{\twi}(\mathcal{B}(L^2_{\rm per}))$ such that
\begin{align*}
    (\Opc(h_{d,0,f})-\zeta)^{-1} =  \Opc(g_\epsilon(\zeta)),
\end{align*}
and
\begin{align}\label{eq:ellep}
    \sup_{\substack{\epsilon\in (0,\epsilon_0]\\\zeta\in\Supp(\widetilde{f})}}\|\Opc(g_\epsilon(\zeta))\|_{\mathcal{B}(\mathcal{H})}=\sup_{\substack{\epsilon\in (0,\epsilon_0]\\\zeta\in\Supp(\widetilde{f})}}\|(\Opc(h_{d,0,f})-\zeta)^{-1}\|_{\mathcal{B}(\mathcal{H})} < \infty.
\end{align}
\end{lemma}
\begin{proof}
Thanks to Lemma \ref{lem:4.2}, we know that $h_{d,0,f} \in S^{\omega}_\twi(
\mathcal{B}(H^2_{\rm per},L^2_{\rm per}))$. Let $g_0(\zeta)=(h_{d,0,f}-\zeta)^{-1}$. According to Lemma \ref{lem:h-t-inversible}, we have that 
\begin{align*}
    g_0(\zeta)\in S^{\omega}_{\twi}(\mathcal{B}(L^2_{\rm per},H^2_{\rm per}))\bigcap S^{1}_{\twi}(\mathcal{B}(L^2_{\rm per})).
\end{align*}
For all $\zeta\in \Supp(\widetilde{f})$,
\begin{align*}
    \left((h_{d,0,f}-\zeta)\widetilde{\#}_{c} g_0(\zeta)\right)_\epsilon(k,X)=\1_{L^2_{\rm per}}-\epsilon R_\epsilon(\zeta)(k,X)
\end{align*}
with $R_\bullet(\zeta)\in S^{\omega^2,\epsilon_0}_\twi(\mathcal{B}(L^2_{\rm per}))$.  On the other hand, by $\tau$-equivariance of $R_\epsilon$, the norm $\|R_\epsilon(\zeta)\|_{\mathcal{B}(L^2_{\rm per})}$ is $\mathbb{L}^*$-periodic. This and \eqref{eq:derivative-tau} imply $R_\bullet(\zeta)\in S^{1,\epsilon_0}_\twi(\mathcal{B}(L^2_{\rm per}))$ for all $\zeta\in \Supp(\widetilde{f})$. Furthermore, from \eqref{eq:inverse-zeta1}-\eqref{eq:inverse-zeta2} we infer that for any $\alpha,\beta,\gamma\in \N^2$,
\begin{align*}
   \sup_{\zeta\in \Supp(\widetilde{f})}\cN_{\alpha,\beta,\gamma}^{1,\epsilon_0,\cB(L^2_{\rm per})}(R_{\epsilon})<\infty.
\end{align*} 

Now, Corollary \ref{cor:ellip} shows that there exists a symbol $g_{1,\bullet}(\zeta)\in S^{1,\epsilon_0}_{\twi}(\mathcal{B}(L^2_{\rm per}))$ such that
\begin{align*}
    \Opc(g_{1,\epsilon}(\zeta))=(\1_{\mathcal{H}}-\epsilon\Opc(R_\epsilon(\zeta)))^{-1}.
\end{align*}
Let $g_\epsilon(\zeta)=g_0(\zeta)\widetilde{\#}_cg_{1,\epsilon}(\zeta)$. Then $g_\bullet(\zeta)\in S^{\omega,\epsilon_0}_{\twi}(\mathcal{B}(L^2_{\rm per},H^2_{\rm per}))\bigcap S^{1,\epsilon_0}_{\twi}(\mathcal{B}(L^2_{\rm per}))$, and \eqref{eq:ellep} follows from the Calderon-Vaillancourt theorem (i.e., Theorem \ref{th:caldVa}). This gives the right inverse:
\begin{align*}
    \Opc(h_{d,0,f}-\zeta)\Opc(g_\epsilon(\zeta))=\1_{L^2_{\rm per}}.
\end{align*}
By the same argument, we can also find a left inverse
\begin{align*}
 \Opc(g'_\epsilon(\zeta))\Opc(h_{d,0,f}-\zeta)=\1_{H^2_{\rm per}}.
\end{align*}
Thus,
\begin{align*}
    g_\epsilon(\zeta)=g'_\epsilon(\zeta) \widetilde{\#}_c( h_{d,0,f}-\zeta)\widetilde{\#}_cg_\epsilon(\zeta)=g_\epsilon'(\zeta).
\end{align*}
\end{proof}
Then thanks to Lemma \ref{lem:inverse}, $(\zeta-\Opc(h_{d,0,f}))^{-1}$ is holomorphic on $\zeta\in \Supp(\widetilde{f})$. Therefore we have
\begin{align*}
    \int_{\mathbb{C}} \overline{\partial}\widetilde{f}(\zeta)(\zeta-\Opc(h_{d,0,f}))^{-1}dL(\zeta)=-\int_{\mathbb{C}} \widetilde{f}(\zeta)\overline{\partial}(\zeta-\Opc(h_{d,0,f}))^{-1}dL(\zeta)=0.
\end{align*}
Thus from \eqref{eq:HF_formula},
    \begin{align}\label{eq:holom}
   \MoveEqLeft f(\Opc(h_{d,0}))=-\frac{1}{\pi}\int_{\mathbb{C}} \overline{\partial}\widetilde{f}(\zeta)[(\zeta-\Opc(h_{d,0}))^{-1}- (\zeta-\Opc(h_{d,0,f}))^{-1}]dL(\zeta)\notag\\
    &=\frac{1}{\pi}\int_{\mathbb{C}} \overline{\partial}\widetilde{f}(\zeta)(\zeta-\Opc(h_{d,0}))^{-1}(\Opc(h_{d,0,f})-\Opc(h_{d,0}))(\zeta-\Opc(h_{d,0,f}))^{-1}dL(\zeta).
\end{align}

Now we turn to the
\begin{proof}[Proof of Proposition \ref{th:sigma1}]
Inserting \eqref{eq:holom} into \eqref{eq:Hd0-chi} and using the fact that $1\leq \chi_N\leq 1$ (see \eqref{eq:chi_1}), we obtain that
\begin{align*}
\MoveEqLeft\sup_{\epsilon\in (0,\epsilon_0]}\|\chi_N(\epsilon\cdot)f(H_{d,\theta})\chi_N(\epsilon\cdot)\|_{\mathfrak{S}_1(L^2(\R^3))}\\
    &\leq \sup_{\epsilon\in (0,\epsilon_0]}\|f(\Opc(h_{d,0}))\Op(\chi_N)\|_{\mathfrak{S}_1(\cH)}\\
    &\leq C\sup_{\substack{\zeta\in \Supp(\widetilde{f})\\ \epsilon\in (0,\epsilon_0]}}\|(\Opc(h_{d,0,f})-\Opc(h_{d,0}))\Op(g_\epsilon(\zeta))\Op(\chi_N)\|_{\mathfrak{S}_1(\mathcal{H})},
\end{align*}
since $\sup_{\epsilon\in (0,\epsilon_0]}\|(\zeta-\Opc(h_{d,0}))^{-1}\|_{\mathcal{B}(\mathcal{H})}\leq |\Im \zeta|^{-1}$ and $|\overline{\partial}\widetilde{f}(\zeta)|=\mathcal{O}(|\Im\zeta|^\infty)$.  

From Lemma \ref{cor:trace} and Lemma \ref{lem:inverse}, it is easy to see that 
\begin{align*}
\big((h^c_{d,0,f}-h^c_{d,0})\#g_\bullet(\zeta)\big)_\bullet \in S^{1,\epsilon_0}_{\twi}(\mathfrak{S}_1(L^2_{\rm per})). 
\end{align*}
Then from Corollary \ref{cor:trace-comp} (with $\chi_N^2$ being replaced by $\chi_N$ in this result), we deduce that there exists a constant $C\in \R_+$ such that for any $N\in \N^*$,
\begin{align}\label{eq:Hdep-chiN}
  \MoveEqLeft \sup_{\epsilon\in (0,\epsilon_0]}\frac{\epsilon^2}{|NJ\Omega|}\|\chi_N(\epsilon \cdot)f(H_{d,\theta})\chi_N(\epsilon\cdot)\|_{\mathfrak{S}_1(L^2(\R^3))}\leq \frac{C|(N+2)J\Omega|}{3|NJ\Omega|}\leq C.
\end{align}
Hence Proposition \ref{th:sigma1}.
\end{proof}

\subsection{Asymptotic expansion for DoS}\label{sec:5.2}
From \eqref{eq:Hdep-chiN}, we infer that $\frac{\epsilon^2}{|NJ\Omega|}f(H_{d,\theta})\chi_N(\epsilon\cdot) \in \mathfrak{S}_1(L^2(\R^3))$ for any $\epsilon\in (0,\epsilon_0]$. As $0\leq \chi_N\leq 1$, we have
\begin{align}\label{eq:chiNfhchiN}
 \frac{\epsilon^2}{|NJ\Omega|}\Tr_{L^2(\R^3)}[\chi_N(\epsilon\cdot)f(H_{d,\theta})\chi_N(\epsilon\cdot)]&=\frac{\epsilon^2}{|NJ\Omega|}\Tr_{L^2(\R^3)}[f(H_{d,\theta})\chi_N^2(\epsilon\cdot)]\notag\\
    &=\frac{\epsilon^2}{|NJ\Omega|}\Tr_{\cH}[f(\Opc(h_{d,0}))\Op(\chi_N^2)].
\end{align}
In this subsection, we are going to study the asymptotic expansion of $f(\Opc(h_{d,0}))\Op(\chi_N^2)$ in $\mathfrak{S}_1(\cH)$:
\begin{lemma}\label{lem:a_j}
There exist a sequence $(f_{d,j})$ in $S_{\twi}^{1}(\mathfrak{S}_1(L^2_{\rm per}))$ and a constant $C\in \R_+$ such that for all $n\in \N^*$, $j\in \N$ and all $\epsilon\in (0,1]$,
\begin{align}\label{eq:dec-fh}
    \frac{\epsilon^2}{|NJ\Omega|}\left\|f(\Opc(h_{d,0}))\Op(\chi_N^2)-\sum_{j=0}^n\epsilon^j \Opc(f_{d,j})\Op(\chi_N^2)\right\|_{\mathfrak{S}_1(\mathcal{H})}\leq C\frac{\epsilon^{n+1}}{|\Omega|}.
\end{align}
In particular, 
\begin{align}\label{eq:a0}
    f_{d,0}(k,X)=f(h_{d,0}(k,X)),
\end{align}

\begin{align}\label{eq:a1}
    f_{d,1}(\zeta)
    =- \frac 1 \pi \int_{\C} \overline{\partial}\widetilde{f}(\zeta ) \Big[
    - \frac i 2 \{(\zeta-h_{d,0})^{-1},(\zeta-h_{d,0})\}(\zeta-h_{d,0})^{-1} 
    \Big] dL(\zeta),
\end{align}

and
\begin{align}\label{eq:a2}
 \MoveEqLeft    f_{d,2}(k,X)=-\frac{1}{\pi}\int_\C \overline{\partial}\widetilde{f}(\zeta)  \Big[-\frac{1}{4}(\zeta-h_{d,0})^{-1}\{(\zeta-h_{d,0}),(\zeta-h_{d,0})^{-1}\}^2\notag\\
     &+\frac{1}{4}\{(\zeta-h_{d,0})^{-1},\{(\zeta-h_{d,0}),(\zeta-h_{d,0})^{-1}\}\}\notag\\
     &+\frac{1}{8}(\zeta-h_{d,0})^{-1}\{(\zeta-h_{d,0}),(\zeta-h_{d,0})^{-1}\}_2\notag\\
     &+\frac{i}{4}(\zeta-h_{d,0})^{-1}\Big(\partial_k h_{d,0}{\rm ad}_{\partial_x} (\zeta-h_{d,0})^{-1}-{\rm ad}_{\partial_x}h_{d,0} \partial_k (\zeta-h_{d,0})^{-1}\Big)\Big](k,X) \, dL(\zeta).
\end{align}
\end{lemma}

One may consider studying the asymptotic expansion $f(\Opc(h_{d,0}))$ directly by using the classical parametrix construction:
\begin{align*}
  \widetilde{f}(\zeta)(\zeta-\Opc(h_{d,0}))^{-1}=  \sum_{j=0}^n \epsilon^j \widetilde{f}(\zeta) \Opc(\widetilde{a}_{j}(\zeta))+\mathcal{O}_{\mathcal{B}(\mathcal{H})}(\epsilon^{n+1}).
\end{align*}
With this approach, we have $\widetilde{f}(\zeta)\widetilde{a}_{j}(\zeta)\in 
\cB(L^2_{\rm per})$, but we do not know a priori that $\widetilde{f}(\zeta)\widetilde{a}_{j}(\zeta)\in \mathfrak{S}_1(L^2_{\rm per})$. To circumvent this issue, we start from \eqref{eq:holom}: we study the asymptotic expansions of $(\zeta-\Opc(h_{d,0}))^{-1}$ and $(\zeta-\Opc(h_{d,0,f}))^{-1}$ in $\cB(L^2_{\rm per})$ separately, and use the fact that $h_{d,0}-h_{d,0,f}\in \mathfrak{S}_1(L^2_{\rm per})$ and Lemma \ref{cor:trace} to guarantee that the expansion holds in $\mathfrak{S}_1(\cH)$.
To prove Lemma~\ref{lem:a_j}, we will need the following result.

\begin{lemma}\label{lem:5.7}
For all $\zeta\in\Supp(\widetilde{f})\setminus\mathbb{R}$, there exist a sequence $(a_{j}(\zeta))_{j \in \N}$ and a sequence $(r_{n,\epsilon}(\zeta))_{n \in \N}$ such that
\begin{align}\label{eq:32'}
  \forall j,n \in \N, \quad  \overline{\partial}\widetilde{f}(\zeta) a_{j}(\zeta) \in S_{\twi}^{1}(\mathcal{B}(L^2_{\rm per})),\quad \overline{\partial}\widetilde{f}(\zeta)r_{n,\bullet}(\zeta)\in S_{\twi}^{1,\epsilon_0}(\epsilon,\mathcal{B}(L^2_{\rm per})),
\end{align}
uniformly in $\zeta$, and for any $n\geq 0$,
\begin{align}\label{eq:32}
   (\zeta-\Opc(h_{d,0}))^{-1}=  \sum_{j=0}^{n}\epsilon^j \Opc(a_{j}(\zeta))-\epsilon^{n+1}(\zeta-\Opc(h_{d,0}))^{-1}\Opc(r_{n,\epsilon}(\zeta)).
\end{align}
In particular, $a_{0}(\zeta)=(\zeta-h_{d,0})^{-1}$, $ a_{1}(\zeta)=-\frac{i}{2}\{(\zeta-h_{d,0})^{-1},\zeta-h_{d,0}\}(\zeta-h_{d,0})^{-1}$ and
\begin{align*}
    a_2(\zeta)&=-\frac{1}{4}(\zeta-h_{d,0})^{-1}\{(\zeta-h_{d,0}),(\zeta-h_{d,0})^{-1}\}^2\\
     &\quad+\frac{1}{4}\{(\zeta-h_{d,0})^{-1},\{(\zeta-h_{d,0}),(\zeta-h_{d,0})^{-1}\}\}\\
     &\quad+\frac{1}{8}(\zeta-h_{d,0})^{-1}\{(\zeta-h_{d,0}),(\zeta-h_{d,0})^{-1}\}_2\\
     &\quad+\frac{i}{4}(\zeta-h_{d,0})^{-1}\Big(\partial_k h_{d,0}{\rm ad}_{\partial_x} (\zeta-h_{d,0})^{-1}-{\rm ad}_{\partial_x}h_{d,0} \partial_k (\zeta-h_{d,0})^{-1}\Big).
\end{align*}
\end{lemma}
\begin{proof}
The proof of this lemma follows from the classical parametrix construction for $\tau$-equivariant symbols (see e.g. \cite[Lemma 1]{panati2003effective}).

Let $a_{0}(\zeta):=(\zeta-h_{d,0})^{-1}$. Using \eqref{eq:31a}, \eqref{eq:31b} and $|\overline{\partial}\widetilde{f}(\zeta)|=\mathcal{O}(|\Im \zeta|^\infty)$, we see that
\begin{align*}
    \overline{\partial}\widetilde{f}(\zeta) a_{0}(\zeta)\in S_{\twi}^{1}(\mathcal{B}(L^2_{\rm per}))\bigcap S^{\omega}_{\twi}(\mathcal{B}(L^2_{\rm per},H^2_{\rm per})),
\end{align*}
uniformly in $\zeta$. According to Proposition \ref{prop:PS-asymp}, we have
\begin{equation} \label{eq:rec0}
   (\zeta-h_{d,0})\widetilde{\#}_c a_{0}(\zeta) =\1_{L^2_{\rm per}}+\epsilon {r}_{0,\epsilon}(\zeta) = \1_{L^2_{\rm per}}+\epsilon p_1(\zeta) + {\mathcal O}_{S^{\omega,\epsilon_0}_{\twi}(\cB(L^2_{\rm per})}(\epsilon^2),
\end{equation}
with 
$$
p_1(\zeta):= \frac i2 \{(\zeta-h_{d,0}),(\zeta-h_{d,0})^{-1}\}\quad \mbox{and} \quad \overline{\partial}\widetilde{f}(\zeta){r}_{0,\bullet}(\zeta) \in S^{\omega^2,\epsilon_0}_{\twi}(\epsilon,\mathcal{B}(L^2_{\rm per})),
$$
uniformly in $\zeta$.
As the $\mathcal{B}(L^2_{\rm per})$ norm of the $\tau$-equivariance functions is $\mathbb{L}^*$-periodic, $\overline{\partial}\widetilde{f}(\zeta)\widetilde{r}_{0,\bullet}(\zeta)\in S^{1,\epsilon_0}_{\twi}(\epsilon,\mathcal{B}(L^2_{\rm per}))$. Applying the twisted Weyl operator ${\rm Op}_\epsilon^c$ to both sides of the first equality in \eqref{eq:rec0} yields
$$
(\zeta-{\rm Op}_\epsilon^c(h_{d,0})) \, {\rm Op}_\epsilon^c(a_0(\zeta)) = 1_\cH + \epsilon {\rm Op}_\epsilon^c(r_{0,\epsilon}(\zeta)).
$$
Thus, \eqref{eq:32'}-\eqref{eq:32} hold true for $j=n=0$. 

We now set $a_1(\zeta)=-(\zeta-h_{d,0})^{-1} p_1(\zeta)$. Using again Proposition \ref{prop:PS-asymp}, we get
\begin{align*}
   (\zeta-h_{d,0})\widetilde{\#}_c \Big( a_{0}(\zeta) + \epsilon a_1(\zeta) \Big) =\1_{L^2_{\rm per}}+\epsilon^2 {r}_{1,\epsilon}(\zeta) = \1_{L^2_{\rm per}}+\epsilon^2 p_2(\zeta) + {\mathcal O}_{S^{\omega,\epsilon_0}_{\twi}(\cB(L^2_{\rm per})}(\epsilon^3),
\end{align*}
with 
\begin{align*}
    p_2(\zeta)&=\frac{1}{4} \{(\zeta-h_{d,0}),(\zeta-h_{d,0})^{-1}\}^2 - \frac{1}{4} ,\{(\zeta-h_{d,0}),(\zeta-h_{d,0})^{-1}\}\} \\
    & \quad - \frac{1}{8} \{(\zeta-h_{d,0}),(\zeta-h_{d,0})^{-1}\}_2\\
     &\quad -\frac{i}{4} \Big(\partial_k h_{d,0}{\rm ad}_{\partial_x} (\zeta-h_{d,0})^{-1}-{\rm ad}_{\partial_x}h_{d,0} \partial_k (\zeta-h_{d,0})^{-1}\Big).
\end{align*}
and 
$$
\overline{\partial}\widetilde{f}(\zeta){r}_{1,\bullet}(\zeta) \in S^{\omega^2,\epsilon_0}_{\twi}(\epsilon,\mathcal{B}(L^2_{\rm per})),
$$
uniformly in $\zeta$. Reasoning as above, we obtain \eqref{eq:32'}-\eqref{eq:32} for $j,n \le 1$. Setting $a_2(\zeta)=-(\zeta-h_{d,0})^{-1}p_2(\zeta)$, and using once more Proposition \ref{prop:PS-asymp}, we see that 
\begin{align*}
   (\zeta-h_{d,0})\widetilde{\#}_c \Big( a_{0}(\zeta) + \epsilon a_1(\zeta) + \epsilon^2 a_1(\zeta) \Big) =\1_{L^2_{\rm per}}+\epsilon^3 {r}_{2,\epsilon}(\zeta),
\end{align*}
with 
$$
\overline{\partial}\widetilde{f}(\zeta){r}_{2,\bullet}(\zeta) \in S^{\omega^2,\epsilon_0}_{\twi}(\epsilon,\mathcal{B}(L^2_{\rm per})),
$$
uniformly in $\zeta$, which leads to \eqref{eq:32'}-\eqref{eq:32} for $j,n \le 2$. The result for any $n$ is obtained by induction.
\end{proof}

\begin{lemma}\label{lem:5.8}
For all $\zeta\in \Supp(\widetilde{f})$, there exist a sequence $(b_{j}(\zeta))_{j \in \N}$,  and a sequence $(\widetilde{r}_{n,\epsilon}(\zeta))_{n \in \N}$ such that
\begin{align}
    b_{j}(\zeta) \in S_{\twi}^{1}(\mathcal{B}(L^2_{\rm per})),\quad \widetilde{r}_{n,\bullet}(\zeta)\in S_{\twi}^{1,\epsilon_0}(\epsilon,\mathcal{B}(L^2_{\rm per}))
\end{align}
and for all $n\geq 0$,
\begin{align}
   (\zeta-\Opc(h_{d,0,f}))^{-1}=  \sum_{j=0}^{n}\epsilon^j \Op(b_{j}(\zeta))-\epsilon^{n+1}\Op(g_\epsilon(\zeta))\Opc(\widetilde{r}_{n}(\epsilon,\zeta)).
\end{align}
In particular, $b_{0}=(\zeta-h_{d,0,f})^{-1}$, $b_{1}=-\frac{i}{2}(\zeta-h_{d,0,f})^{-1}\{\zeta-h_{d,0,f},(\zeta-h_{d,0,f})^{-1}\}$ and
\begin{align*}
    b_2(\zeta)&=-\frac{1}{4}(\zeta-h_{d,0,f})^{-1}\{(\zeta-h_{d,0,f}),(\zeta-h_{d,0,f})^{-1}\}^2\\
     &\quad+\frac{1}{4}\{(\zeta-h_{d,0,f})^{-1},\{(\zeta-h_{d,0,f}),(\zeta-h_{d,0,f})^{-1}\}\}\\
     &\quad+\frac{1}{8}(\zeta-h_{d,0,f})^{-1}\{(\zeta-h_{d,0,f}),(\zeta-h_{d,0,f})^{-1}\}_2\\
     &\quad+\frac{i}{4}(\zeta-h_{d,0,f})^{-1}\Big(\partial_k h_{d,0,f}{\rm ad}_{\partial_x} (\zeta-h_{d,0,f})^{-1}-{\rm ad}_{\partial_x}h_{d,0,f} \partial_k (\zeta-h_{d,0,f})^{-1}\Big).
\end{align*}
\end{lemma}
\begin{proof}
We proceed as in the proof of Lemma \ref{lem:5.7}. The only difference is that instead of using $|\overline{\partial}\widetilde{f}(\zeta)|=\mathcal{O}(|\Im \zeta|^\infty)$, we use \eqref{h-elle}. 
\end{proof}

We now turn to the
\begin{proof}[Proof of Lemma \ref{lem:a_j}]
According to Lemma \ref{cor:trace}, 
\[
h_{d,0}-h_{d,0,f}\in S_{\twi}^{1}(\mathfrak{S}_1(L^2_{\rm per})).
\]
Let $\zeta\in \Supp(\widetilde{f})\setminus\mathbb{R}$. Using Lemma \ref{lem:5.7} and Lemma \ref{lem:5.8}, we can find a sequence $(d_{j}(\zeta))_{j \in \N}$ and two sequence $(\mathcal{R}_{1,n,\epsilon}(\zeta))_{n \in \N}$ and $(\mathcal{R}_{2,n,\epsilon}(\zeta))_{n \in \N}$ such that
\begin{align*}
    \overline{\partial}\widetilde{f}(\zeta )d_{j}(\zeta)\in S_{\twi}^{1}(\mathfrak{S}_1(L^2_{\rm per})) \quad\textrm{and}\quad \overline{\partial}\widetilde{f}(\zeta )\mathcal{R}_{1,n,\bullet},\; \overline{\partial}\widetilde{f}(\zeta )\mathcal{R}_{2,n,\bullet}\in  S_{\twi}^{1,\epsilon_0}(\epsilon,\mathfrak{S}_1(L^2_{\rm per})),
\end{align*}
uniformly in $\zeta$, and
\begin{align}\label{eq:38}
 \MoveEqLeft (\zeta-\Opc(h_{d,0}))^{-1}(\Opc(h_{d,0})-\Opc(h_{d,0}))(\zeta-\Opc(h_{d,0,f}))^{-1}\notag\\
 &=\sum_{j=0}^n\epsilon^j[\Opc(d_{j}(\zeta))+\epsilon^{n+1}[(\zeta-\Opc(h_{d,0}))^{-1}\Opc(\mathcal{R}_{1,n,\epsilon})+\Opc(\mathcal{R}_{2,n,\epsilon})].
\end{align}
Thanks to \eqref{eq:HF_formula}, from Corollary \ref{cor:trace-comp} we infer that there exists $C\in \R_+$ such that for all $N\in \N^*$,
\begin{align}\label{eq:39}
   \MoveEqLeft \sup_{\epsilon\in (0,\epsilon_0]} \frac{\epsilon^2}{|NJ\Omega|}\left\|\int_{\C} \overline{\partial}\widetilde{f}(\zeta-\Opc(h_{d,0}))^{-1}\Opc(\mathcal{R}_{1,n,\epsilon}) \Op(\chi_N^2) dL(\zeta)\right\|_{\mathfrak{S}_1(\mathcal{H})}\notag\\
   &\leq \sup_{\epsilon\in (0,\epsilon_0]} \frac{\epsilon^2}{|NJ\Omega|}\int_{\C} |\overline{\partial}\widetilde{f}(\zeta )||\Im\zeta|^{-1}\left\|\Opc(\mathcal{R}_{1,n,\epsilon})\Op(\chi_N^2) \right\|_{\mathfrak{S}_1(\mathcal{H})}dL(\zeta)\leq C.
\end{align}
Analogously, there exists $C\in \R_+$ such that for all $N\in \N^*$,
\begin{align}\label{eq:40}
    \sup_{\epsilon\in (0,\epsilon_0]} \frac{\epsilon^2}{|NJ\Omega|}\left\|\int_{\C} \overline{\partial}\widetilde{f}(\zeta )\Opc(\mathcal{R}_{2,n,\epsilon})\Op(\chi_N^2) dL(\zeta)\right\|_{\mathfrak{S}_1(\mathcal{H})}\leq C.
\end{align}
Let 
\begin{align}\label{eq:41}
    f_{d,j}:=\frac{1}{\pi}\int_{\C} \overline{\partial}\widetilde{f}(\zeta ) d_{j}(\zeta)dL(\zeta)\in S^{1,\epsilon_0}_{\rm twi}(\mathfrak{S}_{1}(L^2_{\rm per})).
\end{align}
Then \eqref{eq:dec-fh} follows from \eqref{eq:holom} and \eqref{eq:38}-\eqref{eq:41}.

\medskip

It remains to show \eqref{eq:a0}-\eqref{eq:a2}. According to Lemma \ref{lem:5.7}, Lemma \ref{lem:5.8} and \eqref{eq:38}, we know that
\begin{align*}
    d_{0}(\zeta)=(\zeta-h_{d,0})^{-1}(h_{d,0,f}-h_{d,0})(\zeta-h_{d,0,f})^{-1}=(\zeta-h_{d,0,f})^{-1}-(\zeta-h_{d,0})^{-1}.
\end{align*}
By using \eqref{eq:41} and the fact that $(\zeta-h_{d,0,f})^{-1}$ is holomorphic on $\Supp(\widetilde{f})$, we get $f_{d,0}(k,X)=f(h_{d,0}(k,X))$.

We turn to prove \eqref{eq:a1}. Notice that
\begin{align*}
    d_{1}(\zeta)&=-\frac{i}{2}\{(\zeta-h_{d,0})^{-1},\zeta-h_{d,0}\}(\zeta-h_{d,0})^{-1}(h_{d,0,f}-h_{d,0})(\zeta-h_{d,0,f})^{-1}\\
    &\quad -\frac{i}{2}(\zeta-h_{d,0})^{-1}(h_{d,0,f}-h_{d,0})(\zeta-h_{d,0,f})^{-1}\{\zeta-h_{d,0,f},(\zeta-h_{d,0,f})^{-1}\}\\
    &\quad+\frac{i}{2}\{(\zeta-h_{d,0})^{-1},(h_{d,0,f}-h_{d,0})\}(\zeta-h_{d,0,f})^{-1}\\
    &\quad+\frac{i}{2}\{(\zeta-h_{d,0})^{-1}(h_{d,0,f}-h_{d,0}),(\zeta-h_{d,0,f})^{-1}\}.
\end{align*}
Let $A(\zeta):=(\zeta-h_{d,0})$ and $B(\zeta):=(\zeta-h_{d,0,f})$. Then using the fact that $h_{d,0,f}-h_{d,0}=A(\zeta)-B(\zeta)$ and $\partial B^{-1}(\zeta)=-B^{-1}\partial B(\zeta) B^{-1}(\zeta)$ with $\partial:=(\partial_k,\partial_X)$, we obtain
\begin{align*}
    2i d_1(\zeta)&= -\{A^{-1}(\zeta),A(\zeta)\}A^{-1}(\zeta)+B^{-1}\{B^{-1}(\zeta),B(\zeta)\}(\zeta)\\
    &\quad+\{A^{-1}(\zeta)B(\zeta),B^{-1}(\zeta)\}+\{A^{-1}(\zeta),B(\zeta)\}B^{-1}(\zeta)-A^{-1}(\zeta)\{B(\zeta),B^{-1}(\zeta)\}\\
    &=-\{A^{-1}(\zeta),A(\zeta)\}A^{-1}(\zeta)+\{B^{-1}(\zeta),B(\zeta)\}B^{-1}(\zeta).
\end{align*}
Since  $\zeta \mapsto [B^{-1}\{B,B^{-1}\}](\zeta)$ is holomorphic on $\Supp(\widetilde{f})$, we obtain \eqref{eq:a1} by multiplying both sides of the above equality by $\overline\partial \widetilde f(\zeta)$ and integrating over $\C$. 

The relation \eqref{eq:a2} can be established in the same way. 
\end{proof}

\subsection{Proof of \texorpdfstring{\eqref{eq:main2}}{}}\label{sec:6.3}

Using Corollary \ref{cor:trace-comp}, we know that
\begin{align*}
    \sup_{\epsilon\in (0,\epsilon_0]}\frac{\epsilon^2}{|NJ\Omega|}\|\Opc(f_{d,j})\Op(\chi_N^2)\|_{\mathfrak{S}_1(\mathcal{H})}\leq C.
\end{align*}
This implies $\sup_{\epsilon\in (0,\epsilon_0]} \frac{\epsilon^2}{|NJ\Omega|}\left|\Tr_{\cH}[\Opc(f_{d,j})\Op(\chi_N^2)]\right|<\infty$. Now we can use Lemma \ref{lem:a_j} to expand $\Opc(f_{d,j})\Op(\chi_N^2)$ in powers of $\epsilon$: from \eqref{eq:chiNfhchiN} and Lemma \ref{lem:a_j}, we infer   
\begin{align*}
\Tr_{L^2(\mathbb{R}^3)}[\chi_N(\epsilon \cdot)f(H_{d,\theta})\chi_N(\epsilon \cdot)]
    &=\sum_{j=0}^n \epsilon^{j}\Tr_{\mathcal{H}}[\Opc(f_{d,j})\Op(\chi_N^2) ]+\mathcal{O}(\epsilon^{n-1}|NJ\Omega|).
    \end{align*}
Then as $f_{d,j}\in S^{1}_{\twi}(\mathfrak{S}_{1}(L^2_{\rm per}))$, from \eqref{eq:trace4}, we obtain
\begin{align}\label{eq:6.21}
\MoveEqLeft \Trb[f(H_{d,\theta})]=\sum_{j=0}^n \frac{\epsilon^j}{(2\pi)^2}\fint_{J\Omega}\int_{\Omega^*}\Tr_{L^2_{\rm per}}[f_{d,j}(k,X)]dkdX+\mathcal{O}(\epsilon^{n+1}).
\end{align}
In addition, from Lemma \ref{lem:a_j}, we can choose $f_{d,0}=f(h_{d,0})$,
\begin{align*}
    f_{d,1}=\frac{i}{2\pi}\int_{\C} \overline{\partial}\widetilde{f}(\zeta )  \{(\zeta-h_{d,0})^{-1},(\zeta-h_{d,0})\}(\zeta-h_{d,0})^{-1}(\zeta)dL(\zeta)
\end{align*}
and
\begin{align*}
 \MoveEqLeft    f_{d,2}(k,X)=-\frac{1}{\pi}\int_\C \overline{\partial}\widetilde{f}(\zeta)  \Big[-\frac{1}{4}(\zeta-h_{d,0})^{-1}\{(\zeta-h_{d,0}),(\zeta-h_{d,0})^{-1}\}^2\\
     &+\frac{1}{4}\{(\zeta-h_{d,0})^{-1}\{(\zeta-h_{d,0}),(\zeta-h_{d,0})^{-1}\}\}\\
     &+\frac{1}{8}(\zeta-h_{d,0})^{-1}\{(\zeta-h_{d,0}),(\zeta-h_{d,0})^{-1}\}_2\\
     &+\frac{i}{4}(\zeta-h_{d,0})^{-1}\Big(\partial_k h_{d,0}{\rm ad}_{\partial_x} (\zeta-h_{d,0})^{-1}-{\rm ad}_{\partial_x}h_{d,0} \partial_k (\zeta-h_{d,0})^{-1}\Big)\Big](k,X) \, dL(\zeta).
\end{align*}
Since $H_{d,\theta}$ are $H_{d,-\theta}$ are unitary equivalent (for instance via the symmetry with respect to the vertical $Ox_1z$ plane), their density of states are the same. This implies that the coefficents of the odd terms in the series expansion \eqref{eq:6.21} vanish.

\subsection{Proof of \texorpdfstring{\eqref{eq:main2'}}{}}

For any $a\in \R_+$, we define the unitary scaling operator $\mathcal{P}_{a}$ by 
\begin{align}
    \forall u \in L^2(\R^3), \quad (\mathcal{P}_{a})u(x,z)=a^{-1}u(a^{-1}x,z).
\end{align}
A straightforward calculation shows that
\begin{align*}
    H_{d,\theta}= \mathcal{P}_{(1+\eta^2)^{1/2}} \widetilde{H}_{d,\eta}\mathcal{P}_{(1+\eta^2)^{1/2}}^{-1},
\end{align*}
with $\eta:=\tan(\theta/2)$ and
\begin{align}
{\widetilde{H}}_{d,\eta}= -\frac{1}{2(1+\eta^2)}\Delta_x-\frac{1}{2}\partial_z^2 +V_{d}(x,\eta x,z).
\end{align}
Thus using cut-off functions $\chi_N$ satisfying \eqref{eq:chi_1}-\eqref{eq:chi_3}, we obtain
\begin{align*}
    \Trb[f(H_{d,\theta})]&=\lim_{N\to +\infty}\frac{1}{|NJ\Omega|}\Tr_{L^2(\mathbb{R}^3)} \left( \chi_N f\Big(\mathcal{P}_{(1+\eta^2)^{1/2}} \widetilde{H}_{d,\weps}\mathcal{P}_{(1+\eta^2)^{1/2}}^{-1}\Big)\chi_N \right)\\
    &=\lim_{N\to +\infty}\frac{1}{|NJ\Omega|}\Tr_{L^2(\mathbb{R}^3)} \left( \chi_N((1+\weps^2)^{-1/2} \cdot) f\Big(\widetilde{H}_{d,\weps}\Big)\chi_N((1+\weps^2)^{-1/2} \cdot) \right) \\
    &= (1+\eta^2)  \Trb[f(\widetilde H_{d,\eta})].
\end{align*}
We now observe that
\begin{align*}
    \widetilde{H}_{d,\weps} = \cU^{-1} {\rm Op}_{\weps} (\widetilde{h}_{d,\weps}) \cU,
\end{align*}
where
\[
\widetilde{h}_{d,\weps}(k,X):=\frac{1}{2(1+\weps^2)}(-i\nabla_x+k)^2-\frac{1}{2}\partial_z^2+V_d(x,X,z).
\]
It is easily seen that $ \widetilde{h}_{d,\weps}$ is a semiclassical symbol in the class $S^{\omega,1}_{\tau,\rm per}(\mathcal{B}(H^2_{\rm per},L^2_{\rm per}))$.
The rest of the proof of \eqref{eq:main2'} is essentially the same as the one of \eqref{eq:main2'}. The only differences are that (i) we can use the usual Weyl quantization operator ${\rm Op}_\eta$ (and not its twisted version), (ii) the symbol depends on $\eta$. For brevity we only sketch the main differences with the proof~\ref{eq:main2} and leave the details to the reader.

\medskip

Setting
\begin{align*}
     \widetilde{h}_{d,\weps,f}(k,X):=\chi_{f}\left(\widetilde{h}_{d,\weps}(k,X)\right) \, \widetilde{h}_{d,\weps}(k,X),
\end{align*}
with $\chi_{f}$ given by \eqref{eq:chi_f}, we have the following analogue to Lemma \ref{lem:4.2}.
\begin{lemma}\label{lem:4.2'}
Let $\omega(k,X):=1+|k|^2$. We have
\begin{align*} 
     \widetilde{h}_{d,\weps}, \widetilde{h}_{d,\weps,f}\in S^{\omega,1/2}_{\tau,\rm per}(
\eta,\mathcal{B}(H^2_{\rm per},L^2_{\rm per}))
\end{align*}
and
\begin{align} \label{eq:S-h-h-tilde'}
    ( \widetilde{h}_{d,\weps}- \widetilde{h}_{d,\weps
    ,f})\in S^{\omega,1/2}_{\tau,\rm per}(\eta,
\mathcal{B}(L^2_{\rm per},H^2_{\rm per}))\bigcap S^{1,1/2}_{\tau,\rm per}(\eta,
\mathcal{B}(L^2_{\rm per}))\bigcap S^{\omega,1/2}_{\tau,\rm per}(\eta,
\mathcal{B}(H^{-2}_{\rm per},L^2_{\rm per})).
\end{align}
\end{lemma}
It is clear that $ \widetilde{h}_{d,\weps}$ is a semiclassical symbol in $S^{\omega,1/2}_{\tau,\rm per}(\mathcal{B}(H^2_{\rm per},L^2_{\rm per}))$ since for $|\eta| < 1$,
\begin{align*}
\widetilde{h}_{d,\weps}(k,X)=  \widetilde{h}_{d,0}(k,X) + \frac 12 \left(\sum_{j=1}^{+\infty} (-1)^j \eta^{2j}\right) (-i\nabla_x+k)^2.
\end{align*}
The fact that
$\widetilde{h}_{d,\weps,f}$ and $\widetilde{h}_{d,\weps}-\widetilde{h}_{d,\weps,f}$ are semiclassical symbols follows from the resolvent formula
\begin{align}\label{eq:resolvent'}
    (\zeta-\widetilde{h}_{d,\weps})^{-1}=\sum_{j=0}^\infty (\zeta-\widetilde{h}_{d,0})^{-1}\Big[(\widetilde{h}_{d,\weps}-\widetilde{h}_{d,0})(\zeta-\widetilde{h}_{d,0})^{-1}\Big]^j
\end{align}
and Helffer-Sj\"ostrand's formula.

Then we have
\begin{lemma}\label{cor:trace'}
The symbol $(\widetilde{h}_{d,\weps}-\widetilde{h}_{d,\weps,f})$ belongs to the class $S_{\tau,\rm per}^{1/2}(\eta,\mathfrak{S}_1(L^2_{\rm per}))$.
\end{lemma}
\begin{proof}[Sketch of proof]
Using \eqref{eq:HF_formula} and the resolvent identity \eqref{eq:resolvent'}, we can write
\begin{align*}
    (\widetilde{h}_{d,\eta}-\widetilde{h}_{d,\eta,f})=\sum_{j=0}^n \weps^j\widetilde{h}_{d,0,f,j}+ \weps^{n+1}\widetilde{h}_{d,\weps,f,n+1}.
\end{align*}
It is easy to see that $\widetilde{h}_{d,0,f,0}=( \widetilde{h}_{d,0}- \widetilde{h}_{d,0
    ,f})$ is in $S_{\tau,\rm per}^{1}(\mathfrak{S}_1(L^2_{\rm per}))$, and $\widetilde{h}_{d,0,f,1}=0$. We consider now the next order term. Using \eqref{eq:h-htilde} this term can be written as 
    \begin{align*}
      \widetilde{h}_{d,0,f,2}= -\frac{1}{\pi}\int_{\C}\overline{\partial}\widetilde\xi_f(\zeta)(z) (\zeta-\widetilde{h}_{d,0})^{-1}(-i\nabla_x+k)^2(\zeta-\widetilde{h}_{d,0})^{-1}dL(\zeta).
    \end{align*}
Then we decompose  $(\zeta-\widetilde{h}_{d,\weps})^{-1}$ as 
\begin{align}\label{eq:A1-A2}
    (\zeta-\widetilde{h}_{d,\weps})^{-1}=A_1(\zeta)+A_2(\zeta),
\end{align}
with
\begin{align*}
    A_1(\zeta):=(\zeta-\widetilde{h}_{d,\widetilde{h}_{d,\weps}})^{-1}\1_{[-M,-\delta/2)}(\widetilde{h}_{d,\weps}),\qquad A_2(\zeta):=(\zeta-\widetilde{h}_{d,\widetilde{h}_{d,\weps}})^{-1}\1_{[-\delta/2,\infty)}(\widetilde{h}_{d,\weps}).
\end{align*}
The function $\zeta \mapsto A_2(\zeta)$ is holomorphic on a neighborhooof of $\Supp(\widetilde\xi_f)$, and we have
\begin{align*}
     \MoveEqLeft   
 \widetilde{h}_{d,0,f,2} = -\frac{1}{\pi}\int_{\C}\overline{\partial}\widetilde\xi_f(\zeta)(z) \Big[A_1(\zeta)(-i\nabla_x+k)^2(\zeta-A_1(\zeta)\\
     &\qquad+A_1(\zeta)(-i\nabla_x+k)^2(\zeta-A_2(\zeta)+A_2(\zeta)(-i\nabla_x+k)^2(\zeta-A_1(\zeta)\Big]dL(\zeta)
    \end{align*}
    Note that $\rank((A_1)(\zeta))<\infty$. Repeating the argument in Lemma \ref{cor:trace}, we obtain $ \widetilde{h}_{d,0,f,2}\in S_{\tau,\rm per}^{1}(\mathfrak{S}_1(L^2_{\rm per}))$. Proceeding similarly with the higher order terms, we obtain that $\widetilde{h}_{d,0,f,j}\in S_{\tau,\rm per}^{1}(\mathfrak{S}_1(L^2_{\rm per}))$ and $\widetilde{h}_{d,\weps,f,n+1}\in S_{\tau,\rm per}^{1,1/2}(\mathfrak{S}_1(L^2_{\rm per}))$. 
\end{proof}

 \subsection{Proof of \texorpdfstring{\eqref{eq:a2=a2'}}{}}
\label{sec:check}
 
Let us show that
    \begin{align}\label{eq:fj=fj~}
    \fint_{J\Omega^*}\int_{\Omega^*}\Tr_{L^2_{\rm per}}[f_{d,2}(k,X)] \, dk \, dX=  \fint_{J\Omega^*}\int_{\Omega^*}\Tr_{L^2_{\rm per}}[f_{d,2}'(k,X)] \, dk \, dX,
 \end{align}
where  $f_{d,2}$ and $f_{d,2}'$ are given by \eqref{eq:f_d2} and \eqref{eq:f_d2'} respectively.
For that, it suffices to show that
\begin{align}\label{eq:8.6}
\MoveEqLeft \fint\limits_{J\Omega^*}\int\limits_{\Omega^*}\Tr_{L^2_{\rm per}}\left[ \int_\C \overline{\partial}\widetilde{f}(\zeta)  \Big[-\frac{i}{4} A^{-1}\Big(\partial_k A{\rm ad}_{\partial_x} A^{-1}-{\rm ad}_{\partial_x}A \partial_k A^{-1}\Big) \Big](k,X) \, dL(\zeta)\right]dkdX \notag\\
  &= \fint\limits_{J\Omega^*}\int\limits_{\Omega^*}\Tr_{L^2_{\rm per}}\left[ \int_\C \overline{\partial}\widetilde{f}(\zeta)  \Big[-\frac{1}{2} A^{-1}(-i\nabla_x+k)^2A^{-1}-A^{-1} \Big](k,X) \, dL(\zeta)\right]dkdX
\end{align}
where $A:=A(\zeta)=\zeta-h_{d,0}(k,X)$.

As $\partial_k A=-(-i\nabla_x+k)$ and $-i{\rm ad}_{\partial_x} A =(-i\nabla_x+k)  A- A (-i\nabla_x+k)$, we have
\begin{align*}
  \MoveEqLeft  \Tr_{L^2_{\rm per}}\left[ \int_\C \overline{\partial}\widetilde{f}(\zeta)  \Big[-\frac{i}{4} A^{-1}\Big(\partial_k A{\rm ad}_{\partial_x} A^{-1}-{\rm ad}_{\partial_x}A \partial_k A^{-1}\Big) \Big](k,X) \, dL(\zeta)\right]\\
    &=\Tr_{L^2_{\rm per}}\left[ \int_\C \overline{\partial}\widetilde{f}(\zeta)  \Tr_{L^2_{\rm per}}
    \Big[-\frac{1}{2}A^{-1}(-i\nabla_x+k)^2 A^{-1}-\frac{1}{2}\nabla_k A\cdot \nabla_k A^{-1}\Big](k,X) \, dL(\zeta)\right].
\end{align*}
Notice that $ \nabla_k A\cdot \nabla_k A^{-1} =\nabla_k\cdot ((\nabla_k A)A^{-1})-(\Delta_k A)A^{-1}=\nabla_k\cdot ((\nabla_k A)A^{-1})+2A^{-1}$. We thus have
\begin{align*}
   \MoveEqLeft \int_{\Omega^*}\Tr_{L^2_{\rm per}}\left[ \int_\C \overline{\partial}\widetilde{f}(\zeta)  \Big[\nabla_k A\cdot \nabla_k A^{-1}\Big](k,X) \, dL(\zeta)\right] dk\\
    &=2\int_{\Omega^*}\Tr_{L^2_{\rm per}}\left[ \int_\C \overline{\partial}\widetilde{f}(\zeta)  A^{-1}(k,X) \, dL(\zeta)\right] dk\\
    &\quad+\int_{\Omega^*}\nabla_k\cdot \left(\Tr_{L^2_{\rm per}}\left[ \int_\C \overline{\partial}\widetilde{f}(\zeta)  [(\nabla_k A) A^{-1}](k,X) \, dL(\zeta)\right] \right)dk\\
    &=2\int_{\Omega^*}\Tr_{L^2_{\rm per}}\left[ \int_\C \overline{\partial}\widetilde{f}(\zeta)  A^{-1}(k,X) \, dL(\zeta)\right] dk,
\end{align*}
where the last equality holds since, according to $\tau$-equivariance, 
\begin{align*}
    \Tr_{L^2_{\rm per}}\left[\int_\C \overline{\partial}\widetilde{f}(\zeta)  [(\nabla_k A) A^{-1}](k,X) \, dL(\zeta)\right]
\end{align*}
is $\mathbb{L}^*$-periodic in $k\in \R^2$. Hence \eqref{eq:8.6}.

\section{Proof of Theorem~\ref{DoS-BM}}\label{sec:red-BM}

The proof of Theorem~\ref{DoS-BM} is very similar to the one of Theorem \ref{DoS-TBG}. So we only sketch it here for the sake of brevity. Throughout this section, the order function is $\omega(\kappa,X):=(1+|\kappa|^2)^{1/2}$. 

\medskip

Let us first study the matrix $T_\epsilon(\kappa)$ defined in~\eqref{eq:T-V-1}. By definition of $\epsilon$, the twist angle~$\theta$ is given by
\begin{align*}
    \theta =2 \arcsin(\epsilon),
\end{align*}
and is therefore real-analytic in $\epsilon$ on $(-1,1)$. Thus, $T_\epsilon(\kappa)$ can be expanded as  
\begin{align}\label{eq:T-epsilon-dec}
    T_\epsilon(\kappa)=\sum_{j=0}^\infty\epsilon^j T_{0,j}(\kappa) \qquad \mbox{with} \quad T_{0,j}(\bullet):=\frac{1}{j!}\frac{\partial^j}{\partial\epsilon^j} T_{\epsilon}|_{\epsilon=0}(\bullet) \in C^\infty(\R^2;\C^{4\times 4}).
\end{align}
In particular, $T_{0,0}(\kappa)=T_0(\kappa)$,
\begin{align*}
    T_{0,1}(\kappa):=
   \begin{pmatrix}
   v_{\rm F} (-\sigma_2,\sigma_1) \cdot \kappa & 0\\
   0&  -v_{\rm F} (-\sigma_2,\sigma_1) \cdot \kappa \end{pmatrix};
\end{align*}
and $
    T_{0,2}(\kappa):=
   -\frac{1}{2} T_0(\kappa)$.
It is easy to see that there exists $C_j\in \R_+$ such that for any $\kappa\in \R^2$,
\begin{align}\label{eq:bound-Tj1}
    \|T_{0,j(\kappa)}\|_{\rm F}\leq C_j|\kappa|,\quad \|\partial_\kappa T_{0,j}(\kappa)\|_{\rm F}\leq C_j
\end{align}
and that for any $\alpha\in \N^2$ such that $|\alpha|\geq 2$,
\begin{align}\label{eq:bound-Tj2}
    \partial_\kappa^\alpha T_{0,j}(\kappa)=0.
\end{align}

\subsection{Properties of the operator-valued symbol \texorpdfstring{$h^{\rm eff}_{d,K,\epsilon}(k,X)$ }{}}\label{sec:t-symbol}

To study the properties of the operator-valued symbol $h^{\rm eff}_{d,K,\epsilon}$ defined in~\eqref{eq:def_tdeps}, we argue as in Section \ref{sec:hd0-symbol}. Because of symmetries, it suffices to consider the case when $0 \le \epsilon \le \frac{\sqrt 3}2$, i.e. $0 \le \theta \le \frac{2\pi}3$.
\begin{lemma}\label{lem:spec-BM}
Let $0\leq \epsilon\leq \frac{\sqrt{3}}{2}$. Under Assumptions \ref{def:MGpotential} and \ref{def:inter}, $h^{\rm eff}_{d,K,\epsilon}(k,X)$ is a self-adjoint operator on $L^2_\#$ with domain $H^1_{\#}$ satisfying
\begin{enumerate}
\item for any $k,X\in \mathbb{R}^2$,
\begin{align}\label{eq:ess-BM}
    \sigma_{\rm ess}(h^{\rm eff}_{d,K,\epsilon}(k,X))=\emptyset;
\end{align}
\item for any $-\infty < E_1<E_2< +\infty$, there is a constant $m_{\rm EMS}(E_1,E_2) \in \N$ such that
\begin{align}\label{eq:rankbound-BM}
 \sup_{\epsilon\in [0,\frac{\sqrt{3}}{2}]} \sup_{\substack{ k,X\in \mathbb{R}^2}}\rank
    \Big(\1_{[E_1,E_2]}\big(h^{\rm eff}_{d,K,\epsilon}(k,X)\big)\Big)\leq m_{\rm EMS}(E_1,E_2).
\end{align}
\end{enumerate}
\end{lemma}
\begin{proof}[Stetch of proof]
 The first assertion follows from the same arguments as in the proof of Lemma \ref{lem:spec}, we have \eqref{eq:ess-BM}.

The second assertion easily follows from the facts that (i) $h^{\rm eff}_{d,K,\epsilon}(k,X)=T_\epsilon(-i\nabla_x+k-K) +\mathcal{V}_{d}(X)$ is a Fourier multiplier operator on $L^2_{\#}$, (ii) its eigenvalues are therefore the eigenvalues of the matrices $\big(T_{\epsilon}(G+k-K)+\mathcal{V}_{d}(X)\big)$, $G\in \mathbb{L}^*$, (iii) the function $\R^2 \ni X \mapsto \mathcal{V}_{d}(X) \in \C^{4 \times 4}$ is bounded, and (iv) the eigenvalues of $T_{\epsilon}(G+k-K)$ are $\pm v_{\rm F} |G+k-K|^2$ (each of them being two-fold degenerate).
\end{proof}

We now introduce a new symbol $h^{\rm eff}_{d,K,f,\epsilon}$. Let $\delta>0$ be a small enough constant. Setting
\[
e_1:= \inf\Supp(\widetilde{f})-2\delta,\quad e_2:= \sup\Supp(\widetilde{f})+2\delta,
\]
we define the operator-valued symbol $h^{\rm eff}_{d,K,f,\epsilon}$ as
\begin{align}\label{eq:t-t-BM}
    h^{\rm eff}_{d,K,f,\epsilon}(k,X)=\chi_{f,\rm EMS}(h^{\rm eff}_{d,K,\epsilon})(k,X)h^{\rm eff}_{d,K,\epsilon}(k,X)
\end{align}
where $\chi_{f,\rm EMS}\in C^\infty(\R)$  satisfies
\begin{align*}
    &\Supp(\chi_{f,\rm EMS})\subset((-\infty,e_1+\delta)\cup (e_2-\delta,+\infty)), \\
    &\chi_{f,\rm EMS}(s)=1\quad \textrm{for }s\in  (-\infty,e_1]\cup [e_2,+\infty)\quad \textrm{and}\quad 0\leq \chi_{f,\rm EMS}\leq 1.
\end{align*}
Analogous to Lemma \ref{lem:4.2} we have the following results.
\begin{lemma}\label{lem:4.2-BM}
We have
\begin{align*}
 h^{\rm eff}_{d,K,\bullet},h^{\rm eff}_{d,K,f,\bullet}\in S^{\omega,\frac{\sqrt{3}}{2}}_{\tau}(
\mathcal{B}(\epsilon,H^1_{\#},L^2_{\#}))
\end{align*}
and
\begin{align}\label{eq:S-t-ttilde}
    (h^{\rm eff}_{d,K,\bullet}-h^{\rm eff}_{d,K,f,\bullet})\in S^{\omega,\frac{\sqrt{3}}{2}}_{\tau}(\epsilon,\mathcal{B}(L^2_{\#},H^1_{\#}))\bigcap S^{1,\frac{\sqrt{3}}{2}}_{\tau}(\epsilon,\mathcal{B}(L^2_{\#}))\bigcap S^{\omega,\frac{\sqrt{3}}{2}}_{\tau}(\epsilon,
\mathcal{B}(H^{-1}_{\#},L^2_{\#})).
\end{align}
\end{lemma}
\begin{proof}[Sketch of proof]
We first recall that $\mathcal{V}_{d,K}(X)$ is in $C^\infty(\R^2;\C^{4\times 4}_{\rm herm})$. As $h^{\rm eff}_{d,K,\epsilon}$is $\tau$-equivariant, it is easy to see that $ h^{\rm eff}_{d,K,\bullet}\in S^{\omega,\frac{\sqrt{3}}{2}}_{\tau}(
\mathcal{B}(H^{1}_{\#},L^2_{\#})).$ Then from \eqref{eq:T-epsilon-dec}-\eqref{eq:bound-Tj2}, we get $h^{\rm eff}_{d,K,\bullet}\in S^{\omega,\frac{\sqrt{3}}{2}}_{\tau}(\epsilon,
\mathcal{B}(H^{1}_{\#},L^2_{\#})).$

To show that $h^{\rm eff}_{d,K,f,\epsilon}\in S^{\omega,\frac{\sqrt{3}}{2}}_{\tau}(
\mathcal{B}(\epsilon,H^1_{\#},L^2_{\#}))$, it suffices to prove \eqref{eq:S-t-ttilde}. Let $\xi_{f,\rm EMS}(s)=s(1-\chi_{f,\rm EMS}(s))$. From Eqn. \eqref{eq:HF_formula}, we infer
\begin{align}\label{eq:t-htilde-BM}
    h^{\rm eff}_{d,K,\epsilon}-h^{\rm eff}_{d,K,f,\epsilon}=-\frac{1}{\pi}\int_{\mathbb{C}}\overline{\partial}\widetilde{\xi}_{f,\rm EMS}(\zeta)(\zeta-h^{\rm eff}_{d,K,\epsilon})^{-1}dL(\zeta).
\end{align}
where $\widetilde{\xi}_{f,\rm EMS}\in C^\infty_{\rm c}(\mathbb{C})$ and $|\overline{\partial}\widetilde{\xi}_{f,\rm EMS}(\zeta)|=\mathcal{O}(|\Im \zeta|^\infty)$. 

Arguing as in Lemma \ref{lem:4.2}, for $a,b\in \{0,1\}$ and $0\leq a+b\leq 1$, from the $\tau$-equivariance of $H_{d,K,\theta}$ we infer
\begin{align}\label{eq:29-BM}
   \MoveEqLeft    \sup_{\epsilon\in [0,\frac{\sqrt{3}}{2}]}\sup_{  k,X\in \R^2}\omega^{-(a+b)}\|(1-\Delta_x)^{a/2} (\zeta-h^{\rm eff}_{d,K,\epsilon})^{-1}(1-\Delta_x)^{b/2}\|_{\mathcal{B}(L^2_{\#})}\leq C'|\Im \zeta|^{-1}.
\end{align}
 This inequality holds by using the fact that for any $u\in H^1_{\#}$ there exists constants $C_1,C_2 \in \R_+$ independent of $k,X$ such that for any $\epsilon\in [0,\frac{\sqrt{3}}{2}]$,
\begin{align}\label{eq:delta-t}
    \|(1-\Delta_x)^{1/2}u\|_{L^2_{\#}}\leq C_1\|h^{\rm eff}_{d,K,\epsilon}([k],X)u\|_{L^2_{\#}}+C_2\|u\|_{L^2_{\#}}.
\end{align}
Thus \eqref{eq:29-BM} gives,
\begin{align*}
 \sup_{\epsilon\in [0,\frac{\sqrt{3}}{2}]} \sup_{k,X\in\mathbb{R}^2}\| (\zeta-h^{\rm eff}_{d,K,\epsilon})^{-1}\|_{\mathcal{B}(L^2_{\#})}&\leq C|\Im \zeta|^{-1},\\
 \sup_{\epsilon\in [0,\frac{\sqrt{3}}{2}]} \sup_{k,X\in\mathbb{R}^2}\omega^{-1}\| (\zeta-h^{\rm eff}_{d,K,\epsilon})^{-1}\|_{\mathcal{B}(L^2_{\#},H^1_{\#})}&\leq C|\Im \zeta|^{-1},\\
 \sup_{\epsilon\in [0,\frac{\sqrt{3}}{2}]} \sup_{k,X\in\mathbb{R}^2}\omega^{-1}\| (\zeta-h^{\rm eff}_{d,K,\epsilon})^{-1}\|_{\mathcal{B}(H^{-1}_{\#},L^2_{\#})}&\leq C|\Im \zeta|^{-1}.
\end{align*}
Then arguing as in Lemma \ref{lem:4.2}, we deduce 
\begin{align*}
    (h^{\rm eff}_{d,K,\bullet}-h^{\rm eff}_{d,K,f,\bullet})\in S^{\omega,\frac{\sqrt{3}}{2}}_{\tau}(
\mathcal{B}(L^2_{\#},H^1_{\#}))\bigcap S^{1,\frac{\sqrt{3}}{2}}_{\tau}(
\mathcal{B}(L^2_{\#}))\bigcap S^{\omega,\frac{\sqrt{3}}{2}}_{\tau}(
\mathcal{B}(H^{-1}_{\#},L^2_{\#})).
\end{align*}

To end the proof, it remains to show $ (h^{\rm eff}_{d,K,\epsilon}-h^{\rm eff}_{d,K,f,\epsilon})$ satisfies \eqref{eq:expansion_symbol}. Notice that by resolvent identity,
\begin{align}\label{eq:resolvent}
    (h^{\rm eff}_{d,K,\epsilon}-\zeta)^{-1}=\sum_{j=0}^\infty (h_{d,K,0}^{\rm eff}-\zeta)^{-1}\Big[(T_{\epsilon}-T_0)(-i\nabla_x+k-K)(h_{d,K,0}^{\rm eff}-\zeta)^{-1}\Big]^j,
\end{align}
and $(T_\epsilon-T_0n)(-i\nabla_k+k-K)$ is also $\tau$-equivariant.
Then inserting \eqref{eq:T-epsilon-dec} into above identity, using \eqref{eq:bound-Tj1}-\eqref{eq:bound-Tj2}, \eqref{eq:t-htilde-BM}, and proceeding as above, we infer \eqref{eq:S-t-ttilde}. This ends the proof.
\end{proof}

Applying Lemma \ref{lem:spec-BM}, we obtain 
\begin{align*}
 \MoveEqLeft   \sup_{\epsilon\in (0,\frac{\sqrt{3}}{2}]}\sup_{k,X\in\mathbb{R}^2}\rank(h^{\rm eff}_{d,K,\epsilon}-h^{\rm eff}_{d,K,f,\epsilon})\\
 &\leq \sup_{\epsilon\in (0,\frac{\sqrt{3}}{2}]}\sup_{k,X\in\mathbb{R}^2}\rank(\1_{[e_1,e_2]}(h^{\rm eff}_{d,K,\epsilon}))\leq m_{\rm EMS}(e_1,e_2).
\end{align*}
Proceeding as for Lemma \ref{cor:trace} and Lemma \ref{lem:4.2'}, we obtain the following.
\begin{corollary}\label{cor:trace-BM}
The symbol $(h^{\rm eff}_{d,K,\bullet}-h^{\rm eff}_{d,K,f,\bullet})$ belongs to $ S_{\tau}^{1,\frac{\sqrt{3}}{2}}(\epsilon,\mathfrak{S}_1(L^2_{\#}))$.
\end{corollary}

\subsection{Finiteness of the LHS of \texorpdfstring{\eqref{eq:DoS-BM}}{}}\label{sec:7.2}
Proceeding as for Corollary \ref{cor:ellip} and Lemma \ref{lem:inverse}, from the $J\mathbb{L}$-periodicity of $X\mapsto h^{\rm eff}_{d,K,\epsilon}(\cdot,X)$,  we also have the following.
\begin{lemma}\label{lem:inverse-BM}
Let $0<\epsilon_0\leq \frac{\sqrt{3}}{2}$ be small enough. There exists an operator-valued symbol $g_{K,\epsilon}(\zeta)\in S^{\omega,\epsilon_0}_{\tau}(\mathcal{B}(L^2_{\#},H^1_{\#}))\bigcap S^{1,\epsilon_0}_{\tau}(\mathcal{B}(L^2_{\#}))$ such that $g_{K,\epsilon}$ is $J\mathbb{L}$-periodic w.r.t. $X$ and for any $\zeta\in \Supp(\widetilde{f})$,
\begin{align*}
    \Op(g_{K,\epsilon}(\zeta))= (\Op(h^{\rm eff}_{d,K,f,\epsilon})-\zeta)^{-1}.
\end{align*}
\end{lemma}

Analogous to \eqref{eq:Ztrans1}, we have 
\begin{align}\label{eq:Ztrans-BM1}
     \mathcal{U}\chi_N(\epsilon \cdot)f(H_{d,K,\theta}^{\rm eff}) \chi_N(\epsilon \cdot)\mathcal{U}^{-1}=\Op(\chi_N) f(\Op(h^{\rm eff}_{d,K,\epsilon}))\Op(\chi_N).
\end{align}

According to Lemma \ref{lem:inverse-BM}, $(\Op(h^{\rm eff}_{d,K,f,\epsilon})-\zeta)^{-1}$ is holomorphic on $\Supp(\widetilde{f})$, then
\begin{align*}
    \int_{\C} \overline{\partial}\widetilde{f}(\zeta)(\zeta-\Op(h^{\rm eff}_{d,K,f,\epsilon}))^{-1}dL(\zeta)=-\int_{\C} \widetilde{f}(\zeta)\overline{\partial}(\zeta-\Op(h^{\rm eff}_{d,K,f,\epsilon}))^{-1}dL(\zeta)=0.
\end{align*}
Thus
    \begin{align}\label{eq:holom-BM}
\MoveEqLeft    f(\Op(h^{\rm eff}_{d,K,\epsilon}))=-\frac{1}{\pi}\int_{\C} \overline{\partial}\widetilde{f}(\zeta)[(\zeta-\Op(h^{\rm eff}_{d,K,\epsilon}))^{-1}- (\zeta-\Op(h^{\rm eff}_{d,K,f,\epsilon}))^{-1}]dL(\zeta)\notag\\
    &=\frac{1}{\pi}\int_{\C} \overline{\partial}\widetilde{f}(\zeta)(\zeta-\Op(h^{\rm eff}_{d,K,\epsilon}))^{-1}\notag\\
    &\quad\times (\Op(h^{\rm eff}_{d,K,f,\epsilon})-\Op(h^{\rm eff}_{d,K,\epsilon}))(\zeta-\Op(h^{\rm eff}_{d,K,f,\epsilon}))^{-1}dL(\zeta).
\end{align}
Then arguing as for the proof of Theorem \ref{DoS-TBG}, there exists $C\in \R_+$ such that for all $N\in \N^*$,
\begin{align*}
    \sup_{\epsilon\in (0,\epsilon_0]}\frac{\epsilon^2}{|NJ\Omega|}\|\chi_N(\epsilon\cdot)f(H_{d,K,\theta}^{\rm eff})\chi_N(\epsilon\cdot)\|_{\mathfrak{S}_1(L^2(\R^2))}\leq \frac{C|(N+1)J\Omega|}{|NJ\Omega|}\leq C.
\end{align*}

\subsection{Asymptotic expansion for DoS}\label{sec:7.3}

\begin{lemma}\label{lem:7.5}
For all $\zeta\in\Supp(\widetilde{f})\setminus\mathbb{R}$, there exist a sequence $(a^{\rm eff}_{K,j}(\zeta))$, $j\in \N$, which are $J\mathbb{L}$-periodic w.r.t. $X$, and a sequence $(r^{\rm eff}_{K,n,\epsilon}(\zeta))$ such that
\begin{align}\label{eq:7.13}
    \overline{\partial}\widetilde{f}(\zeta) a^{\rm eff}_{K,j}(\zeta) \in S_{\tau}^{1}(\mathcal{B}(L^2_{\#})),\quad \overline{\partial}\widetilde{f}(\zeta)r^{\rm eff}_{n,\bullet}(\zeta)\in S_{\tau}^{1,\epsilon_0}(\epsilon,\mathcal{B}(L^2_{\#})),
\end{align}
and for any $n\geq 0$,
\begin{align}\label{eq:7.14}
   (\zeta-\Opc(h^{\rm eff}_{d,K,\epsilon}))^{-1}=  \sum_{j=0}^{n}\epsilon^j \Opc(a^{\rm eff}_{K,j}(\zeta))-\epsilon^{n+1}(\zeta-\Opc(h^{\rm eff}_{d,K,\epsilon}))^{-1}\Opc(r^{\rm eff}_{K,n,\epsilon}(\zeta)).
\end{align}
In particular, $a^{\rm eff}_{K,0}(\zeta)=(\zeta-h^{\rm eff}_{d,K,0})^{-1}$, 
\begin{align*}
     a^{\rm eff}_{K,1}(\zeta)&=-\frac{i}{2}\{(\zeta-h^{\rm eff}_{d,K,0})^{-1},\zeta-h^{\rm eff}_{d,K,0}\}(\zeta-h^{\rm eff}_{d,K,0})^{-1}\\
     &\quad+(\zeta-h^{\rm eff}_{d,K,0})^{-1}T_{0,1}(-i\nabla_x+k-K)(\zeta-h^{\rm eff}_{d,K,0})^{-1}
\end{align*}
and
\begin{align*}
     a^{\rm eff}_{K,2}(\zeta)&=-\frac{1}{4}(\zeta-h^{\rm eff}_{d,K,0})^{-1}\{(\zeta-h_{d,0}),(\zeta-h^{\rm eff}_{d,K,0})^{-1}\}^2\\
     &\quad+\frac{1}{4}\{(\zeta-h^{\rm eff}_{d,K,0})^{-1},\{(\zeta-h^{\rm eff}_{d,K,0}),(\zeta-h^{\rm eff}_{d,K,0})^{-1}\}\}\\
     &\quad+\frac{1}{8}(\zeta-h^{\rm eff}_{d,K,0})^{-1}\{(\zeta-h^{\rm eff}_{d,K,0}),(\zeta-h^{\rm eff}_{d,K,0})^{-1}\}_2\\
     &\quad+(\zeta-h^{\rm eff}_{d,K,0}(k,X))^{-1}T_{0,2}(-i\nabla_x+k-K)(\zeta-h^{\rm eff}_{d,K,0}(k,X))^{-1}\\
     &\quad +(\zeta-h^{\rm eff}_{d,K,0}(k,X))^{-1}\Big[T_{0,1}(-i\nabla_x+k-K)(\zeta-h^{\rm eff}_{d,K,0}(k,X))^{-1}\Big]^2\\
     &\quad+ \frac{i}{2}\{(\zeta-h_{d,K,0}^{\rm eff})^{-1}, T_{0,1}(-i\nabla_x+k-K)\}(\zeta-h_{d,K,0}^{\rm eff})^{-1}\\
     &\quad- \frac{i}{2}\{(\zeta-h_{d,K,0}^{\rm eff})^{-1}T_{0,1}(-i\nabla_x+k-K)(\zeta-h_{d,K,0}^{\rm eff}),\zeta-h_{d,K,0}^{\rm eff}\}(\zeta-h_{d,K,0}^{\rm eff})^{-1}\\
     &\quad-\frac{i}{2}\{(\zeta-h_{d,K,0}^{\rm eff}),\zeta-h_{d,K,0}^{\rm eff}\}(\zeta-h_{d,K,0}^{\rm eff})^{-1}T_{0,1}(-i\nabla_x+k-K)(\zeta-h_{d,K,0}^{\rm eff}).
\end{align*}
\end{lemma}
\begin{proof}[Sketch of proof]
This proof is essentially the same as that for Lemma \ref{lem:5.7}. We only point out the differences. Setting $a^{\rm eff}_{K,0,\epsilon}(\zeta):=(\zeta-h^{\rm eff}_{d,K,\epsilon}(k,X))^{-1}$, we deduce from \eqref{eq:T-epsilon-dec}-\eqref{eq:bound-Tj2} that
\begin{align*}
    \widetilde{f}(\zeta) a^{\rm eff}_{K,0,\bullet}(\zeta)\in S^{1,\frac{\sqrt{3}}{2}}(\epsilon,\mathcal{B}(L^2_{\#}))\bigcap S^{\omega,\frac{\sqrt{3}}{2}}_{\tau}(\epsilon,\mathcal{B}(L^2_{\#},H^2_{\#})).
\end{align*}
In particular,  we have
\begin{align}\label{eq:7.15}
    \widetilde{f}(\zeta) a^{\rm eff}_{K,0,\epsilon}(\zeta)= \sum_{j=0}^n \epsilon^j \widetilde{f}(\zeta) a^{\rm eff}_{K,0,j}(\zeta)+\mathcal{O}(\epsilon^{n+1})\quad \textrm{in } S^{1,\frac{\sqrt{3}}{2}}(\mathcal{B}(L^2_{\#}))\bigcap S^{\omega,\frac{\sqrt{3}}{2}}_{\tau}(\mathcal{B}(L^2_{\#},H^2_{\#}))
\end{align}
with
\begin{align*}
    a^{\rm eff}_{K,0,0}(\zeta)&= (\zeta-h^{\rm eff}_{d,K,0}(k,X))^{-1}\\
    a^{\rm eff}_{K,0,1}(\zeta)&= (\zeta-h^{\rm eff}_{d,K,0}(k,X))^{-1}T_{0,1}(-i\nabla_x+k-K)(\zeta-h^{\rm eff}_{d,K,0}(k,X))^{-1}
\end{align*}
and
\begin{align*}
    a^{\rm eff}_{K,0,2}(\zeta)&= (\zeta-h^{\rm eff}_{d,K,0}(k,X))^{-1}T_{0,2}(-i\nabla_x+k-K)(\zeta-h^{\rm eff}_{d,K,0}(k,X))^{-1}\\
    &\quad+(\zeta-h^{\rm eff}_{d,K,0}(k,X))^{-1}\Big[T_{0,1}(-i\nabla_x+k-K)(\zeta-h^{\rm eff}_{d,K,0}(k,X))^{-1}\Big]^2.
\end{align*}
Then proceeding as for Lemma \ref{lem:5.7} and using \eqref{eq:7.15}, we obtain \eqref{eq:7.13}-\eqref{eq:7.14}, and the formulas for $a_{K,0}^{\rm eff}$, $a_{K,1}^{\rm eff}$ and $a_{K,2}^{\rm eff}$.
\end{proof}
We also have an equivalent of Lemma \ref{lem:5.8} for $(\zeta-\Op(h_{d,K,f,\epsilon}))^{-1}$. Finally, we have the following equivalent of Lemma \ref{lem:a_j} for $f(\Opc(h^{\rm eff}_{d,K,\epsilon}))$.
\begin{lemma}\label{lem:a_j-t}
There exist a sequence $(f^{\rm eff}_{d,K,j})_{j \in \N}$ of operator-valued symbols in $S_{\tau}^{1}(\mathfrak{S}_1(L^2_{\#}))$, which are $J\mathbb{L}$-periodic w.r.t. $X$, and a constant $C\in \R_+$ such that for all $n\in \N^*$, $j\in \N$ and all $\epsilon\in (0,\frac{\sqrt 3}2]$,
\begin{align}\label{eq:dec-fh-t}
    \frac{\epsilon^2}{|NJ\Omega|}\left\|f(\Opc(h^{\rm eff}_{d,K,\epsilon}))\Op(\chi_N^2)-\sum_{j=0}^n\epsilon^j \Opc(f^{\rm eff}_{d,K,j})\Op(\chi_N^2)\right\|_{\mathfrak{S}_1(\widetilde{\mathcal{H}})}\leq C\frac{\epsilon^{n+1}}{|\Omega|}.
\end{align}
In particular, 
\begin{align}\label{eq:a0-t}
    f^{\rm eff}_{d,K,0}(k,X)=f(h^{\rm eff}_{d,K,0}(k,X)),
\end{align}
\begin{align}\label{eq:a1-t}
    f_{d,K,1}(\kappa,X)&:=\frac{i}{2\pi}\int_{\C} \overline{\partial}\widetilde{f}(\zeta )  \{(\zeta- \mathfrak{h}_{d,K,0}^{\rm eff})^{-1},(\zeta- \mathfrak{h}_{d,K,0}^{\rm eff})\}(\kappa,X)(\zeta- \mathfrak{h}_{d,K,0}^{\rm eff}(\kappa,X))^{-1} \, dL(\zeta)\\
    &\quad-\frac{1}{\pi}\int_{\C} \overline{\partial}\widetilde{f}(\zeta )  (\zeta- \mathfrak{h}_{d,K,0}^{\rm eff}(\kappa,X))^{-1}T_{0,1}(\kappa)(\zeta- \mathfrak{h}_{d,K,0}^{\rm eff}(\kappa,X))^{-1}dL(\zeta)
\end{align}
and
\begin{align*}
  \MoveEqLeft  f^{\rm eff}_{d,K,2}(k,X)=-\frac{1}{\pi}\int_{\C} \overline{\partial}\widetilde{f}(\zeta )  \Big[-\frac{1}{4}(\zeta-h^{\rm eff}_{d,K,0})^{-1}\{(\zeta-h_{d,K,0}^{\rm eff}),(\zeta-h^{\rm eff}_{d,K,0})^{-1}\}^2\\
     &+\frac{1}{4}\{(\zeta-h^{\rm eff}_{d,K,0})^{-1},\{(\zeta-h^{\rm eff}_{d,K,0}),(\zeta-h^{\rm eff}_{d,K,0})^{-1}\}\}\\
     &+\frac{1}{8}(\zeta-h^{\rm eff}_{d,K,0})^{-1}\{(\zeta-h^{\rm eff}_{d,K,0}),(\zeta-h^{\rm eff}_{d,K,0})^{-1}\}_2\\
     &+(\zeta-h^{\rm eff}_{d,K,0}(k,X))^{-1}T_{0,2}(-i\nabla_x+k-K)(\zeta-h^{\rm eff}_{d,K,0}(k,X))^{-1}\\
     & +(\zeta-h^{\rm eff}_{d,K,0}(k,X))^{-1}\Big[T_{0,1}(-i\nabla_x+k-K)(\zeta-h^{\rm eff}_{d,K,0}(k,X))^{-1}\Big]^2\\
     &+ \frac{i}{2}\{(\zeta-h_{d,K,0}^{\rm eff})^{-1}, T_{0,1}(-i\nabla_x+k-K)\}(\zeta-h_{d,K,0}^{\rm eff})^{-1}\\
     &- \frac{i}{2}\{(\zeta-h_{d,K,0}^{\rm eff})^{-1}T_{0,1}(-i\nabla_x+k-K)(\zeta-h_{d,K,0}^{\rm eff})^{-1},\zeta-h_{d,K,0}^{\rm eff}\}(\zeta-h_{d,K,0}^{\rm eff})^{-1}\\
     &-\frac{i}{2}\{(\zeta-h_{d,K,0}^{\rm eff})^{-1},\zeta-h_{d,K,0}^{\rm eff}\}(\zeta-h_{d,K,0}^{\rm eff})^{-1}T_{0,1}(-i\nabla_x+k-K)(\zeta-h_{d,K,0}^{\rm eff})^{-1}\Big]dL(\zeta).
\end{align*}
\end{lemma}
\begin{proof}[Sketch of proof]
We only need to show \eqref{eq:a1-t}. Let $A_{\rm eff}(\zeta):=\zeta-h^{\rm eff}_{d,K,0}$, $B_{\rm eff}(\zeta):=\zeta-h^{\rm eff}_{d,K,f,0}$, and $C_{\rm eff}:=T_{0,1}(-i\nabla_x+k-K)$.
Proceeding as for Lemma \ref{lem:a_j}, we have
\begin{align*}
   f^{\rm eff}_{d,K,1}(\zeta)&=-\frac{1}{\pi}\int_{\C} \overline{\partial}\widetilde{f}(\zeta )  \Big[\frac{i}{2}\{A_{\rm eff}^{-1}(\zeta),A_{\rm eff}(\zeta)\}A_{\rm eff}^{-1}(\zeta)\notag\\
    &\quad+ A_{\rm eff}^{-1}(\zeta)C_{\rm eff}A_{\rm eff}^{-1}(\zeta)\big(A_{\rm eff}(\zeta)-B_{\rm eff}(\zeta)\big)B_{\rm eff}^{-1}(\zeta)\\
    &\quad + A_{\rm eff}^{-1}(\zeta)\big(A_{\rm eff}(\zeta)-B_{\rm eff}(\zeta)\big)B_{\rm eff}^{-1}(\zeta)C_{\rm eff}B_{\rm eff}^{-1}(\zeta)\Big]dL(\zeta).
\end{align*}
Notice that, by the functional calculus, $\zeta \mapsto B_{\rm eff}(\zeta)$ is holomorphic in a neighborhood of $\Supp(\widetilde{f})$. Thus, 
\begin{align*}
  \MoveEqLeft f^{\rm eff}_{d,K,1}(\zeta)=-\frac{1}{\pi}\int_{\C} \overline{\partial}\widetilde{f}(\zeta ) \Big[ -\frac{i}{2}\{A_{\rm eff}^{-1}(\zeta),A_{\rm eff}(\zeta)\}A_{\rm eff}^{-1}(\zeta)+A_{\rm eff}^{-1}(\zeta)C_{\rm eff}A_{\rm eff}^{-1}(\zeta) \Big]dL(\zeta).
\end{align*}
The expression of $f_{d,K,2}^{\rm eff}$ is obtained similarly.
\end{proof}

\subsection{End of the proof of Theorem \ref{DoS-BM}}
Now we turn to the
\begin{proof}[Sketch proof of Theorem \ref{DoS-BM}]

Proceeding as for Corollary \ref{cor:trace-comp} and Theorem \ref{DoS-TBG},  we infer from Lemma \ref{lem:a_j-t} that 
\begin{align}\label{eq:7.12}
\MoveEqLeft\Trb[f(H_{d,K,\theta}^{\rm eff})]=\sum_{j=0}^n \frac{\epsilon^j}{(2\pi)^2}\fint_{J\Omega}\int_{\Omega^*}\Tr_{L^2_{\#}}\Big(f^{\rm eff}_{d,K,j}(k,X)\Big)dkdX+\mathcal{O}(\epsilon^{n+1}).
\end{align}

Recall that
\begin{align*}
    \mathfrak{h}_{d,K,0}^{\rm eff}(k,X):=T_0(k-K)+\mathcal{V}_{d,K}(X) \in \C^{4 \times 4}_{\rm herm}.
\end{align*}
Notice that $h^{\rm eff}_{d,K,0}(k,X)$ is a multiplication operator in the Fourier representation since $h^{\rm eff}_{d,K,0}(k,X)= \mathfrak{h}_{d,K,0}^{\rm eff}(-i\nabla_x+k,X)$. As $(T_{0,\ell}(-\nabla_x+k-K))_{\ell\in \N}$ are also multiplication operators in the Fourier representation, for any $0\leq j\leq n$, the operator $a_{d,K,j}$ are Fourier multipliers since $a_{d,K,j}$ is a composition of $\partial_k^\alpha h^{\rm eff}_{d,K,0}$, $\partial_X^\beta h^{\rm eff}_{d,K,0}$ and $\partial_k^\alpha T_{0,\ell}$ for some $\alpha,\beta \in \N^2$ and $\ell\in \N$. Thus there is a sequence of functions $(f_{d,K,j}(k,X))_{j \in \N}$ in $S^1_{\tau}(\C^{4\times 4})$ such that $f^{\rm eff}_{d,K,j}(k,X)=f_{d,K,j}(-i\nabla_x+k,X)$. As a result,
\begin{align*}
  \MoveEqLeft  \fint_{J\Omega}\int_{\Omega^*}\Tr_{L^2_{\#}}[f^{\rm eff}_{d,K,j}(k,X)] \, dk \, dX=\fint_{J\Omega}\int_{\Omega^*}\Tr_{L^2_{\#}}[f_{d,K,j}(-i\nabla_x+k,X)]d k dX\\  
&=\sum_{G\in\mathbb{L}^*}\fint_{J\Omega}\int_{\Omega^*}\Tr_{\C^4}[f_{d,K,j}(k+G,X)]d k dX=\fint_{J\Omega}\int_{\R^2}\Tr_{\C^4}[f_{d,K,j}(k,X)]dkdX.
    \end{align*}
Inserting this into \eqref{eq:7.12}, we obtain \eqref{eq:DoS-BM}. In particular, we can take
\begin{align*}
    f_{d,K,0}(k,X)=f(\mathfrak{h}_{d,K,0}^{\rm eff}(k,X)),
\end{align*}
\begin{align*}
    f_{d,K,1}(k,X):&=\frac{i}{2\pi}\int_{\C} \overline{\partial}\widetilde{f}(\zeta )  \{(\zeta- \mathfrak{h}_{d,K,0}^{\rm eff})^{-1},(\zeta- \mathfrak{h}_{d,K,0}^{\rm eff})\}(k,X) (\zeta-\mathfrak{h}_{d,K,0}^{\rm eff}(k,X))^{-1} \, dL(\zeta)\\
    &\quad-\frac{1}{\pi}\int_{\C} \overline{\partial}\widetilde{f}(\zeta )  (\zeta- \mathfrak{h}_{d,K,0}^{\rm eff}(k,X))^{-1}T_{0,1}(k-K)(\zeta- \mathfrak{h}_{d,K,0}^{\rm eff}(k,X))^{-1} \, dL(\zeta)
\end{align*}
and
\begin{align*}
  \MoveEqLeft  f^{\rm eff}_{d,K,2}(k,X)=-\frac{1}{\pi}\int_{\C} \overline{\partial}\widetilde{f}(\zeta )  \Big[-\frac{1}{4}(\zeta-h^{\rm eff}_{d,K,0})^{-1}\{(\zeta-h_{d,K,0}^{\rm eff}),(\zeta-h^{\rm eff}_{d,K,0})^{-1}\}^2\\
     &+\frac{1}{4}\{(\zeta-h^{\rm eff}_{d,K,0})^{-1},\{(\zeta-h^{\rm eff}_{d,K,0}),(\zeta-h^{\rm eff}_{d,K,0})^{-1}\}\}\\
     &+\frac{1}{8}(\zeta-h^{\rm eff}_{d,K,0})^{-1}\{(\zeta-h^{\rm eff}_{d,K,0}),(\zeta-h^{\rm eff}_{d,K,0})^{-1}\}_2\\
     &+(\zeta-h^{\rm eff}_{d,K,0}(k,X))^{-1}T_{0,2}(-i\nabla_x+k-K)(\zeta-h^{\rm eff}_{d,K,0}(k,X))^{-1}\\
     & +(\zeta-h^{\rm eff}_{d,K,0}(k,X))^{-1}\Big[T_{0,1}(-i\nabla_x+k-K)(\zeta-h^{\rm eff}_{d,K,0}(k,X))^{-1}\Big]^2\\
     &+ \frac{i}{2}\{(\zeta-h_{d,K,0}^{\rm eff})^{-1}, T_{0,1}(-i\nabla_x+k-K)\}(\zeta-h_{d,K,0}^{\rm eff})^{-1}\\
     &- \frac{i}{2}\{(\zeta-h_{d,K,0}^{\rm eff})^{-1}T_{0,1}(-i\nabla_x+k-K)(\zeta-h_{d,K,0}^{\rm eff})^{-1},\zeta-h_{d,K,0}^{\rm eff}\}(\zeta-h_{d,K,0}^{\rm eff})^{-1}\\
     &-\frac{i}{2}\{(\zeta-h_{d,K,0}^{\rm eff})^{-1},\zeta-h_{d,K,0}^{\rm eff}\}(\zeta-h_{d,K,0}^{\rm eff})^{-1}T_{0,1}(-i\nabla_x+k-K)(\zeta-h_{d,K,0}^{\rm eff})^{-1}\Big]dL(\zeta).
\end{align*}
\end{proof}

\section*{Acknowledgements} This project has received funding from the European Research Council (ERC) under the European Union’s Horizon 2020 research and innovation program (grant agreement EMC2 No. 810367) and from the Simons Targeted Grant Award No. 896630. We thank Simon Becker, Clotilde Fermanian Kammerer, David Gontier, and Maciej Zworski for helpful discussions.

\medskip
\begin{refcontext}[sorting=nyt]
\printbibliography[heading=bibintoc, title={Bibliography}]
\end{refcontext}
\end{document}